\documentclass[11pt]{article}

\usepackage{float, amssymb, amsmath, amsthm, graphicx, bbm, multirow, siunitx, placeins, authblk, geometry, fullpage}
\usepackage[round]{natbib}

\usepackage[dvipsnames]{xcolor}
\usepackage{diagbox}
\usepackage{tikz}
\usetikzlibrary{arrows.meta}
\usetikzlibrary{decorations.pathreplacing}
\usepackage[shortlabels]{enumitem}
\setlist{nolistsep}

\usepackage{array}
\newcolumntype{C}[1]{>{\centering\arraybackslash}p{#1}}

\usepackage{algorithm}
\usepackage{algpseudocode}
\algnewcommand\algorithmicinput{\textbf{INPUT:}}
\algnewcommand\INPUT{\item[\algorithmicinput]}
\algnewcommand\algorithmicoutput{\textbf{OUTPUT:}}
\algnewcommand\OUTPUT{\item[\algorithmicoutput]}

\usepackage{xr-hyper}
\usepackage{hyperref}[]
\hypersetup{
    colorlinks=true,
    linkcolor=blue,
    filecolor=magenta,      
    urlcolor=cyan,
      citecolor=blue,
    }

\usepackage{cleveref}

\newtheorem{theorem}{Theorem}
\newtheorem{lemma}[theorem]{Lemma}
\newtheorem{proposition}[theorem]{Proposition}
\newtheorem{corollary}[theorem]{Corollary}

\newtheorem{definition}{Definition}
\newtheorem{remark}{Remark}
\newtheorem{assumption}{Assumption}

\makeatletter
\newcommand{\neutralize}[1]{\expandafter\let\csname c@#1\endcsname\count@}
\makeatother

\allowdisplaybreaks

\DeclareMathOperator*{\argmin}{arg\,min}

\usepackage{mathtools}

\def\E{\mathbb{E}}
\def\P{\mathbb{P}}

\def\R{\mathbb{R}}

\title{Transfer learning for piecewise-constant mean estimation: \\ Optimality, $\ell_1$- and $\ell_0$-penalisation}
\author{Fan Wang and Yi Yu}
\affil{Department of Statistics, University of Warwick}
\begin{document}

\maketitle

\begin{abstract}
We study transfer learning for estimating piecewise-constant signals when source data, which may be relevant but disparate, are available in addition to the target data. We first investigate transfer learning estimators that respectively employ $\ell_1$- and $\ell_0$-penalties for unisource data scenarios and then generalise these estimators to accommodate multisources. To further reduce estimation errors, especially when some sources significantly differ from the target,  we introduce an informative source selection algorithm. We then examine these estimators with multisource selection and establish their minimax optimality. Unlike the common narrative in the transfer learning literature that the performance is enhanced through large source sample sizes, our approaches leverage higher observational frequencies and accommodate diverse frequencies across multiple sources. Our theoretical findings are supported by extensive numerical experiments, with the code available online\footnote{\url{https://github.com/chrisfanwang/transferlearning}}.  

\end{abstract}

\section{Introduction}\label{sec-intro}

Consider an unknown signal vector $f = (f_1, \ldots, f_{n_0})^{\top} \in \R^{n_0}$, observed with additive noise that
\begin{equation}\label{model-target}
   y_i =  f_i +\epsilon_i, \quad i =1, \ldots, n_0,
\end{equation}
where $\{ \epsilon_i\}_{i = 1}^{n_0}$ are mutually independent mean-zero random variables and $f$ possesses a piecewise-constant pattern, i.e.~there exists a set of change points 
\begin{equation}\label{def-S}
 \mathcal{S} = \big\{i \in \{1, \dots, n_0-1 \}\colon f_i \neq f_{i+1}  \big\},
\end{equation}
with cardinality $|\mathcal{S}| = s_0$.  

Widely used in signal processing and machine learning literature,  central to the model \eqref{model-target} is the estimation of piecewise-constant signals and localising change points. For estimation, numerous methods have been proposed and investigated, including $\ell_1$-penalised estimators \citep[e.g.][]{lin2017sharp, ortelli2019prediction, guntuboyina2020adaptive} and $\ell_0$-penalised estimators \citep[e.g.][]{fan2018approximate, shen2022phase}. For change point detection, the methods can be categorised into two types, as summarised in \cite{cho2021data}:~global optimisation methods, for example, $\ell_1$-penalised estimators \citep[e.g.][]{lin2017sharp} and $\ell_0$-penalised estimators \citep[e.g.][]{wang2020univariate}, and local testing methods, for example, wild binary segmentation algorithm \citep{fryzlewicz2014wild} and moving sum procedures \citep{chu1995mosum}.

With the explosion of data collected and stored, we increasingly encounter scenarios where additional data are available.  These data may share similar albeit different patterns from our target data.  It is therefore vital to understand how one can utilise the additional information.  To be specific, we consider additional data from $K \in \mathbb{N}^{*}$ source studies $\{y_i^{(k)} \}_{i=1, k = 1}^{n_k, K}$, with 
\begin{equation}\label{model-aux}
y^{(k)}_i = f_i^{(k)} + \epsilon^{(k)}_i, \quad i \in \{1, \dots, n_k\}, \quad k \in \{1, \dots, K\},
\end{equation}
where $\{\epsilon_i^{(k)}\}_{i = 1, k = 1}^{n_k, K}$ are mutually independent mean-zero random variables, and for $k \in \{1, \dots, K\}$, $f^{(k)} = (f_1^{(k)}, \ldots, f_{n_k}^{(k)})^{\top} \in \mathbb{R}^{n_k}$ are unknown signal vectors. These vectors are different from but related to the target signal $f$ introduced in \eqref{model-target}.  In this paper, we are in particular interested in cases where $f^{(k)}$'s have higher observational frequencies than $f$ and where they are not necessarily piecewise-constant.

As a motivating example, consider studying Hungary's Gross Domestic Product (GDP), which is a key economic indicator and is officially released quarterly. 
In addition, monthly released industrial production (IP) data from Hungary and Hungary's neighbours -  Slovakia and Romania  - are also available.   IP has long served as a reliable indicator of GDP growth trends in many economies. Considering the economic structural similarities among these three nations (all being Eastern European countries with a strong focus on manufacturing and industrial sectors), it is hence worth considering enhancing Hungary's GDP trend estimation by leveraging the higher-frequency IP datasets.  More analysis can be found in \Cref{sec-real-data}.

The growing demand to utilise different sources to improve estimation fosters the research in transfer learning in machine learning \citep[e.g.][]{torrey2010transfer}.  Specific areas of applications include natural language processing \citep{daume2009frustratingly}, computer vision \citep{pan2009survey}  and health informatics \citep{tajbakhsh2016convolutional}.  Owing to its successes in applications, transfer learning has attracted much recent attention in statistics and has been studied in various problems. \cite{cai2019transfer} and \cite{reeve2021adaptive} considered nonparametric classification, \cite{cai2022transfer}  explored nonparametric regression, \cite{bastani2021predicting} studied high-dimensional linear regression models and \cite{cai2023transfer} investigated functional data analysis.

In the aforementioned studies, the improvement achieved through transfer learning relies on all sources being beneficial for transfer. In some applications, however,  identifying truly informative sources may not be straightforward. Blindly transferring from arbitrary sources might even worsen the performance compared to only using the target data. In such complex scenarios,  high-dimensional linear regression models have been investigated by \cite{li2022transfer} and generalised linear models by \cite{tian2022transfer} and \cite{li2023estimation}.

In this paper, we study transfer learning for piecewise-constant mean estimation.  We focus on situations where, in addition to the target data, one or more source datasets are available, with some sources that may be substantially different from the target. 
To illustrate the gains and losses of transfer learning, we present a simulation study to numerically compare estimators. These include estimators only using the target data ($\ell_1$ and $\ell_0$), transfer learning estimators using a single informative source ($\ell_1$-T-$1$ and $\ell_0$-T-$1$), multiple informative sources ($\ell_1$-T-$\mathcal{A}$ and $\ell_0$-T-$\mathcal{A}$), estimated informative multisources ($\ell_1$-T-$\widehat{\mathcal{A}}$ and $\ell_0$-T-$\widehat{\mathcal{A}}$) and all sources ($\ell_1$-T-$[K]$ and $\ell_0$-T-$[K]$). We show the results in \Cref{fig:all}; see \Cref{sec-simulation} for details on the simulation settings.
We can clearly see that transfer learning estimators using informative unisource enhance the estimation performance, with further improvement observed when using estimated informative multisources. The
best estimation performance is achieved by utilising predefined informative multisources.  In contrast, transfer learning estimators using all sources perform worse than those only using the target data.

\begin{figure}[t]
        \centering
        \includegraphics[width=0.85 \textwidth]{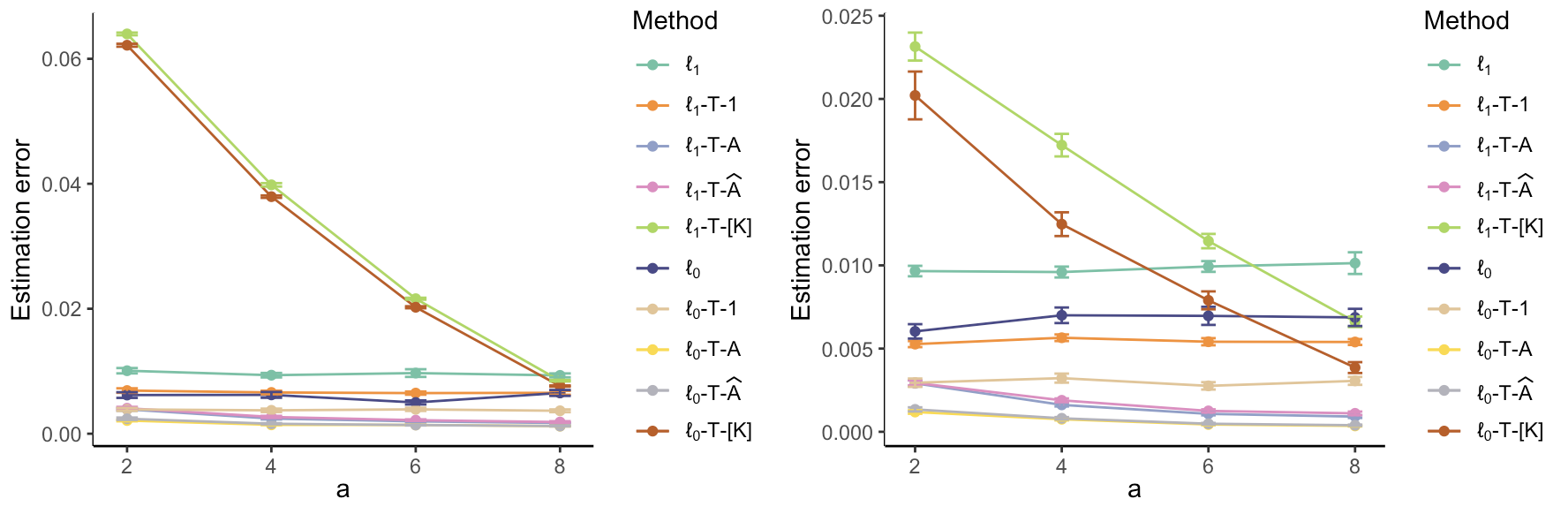}
        \caption{Estimation results with details in Scenario 1 in \Cref{sec-simulation}, the number of source datasets  $K = 10$ and the cardinality of the informative set $a \in \{2, 4, 6, 8\}$. Left panel: Configuration 1.    Right panel: Configuration 2.}
        \label{fig:all}
\end{figure}

\subsection{List of contributions}
The main contributions of this paper are summarised as follows.

Firstly, to the best of our knowledge, this is the first study focusing on the transfer learning framework in the context of estimating piecewise-constant signals.  We provide a comprehensive analysis, characterised by: (i) exploring both uni- and multisource scenarios in Sections~\ref{sec-one-source} and~\ref{sec-multiple-source}; (ii) introducing and evaluating both $\ell_1$- and $\ell_0$-penalised estimators in Sections~\ref{sec-l_1} and \ref{sec-l_0}; (iii) addressing cases that employ all multisources for transfer in \Cref{sec-ora-estimators}, as well as a more complex setting, where some sources may significantly deviate from the target and beneficial sources remain unidentified in \Cref{sec-det-estimators}; and (iv) discussing a range of extensions including affine transformation in quantifying the difference between source and target, as well as using target data in different ways in \Cref{sec-extensions}.

Secondly, our work addresses challenges in the following aspects.
\begin{itemize}
    \item Different from the majority of the existing transfer learning literature \citep[e.g.][]{li2022transfer, tian2022transfer, li2023estimation}, where the improvement of performance is achieved through larger sample sizes of sources, we focus on leveraging higher and potentially different observational frequencies of sources.  Through our approach, we elucidate the direct relationship between the final estimation guarantees and the varying observational frequencies. \cite{cai2023transfer} also examined the higher-frequency framework in the context of functional data analysis.  Allowing for different observational frequencies across multisources, we discover some interesting phenomena that were overlooked in \cite{cai2023transfer}, as discussed in Sections~\ref{sec-optional-selction} and~\ref{sec-num-sup}.
    \item Although the target signal is assumed to be piecewise-constant, the source signals are allowed to possess arbitrary patterns.
\end{itemize}

Thirdly, we introduce transferred estimators respectively utilising $\ell_1$- and $\ell_0$-penalties, and outline their associated theoretical estimation error bounds. We illustrate that the transferred $\ell_0$-penalised estimators are theoretically superior to their $\ell_1$-penalised counterparts, but sacrifice some computational efficiency. This observation resonates with the conventional findings on $\ell_1$- and $\ell_0$-penalisation \citep[see e.g.][]{fan2018approximate}.

Fourthly, we introduce an algorithm to identify informative sources.  We validate its estimation consistency in a non-asymptotic framework, detailed in \Cref{sec-det-estimators}. Following this, we propose data-driven estimators, which are shown to be minimax optimal, see \Cref{sec-minimax}. 

Lastly, we explore extensions by allowing for general affine transformations between the source and target in \Cref{sec-aff} and investigate the effects of utilising target data in unisource scenarios in \Cref{sec:target-uni}.
The theoretical findings in this paper are validated through extensive numerical experiments in \Cref{sec-num}.

\subsection{Notation and organisation}

For any $a \in \mathbb{N}$, denote $[a] = \{1, \ldots, a\}$ with $[0] = \emptyset$, and $[0:a] = \{0\} \cup [a]$. For any set $\mathcal{M}$, let $\vert \mathcal{M}\vert$ denote its cardinality. 

Given $n \in \mathbb{N}^{*}$, for any vector $v \in \R^n$, $\| v\|_2$, $\| v\|_1$ and $\| v\|_0$ represent its $\ell_2$-, $\ell_1$- and $\ell_0$-norms, respectively. We also define $\| v\|_{1/n} = n^{-1/2}\| v \|_2$.  For any set $\emptyset \neq \mathcal{M} \subseteq  [n]$, let a sub-vector of $v$ be $v_\mathcal{M} = (v_i, i \in \mathcal{M})^{\top} \in \mathbb{R}^{\vert \mathcal{M} \vert}$; and $v_{\mathcal{M}} = 0$, if $\mathcal{M} = \emptyset$.  Let the subscript $-\mathcal{M} = [n] \setminus \mathcal{M}$.  

For any matrix $Q\in \R^{n \times m}$, let $Q_{i, j}$ denote the $(i, j)$ entry of $Q$, $(i, j) \in [n] \times [m]$. Let $\|Q\|$ and $\| Q\|_{\mathrm{F}}$ be the spectral and Frobenius norms of $Q$, respectively.  For any set $\emptyset \neq \mathcal{M} \subseteq  [n]$, let $Q_{\mathcal{M}} \in \R^{\vert \mathcal{M} \vert \times m}$ be a submatrix of $Q$ only  containing rows indexed by $\mathcal{M}$.  Let $D \in \mathbb{R}^{(n_0-1) \times n_0}$ be the difference operator, defined as 
\begin{equation}\label{def-D}
        D_{i, j} = \mathbbm{1}\{i = j\} - \mathbbm{1}\{j-i = 1\}, \quad  (i, j) \in [n_0-1] \times [n_0].
\end{equation}

For any $n, m \in \mathbb{N}^{*}$,   let the alignment operators $P^{n, m} \in \R^{n \times m}$ and $\widetilde{P}^{n, m} \in \R^{n \times m}$ be defined as
 \begin{equation}\label{def-P}
   (P^{n, m})_{i, j}= 
      \mathbbm{1}_{ \{ \lceil(j-1)  n/ m \rceil +1 \leq i  \leq \lceil j n / m \rceil \}  },  \quad (i, j) \in [n] \times [m].
 \end{equation}
and with $m \geq n$,
 \begin{equation}\label{def-P-alt}      
      (\widetilde{P}^{n, m})_{i, j}  = \frac{\mathbbm{1}_{ \{ \lceil(i-1)m/n  \rceil  + 1 \leq j \leq \lceil im/n \rceil  \} }}{\lceil im/n \rceil - \lceil(i-1)m/n\rceil},
   \quad (i, j) \in [n] \times [m],
 \end{equation}
respectively.

For any $\sigma>0$, a mean-zero random variable $Z$ is said to be $\sigma$-sub-Gaussian distributed if its Olicz-$\psi_2$-norm $\| Z\|_{\psi_2} = \inf \{t>0:~\E\{\exp(Z^2 /t^2)\} \leq 2\} \leq \sigma$.

In the sequel, we refer to a single source as unisource and multiple sources as multisource.  The unisource and multisource cases are studied in Sections~\ref{sec-one-source} and \ref{sec-multiple-source}, respectively.  Extensions and numerical results are collected in Sections~\ref{sec-extensions} and \ref{sec-num}, with conclusions in \Cref{sec-conclusion}.

\section{Transfer learning with unisource data}\label{sec-one-source}

In this section, we investigate the unisource scenario, where a single source dataset $\{ y_i^{(1)}\}_{i=1}^{n_1}$, is available in addition to the target dataset $\{y_i \}_{i=1}^{n_0}$. 
We measure the discrepancy between the source and target by $n_1^{-1}\| \delta \|_2$, where 
\begin{equation}\label{delta-unisource}
    \delta = f^{(1)} -  P^{n_1, n_0} f  \in \mathbb{R}^{n_1} 
\end{equation}
with $P^{n_1, n_0} \in \R^{n_1 \times n_0}$ defined in \eqref{def-P}.
Consider, for instance, the example in \Cref{sec-intro} on Hungary's quarterly GDP data, with additional access to Hungary's monthly IP data.    
The operator  $P^{n_1, n_0}$ aligns the quarterly GDP data to the monthly IP data to compare datasets despite different observational frequencies.  Since both datasets are in the form of the percentage change compared to the same period in the previous year, the discrepancy vector $\delta$ represents the difference between the percentage change of the aligned quarterly GDP data and the monthly IP data, and the dimension-normalised $\ell_2$-norm provides a measure of this difference. 

Recalling the key message of this paper is to leverage the higher observational frequency of source data, in this section, we focus on the case that $n_1 \geq n_0$.  The complement case $n_1 < n_0$ is considered in \Cref{sec:target-uni} for completeness.

To recover piecewise-constant signals, the $\ell_1$- and $\ell_0$-penalisation are, arguably, the most popular methods.   In the transfer learning context, we study $\ell_1$- and $\ell_0$-penalised estimators in Sections~\ref{sec-l_1} and \ref{sec-l_0}, respectively.  Discussions on the trade-off between the potential theoretical advantages of $\ell_0$-penalised estimators and the computational efficiency of $\ell_1$-penalised estimators can be found in \Cref{sec-1-dis}.

\subsection{\texorpdfstring{Transferred $\ell_1$-penalised estimators}{Transferred l1-penalised estimators}}\label{sec-l_1}

The $\ell_1$-penalisation method aims to encourage model sparsity, by imposing an $\ell_1$-penalty $\|D \cdot \|_1$, with the difference operator $D$ defined in \eqref{def-D}.  In signal processing and statistics literature, such methods have been heavily exploited and referred to as total variation denoising \citep[e.g.][]{rudin1992nonlinear} or fused lasso \citep[e.g.][]{tibshirani2005sparsity}.  
A significant advantage of $\ell_1$-penalisation methods is their convexity, which allows for the exact minimisation within a linear time frame \citep[e.g.][]{johnson2013dynamic}.  

With the target data $\{y_i\}_{i=1}^{n_0}$ in \eqref{model-target} and unisource data $\{y_i^{(1)}\}_{i=1}^{n_1}$ in \eqref{model-aux}, we consider a transfer learning estimator with an $\ell_1$-penalty, namely the unisource-transferred $\ell_1$-penalised estimator, i.e.
\begin{align}\label{estimator_l_1-1}
  \widehat{f} = \widehat{f}(\lambda) = \argmin_{ \theta \in \R^{n_0}} \bigg\{  \frac{1}{2n_0} \Big\|\widetilde{P}^{n_0, n_1} y^{(1)} - \theta \Big\|_2^2   +\lambda \|D   \theta \|_1  \bigg\}, 
\end{align}
where $\widetilde{P}^{n_0, n_1} \in \R^{n_0, n_1}$ is defined in \eqref{def-P-alt}, $D \in \R^{(n_0 - 1) \times n_0}$ in \eqref{def-D} and $\lambda > 0$ is a tuning parameter.

To establish the estimation error bound of $\widehat{f}$, we first introduce the minimal length condition.

\begin{assumption}[Minimal length]\label{ass-fusedlasso}
For the target model defined in \eqref{model-target},  let $\mathcal{S}$ be the set of change points defined in \eqref{def-S}, with $|\mathcal{S}| = s_0 \in \mathbb{N}$. Denote $\mathcal{S} = \{ t_1, \dots, t_{s_0}\}$, for positive $s_0$ and $\mathcal{S} = \emptyset$ otherwise.  With $t_0 = 0$ and $t_{s_0+1} = n_0$, for $i \in [s_{0}+1]$, define $n_i^{(0)} = t_i - t_{i-1}$, $n_{\max}^{(0)} = \max_{i  \in [s_{0}+1]} n_i$ and $n_{\min}^{(0)} = \min_{i  \in [s_{0}+1]} n_i$.  Assume that there exist absolute constants $c_{\max} \geq c_{\min} > 0$, such that   $c_{\min} n^{(0)}_{\min} \leq  n^{(0)}_{\max} \leq c_{\max} n^{(0)}_{\min}$.
\end{assumption}

This condition is also adopted in  \cite{ortelli2019prediction}, \cite{vandegeer2020logistic} and \cite{guntuboyina2020adaptive}. See \Cref{remark_l_1_one} for more discussions.

\begin{theorem}\label{theorem_l_1-1}
Let the target data $\{y_i\}_{i= 1}^{n_0}$ be from \eqref{model-target} satisfying \Cref{ass-fusedlasso}, and  unisource data $\{y_i^{(1)}\}_{i=1}^{n_1}$ be from \eqref{model-aux} with $n_1 \geq n_0$.  Assume that $\{ \epsilon_i^{(1)}\}_{i=1}^{n_1}$ are mutually independent mean-zero $C_{\sigma}$-sub-Gaussian distributed with an absolute constant $C_{\sigma} >0$. Let $\widehat{f}$ denote the estimator defined in \eqref{estimator_l_1-1}, with tuning parameter 
\begin{align}\label{tuning-parameter-1}
   \lambda =  C_\lambda \{(s_0+1)n_1\}^{-1/2},   
\end{align}
where $C_{\lambda} > 0$ is an absolute constant.  It holds with probability at least $1 - n_0^{-c}$ that 
\begin{align}\label{upper_bound_l_1}
    \big\|\widehat{f} - f   \big\|_{1/n_0}^2  \leq   C \frac{(s_{0}+1) \big\{1 + \log \big( n_{0}/ (s_0+1)\big) \big\} +  \| \delta\|_2^2 }{ n_1 }, 
\end{align}
where $\delta \in \R^{n_1}$ is defined in  \eqref{delta-unisource}, and $C, c>0$ are absolute constants. 
\end{theorem}

The proof of \Cref{theorem_l_1-1} is deferred to \Cref{app_1}.

\begin{remark}\label{remark_l_1_one}
\Cref{ass-fusedlasso} requires that all change points in the target model are equally spaced.  To understand the role of \Cref{ass-fusedlasso}, we elaborate on the estimation error bound in its absence.
Based on our analysis of the proof of \Cref{theorem_l_1-1}, with $n_{\max}^{(0)}$ defined in \Cref{ass-fusedlasso}, if  
\begin{align}\label{remark_1-lambda}
    \lambda =  C_{\lambda} \{n_{\max}^{(0)} /( n_1 n_0)\}^{1/2}, 
\end{align}
it holds with probability at least $1 - n_0^{-c}$ that   
\begin{equation}\label{upper_bound_l_1_1}
    \big\|\widehat{f} - f   \big\|_{1/n_0}^2  
     \leq    C \frac{   n_{\max}^{(0)} /\widetilde{n}_{\min}^{(0)}    (s_{0} +1) \big\{ 1+ \log \big(n_{\max}^{(0)} \big)\big\} +  \|\delta\|_2^2 }{n_1},  
\end{equation}
with
\[
   \widetilde{n}_{\min}^{(0)}  = \begin{cases}
       n_0, & s_0 = 0, \\
       \min_{i \in   \{ i \in \{2, \dots, s_{0} \}\colon \mathrm{sign} ( ( D f )_{t_i}) \neq \mathrm{sign} ( ( D f )_{t_{i-1}})  \} \cup \{1,  s_{0}+1 \}} n_i, & \mbox{otherwise}.   \end{cases}
\]
The quantity $\widetilde{n}_{\min}^{(0)}$ represents the minimal distance between change points of the target signals where the change direction alternates, i.e.~transitions from an uptrend to a downtrend or vice versa.  Comparing \eqref{upper_bound_l_1_1} and \eqref{upper_bound_l_1}, it is evident that, at least in our current proofs, without \Cref{ass-fusedlasso}, the upper bound in \eqref{upper_bound_l_1} may suffer a deterioration of order up to $n_0$ in the worst cases.

\cite{ortelli2019prediction} provided a possible relaxation that instead assuming $\widetilde{n}_{\min}^{(0)} \asymp n_{\max}^{(0)}$.  For monotone $f$, this alternative requires $\min \{n_1^{(0)}, n_{s_0+1}^{(0)}\} \asymp n^{(0)}_{\max}$, which is weaker than \Cref{ass-fusedlasso}; otherwise, the two assumptions coincide. Under this relaxed condition and following the proof of \Cref{theorem_l_1-1}, we have that for $\lambda$ in \eqref{remark_1-lambda}, it holds with probability at least $1 - n_0^{-c}$ that  
\begin{equation}\label{upper_bound_l_1_2}
    \big\|\widehat{f} - f   \big\|_{1/n_0}^2  
     \leq    C \frac{    (s_{0} +1) \big\{ 1+ \log \big(n_{\max}^{(0)} \big)\big\} +  \|\delta\|_2^2 }{n_1}. 
\end{equation}
Comparing \eqref{upper_bound_l_1_2} with \eqref{upper_bound_l_1}, we see that this relaxed condition only results in a logarithmic factor deterioration. 
\end{remark}

When working solely with the target data, \cite{guntuboyina2020adaptive} showed that estimation error in terms of $\|\cdot\|_{1/n_0}^2$-loss for $\ell_1$-penalised estimators is of order $(s_0+1)\log \big(n_0/(s_0 +1) \big)/n_0$ under the minimal length condition. This rate is minimax optimal, suggested by a matching lower bound \citep{padilla2017dfs}. 
Under the condition $n_1 \geq n_0$,  if $n_1^{-1}\| \delta\|_2^2 \leq (s_0+1)  \log\big(n_0/(s_0+1)\big) /n_0 $, for the unisource-transferred $\ell_1$-penalised estimator, \Cref{theorem_l_1-1} offers a sharper upper bound and is also minimax optimal, as it matches the lower bound shown in \Cref{theorem_minimax} in \Cref{sec-minimax}.  
This suggests that when the discrepancy level between target and source signals is sufficiently small, and the observational frequency of the source data is higher than that of target data, leveraging information from the source data can improve the estimation performance.

The estimation error bound in \eqref{upper_bound_l_1_2} resonates with the typical structure of estimation errors in transfer learning literature \citep[e.g.][]{cai2019transfer, bastani2021predicting, tian2022transfer}, encompassing two elements:~a fluctuation term and a bias term.  The fluctuation term, $ (s_0 + 1) \{1 + \log(n_0/(s_0+1))\}/ n_1$, reflects the improvement in estimation by transferring from the source with a high observational frequency. The bias term, $n_1^{-1}\|\delta\|_2^2$,  acts as a dimension-normalised measure of the distance between target and source signals, serving as the inherent cost for the transfer process.  It is important to emphasise that \Cref{theorem_l_1-1} does not require the source signal vector $f^{(1)}$ or the difference vector $\delta$ to follow piecewise-constant patterns. Furthermore, there are no constraints on~$n_1^{-1}\|\delta\|_2^2$, the normalised squared-$\ell_2$ distance between the source and target signals. 

Different from the majority in transfer learning literature that achieves improvement through large source sample sizes, we emphasise the high observational frequency from the source data.  Similar emphasis is also noted in \cite{cai2023transfer} in the context of functional data analysis. 
 A more in-depth comparison with their framework will be provided in \Cref{sec-ora-estimators}.

Observe that the estimator $\widehat{f}$ is independent of the target data.  In \Cref{sec-multiple-source}, we do not incorporate the target data for multisource scenarios either, unless we need to identify beneficial sources for transfer.
This approach stems from our assumption that the source data have a higher observational frequency compared to the target data. Concurrently, as corroborated in \Cref{sec-minimax}, the estimation error rate is optimal.
As a byproduct, our theoretical framework does not make any assumptions on the errors of the target model. This indicates that employing transfer learning can not only improve estimation but also be robust against heavy-tailed, temporal dependent or heterogeneous target noise random variables. Despite being minimax optimal, we concur with the sentiment of not using target data.  In Section \ref{sec:target-uni} and Appendix \ref{sec:target-multi}, we propose and study methods incorporating target data in unisource and multisource scenarios, respectively.

\subsection{\texorpdfstring{Transferred $\ell_0$-penalised estimators}{Transferred l0-penalised estimators}}\label{sec-l_0}

With the $\ell_0$-sparsity assumptions, $\ell_1$-penalties can be seen as a convex relaxation of $\ell_0$-penalties, which are of the form $\|D \cdot\|_0$, see \eqref{def-D} for $D$.  Despite the increased computational complexity, in the line graphs, $\ell_0$-penalised convex optimisation can still be solved in polynomial time \citep[e.g.][]{friedrich2008complexity}.  Trading off some computational efficiency, for problems on piecewise-constant signals, $\ell_0$-penalisation enjoys its superior theoretical performance \citep[e.g.~for change point localisation, see][]{wang2020univariate}.

We replace the $\ell_1$-penalty in \eqref{estimator_l_1-1} with an $\ell_0$-penalty.  The counterpart of~\eqref{estimator_l_1-1} is 
\begin{equation}\label{estimator_l_0-1}
    \widetilde{f} = \widetilde{f}(\widetilde{\lambda}) = \argmin_{\theta \in \R^{n_0}} \bigg\{  \frac{1}{2n_0} \Big\|\widetilde{P}^{n_0, n_1} y^{(1)} - \theta \Big\|_2^2   + \widetilde{\lambda} \|D   \theta \|_0  \bigg\},  
\end{equation}
where $\widetilde{P}^{n_0, n_1} \in \R^{n_0, n_1}$ is defined in \eqref{def-P-alt}, $\widetilde{\lambda} > 0$ is a tuning parameter, and $D \in \R^{(n_0 - 1) \times n_0}$ is defined in \eqref{def-D}.  To investigate the potentially different performances of $\ell_1$- and $\ell_0$-penalisation in the transfer learning framework, we present the following theorem, as a counterpart of \Cref{theorem_l_1-1}.

\begin{theorem}\label{theorem_l_0-1}
Let the target data $\{y_i\}_{i= 1}^{n_0}$ be from \eqref{model-target} and  unisource data $\{y_i^{(1)}\}_{i=1}^{n_1}$ be from \eqref{model-aux} with $n_1 \geq n_0$.  Assume that $\{ \epsilon_i^{(1)}\}_{i=1}^{n_1}$ are mutually independent mean-zero $C_{\sigma}$-sub-Gaussian distributed with an absolute constant $C_{\sigma} >0$.  Let $\widetilde{f}$ be defined in \eqref{estimator_l_0-1}, with tuning parameter 
\begin{align}\label{tuning-parameter-0} 
  \widetilde{\lambda} = C_{\widetilde{\lambda}}\frac{ 1 + \log \big(n_0/(s_0+1) \big)}{n_1 },  
 \end{align}
where $C_{\widetilde{\lambda}}> 0$ is an absolute constant.  It holds with probability at least $1 - n_0^{-c}$ that 
 \begin{align}\label{upper_bound_l_0} 
    \big\|\widetilde{f} - f  \big\|_{1/n_0}^2  
     \leq &    C\frac{  (s_{0}+1) \big\{1+\log \big(n_0/(s_0+1) \big) \big\}+\|\delta\|_2^2 }{n_1},
 \end{align}
 where $\delta \in \R^{n_1}$ is defined in  \eqref{delta-unisource}, and $C, c>0$ are absolute constants. 
\end{theorem}

The proof of \Cref{theorem_l_0-1} can be found in \Cref{app_2}. Comparing \Cref{theorem_l_0-1} with \Cref{theorem_l_1-1}, we can see that the $\ell_0$- and $\ell_1$-penalised estimators have the same orders of estimation error bounds, but the performance of the $\ell_1$-penalised estimator depends on a minimal length condition.  
Both Theorems~\ref{theorem_l_1-1} and \ref{theorem_l_0-1} require that $n_1 \geq n_0$, this condition will be relaxed in \Cref{sec:target-uni}.
\citet{fan2018approximate} showed that when solely using the target data, the $\ell_0$-penalised estimator  achieves a minimax optimal estimation error bound of order $(s_0+1) \log\big(n_0/(s_0+1) \big)/n_0$, identical to the $\ell_1$-penalised estimator under a minimal length condition.

\subsection{\texorpdfstring{Comparison of transferred $\ell_1$- and $\ell_0$-penalised estimators}{Comparison of transferred l1- and l0-penalised estimator}}\label{sec-1-dis}

We provide a comprehensive comparison between $\ell_1$- and $\ell_0$-penalised estimators within a transfer learning framework. We focus on aspects including theoretical performance, tuning parameters and computational complexities.

The theoretical performance of the transferred $\ell_0$-penalised estimator does not rely on the minimal length condition, in contrast to the $\ell_1$-penalised estimator. To be specific, without this condition, the $\ell_1$-penalised estimator experiences a deterioration of order $n_0$ in the worst cases, see \Cref{remark_l_1_one}, based on the current proofs.

The tuning parameters for both transferred $\ell_1$- and $\ell_0$-penalised estimators, defined respectively in \eqref{tuning-parameter-1} and \eqref{tuning-parameter-0}, exhibit dependency on the number of change points, $s_0$, in target signals. This unsatisfactory dependency is, however, commonly seen in the literature on piecewise-constant signals \citep[e.g.][]{ortelli2019prediction, guntuboyina2020adaptive, fan2018approximate}, to prompt theoretical optimality. Dropping the dependence of unknown $s_0$ in the tuning parameters, for the $\ell_1$-penalised estimator, adopting  $\lambda \asymp n_1^{-1/2}$, which solely depends on the number of source observations, results in an error bound of order
\[
    \frac{ (s_{0}+1)^2 \big\{ 1+ \log \big( n_{0}/ (s_0+1)\big) \big\} +  \| \delta\|_2^2 }{ n_1 },   
\]
which is sub-optimal by a factor of $s_0+1$.  In contrast, for the $\ell_0$-penalised estimator, adopting $\widetilde{\lambda} \asymp \log(n_0)/n_1$ yields a sharper estimation error bound - sharper than its $\ell_1$ counterpart - of order
 \begin{align}
       \frac{ (s_{0}+1) \big\{1 + \log ( n_{0}) \big\}+  \| \delta\|_2^2 }{ n_1 }.  \nonumber
 \end{align}
The less dependence of the $\ell_0$-penalised estimator on $s_0$ highlights its theoretical superiority over the $\ell_1$-penalised estimator.

As for the computational cost, both the $\ell_1$- and $\ell_0$-penalised estimators start with the same computation of $\widetilde{P}^{n_0, n_1} y^{(1)}$, with a computational cost of order $O(n_1)$.  Upon obtaining $\widetilde{P}^{n_0, n_1} y^{(1)}$, the computational costs of solving the $\ell_1$- and $\ell_0$-penalisation are of order $O(n_0)$ and $O(n_0^2)$, respectively \citep[e.g.][]{johnson2013dynamic,friedrich2008complexity}.

\section{Transfer learning with multisource data}\label{sec-multiple-source}

Owing to the abundance of data, we often encounter situations where multiple sources are available. In the example of Hungary's GDP data in \Cref{sec-intro}, the IP datasets with higher observational frequencies, offer some potential to enhance our comprehension of Hungary's economic trends.

In this section, we propose transferred $\ell_1$- and $\ell_0$-penalised estimators with multiple sources.  In addition to the target data $\{ y_i\}_{i=1}^{n_0}$ in \eqref{model-target}, we have access to $K \in \mathbb{N}^{*}$ mulitsources, denoted by $\{y_i^{(k)}\}_{i=1, k = 1}^{n_k, K}$.  For $k \in [K]$, we measure the discrepancy between the $k$th source to the target by $n_k^{-1/2} \| \delta^{(k)}\|_2$, where
\begin{equation}\label{def-delta}
    \delta^{(k)} =f^{(k)} -  P^{n_k, n_0} f \in \mathbb{R}^{n_k},
\end{equation}
with the alignment operator $P^{n_k, n_0}$ in \eqref{def-P}. 

In \Cref{sec-ora-estimators} we start with generalising the methods studied in \Cref{sec-one-source} to accommodate multisources without selecting beneficial sources.  To maximise the transfer learning benefit, we introduce an informative source selection algorithm and examine these estimators with multisource selection in \Cref{sec-det-estimators}. The associated minimax lower bounds on the estimation accuracy are established in \Cref{sec-minimax}.  We again only focus on the case when $\min_{k \in [K]} \geq n_0$ in this section, with the complement case discussed in \Cref{sec:target-multi} for completeness.

\subsection{Estimation with multisources}\label{sec-ora-estimators}

In \Cref{sec-one-source}, we proposed $\ell_1$- and $\ell_0$-penalised transfer learning estimators for unisource scenarios. The estimation error bounds therein accommodate any level of discrepancy.  In this subsection, our focus shifts from a single to multiple sources, while maintaining the zero-constraint on the discrepancy level.  The results will guide us to properly choose a set of sources and achieve minimax optimality in the sequel.

We introduce the multisource-transferred $\ell_1$-penalised estimator
\begin{equation}\label{estimator-ora-1}
  \widehat{f}^{[K]}= \widehat{f}^{[K]}(\lambda)=  \argmin_{\theta \in \R^{n_0}} \bigg\{  \frac{1}{2n_{0}} \bigg\|\frac{1}{K} \sum_{k \in [K] } \widetilde{P}^{n_0, n_k} y^{(k)}  - \theta\bigg\|_2^2   +\lambda \|D   \theta\|_1  \bigg\},
\end{equation}
and its $\ell_0$ analogue
\begin{equation}\label{estimator-ora-0}
    \widetilde{f}^{[K]} = \widetilde{f}^{[K]}(\widetilde{\lambda}) = \argmin_{\theta \in \R^{n_0}} \bigg\{  \frac{1}{2n_{0}} \bigg\|\frac{1}{K} \sum_{k \in [K]} \widetilde{P}^{n_0, n_k} y^{(k)}  - \theta\bigg\|_2^2   +\widetilde{\lambda} \|D \theta \|_0  \bigg\}, 
\end{equation}
where for any $k \in [K]$, the alignment operator $\widetilde{P}^{n_0, n_k} \in \R^{n_0, n_k}$ is  defined in \eqref{def-P-alt}, $\lambda, \widetilde{\lambda} > 0$ are tuning parameters and $D \in \R^{(n_0 - 1) \times n_0}$ in \eqref{def-D}.  The computational cost of solving \eqref{estimator-ora-1} and \eqref{estimator-ora-0} are of order $O(\sum_{k \in [K]} n_k + n_0)$ and $O(\sum_{k \in [K]} n_k + n_0^2)$, respectively.

As discussed in \Cref{sec-l_1} and \Cref{remark_l_1_one}, the theoretical performance of the $\ell_1$-penalised estimators relies on the minimal length condition (\Cref{ass-fusedlasso}).  It essentially requires that the change points in the target signal spread out balanced.  The lack of such a condition results in a deterioration in estimation. In contrast, for $\ell_0$-penalised estimators, the minimal length condition can be discarded while maintaining the same estimation accuracy.  These arguments remain valid in the multisource scenario, presented below.

\begin{proposition}\label{prop-ora}
Let the target data $\{y_i\}_{i= 1}^{n_0}$ be from \eqref{model-target} and multisource data $\{y_i^{(k)}\}_{i=1, k = 1}^{n_k, K}$ be from~\eqref{model-aux} with $K \in \mathbb{N}^*$ and $\min_{k \in [K]}n_k \geq n_0$.  Assume that $\{ \epsilon_i^{(k)}\}_{i=1, k = 1}^{n_k, K}$ are mutually independent mean-zero $C_{\sigma}$-sub-Gaussian distributed with an absolute constant $C_{\sigma} >0$.  

Let $\widehat{f}^{[K]}$ and $\widetilde{f}^{[K]}$  denote the estimators defined in \eqref{estimator-ora-1} and \eqref{estimator-ora-0}, with tuning parameters 
\begin{equation}\label{tuning-ora}
    \lambda =  C_{\lambda} K^{-1} \sqrt{\frac{\sum_{k = 1}^K n_k^{-1}}{s_0 + 1}} \quad \mbox{and} \quad \widetilde{\lambda} =  C_{\widetilde{\lambda}}\frac{  1 + \log\big(n_{0} / (s_{0}+1) \big)}{K^2 \big(\sum_{k =1}^K n_k^{-1}\big)^{-1}}, 
\end{equation}
respectively, and $C_{\lambda}, C_{\widetilde{\lambda}}  > 0$ being absolute constants.  For $\{\delta^{(k)}\}_{k \in [K]}$ defined in \eqref{def-delta} and absolute constants $C, c>0$, it holds with probability at least $1 - n_0^{-c}$ that 
\begin{align}\label{upper_bound_l_0_ora}
    \big\|\widetilde{f}^{[K]} - f\big\|_{1/n_0}^2 \leq C \Bigg\{\frac{(s_0+1)\big\{ 1 + \log\big(n_0/ (s_0+1)\big)\big\}}{K^2 \big(\sum_{k =1}^K n_k^{-1}\big)^{-1}} + \frac{1}{K} \sum_{k = 1}^K \frac{\|\delta^{(k)}\|_2^2}{n_k}\Bigg\};
\end{align}
if additionally \Cref{ass-fusedlasso} holds, then it holds with probability at least $1 - n_0^{-c}$ that 
\begin{align}\label{upper_bound_l_1_ora}
    \big\|\widehat{f}^{[K]} - f  \big\|_{1/n_0}^2 \leq C \Bigg\{\frac{(s_{0}+1)  \big\{1+\log \big(n_{0}/(s_0+1)  \big) \big\}}{K^2 \big(\sum_{k =1}^{K}n_k^{-1}\big)^{-1}}  + \frac{1}{K}\sum_{k =1}^{K}\frac{\|\delta^{(k)}\|_2^2}{n_k} \Bigg\}.
\end{align}
\end{proposition}

\Cref{prop-ora} presents the estimation error upper bounds for the estimators $\widehat{f}^{[K]}$ and $\widetilde{f}^{[K]}$, demonstrating that both estimators possess the same estimation error bounds. The proof is provided in \Cref{app_3_proposition1}. When $K = 1$, i.e.~the unisource scenario, the upper bounds provided by \eqref{upper_bound_l_0_ora} and \eqref{upper_bound_l_1_ora} degenerate to those of \eqref{upper_bound_l_0} and~\eqref{upper_bound_l_1}, respectively.  We can interpret the two terms in \eqref{upper_bound_l_0_ora} or \eqref{upper_bound_l_1_ora} as the fluctuation and bias.

Considering the fluctuation term 
\[
    \frac{(s_{0}+1)  \big\{ 1+\log \big(n_{0}/(s_0+1)  \big) \big\}}{K^2 \big(\sum_{k =1}^{K}n_k^{-1}\big)^{-1}},
\]
its denominator can be expressed as
\[
    K \frac{K}{\sum_{k = 1}^K n_k^{-1}} = \# \mbox{sources} \, \times \, \mbox{harmonic mean of } \#\mbox{observations}.
\]
This deviates from the general wisdom in transfer learning literature \citep[e.g.][]{li2022transfer, tian2022transfer}, where the arithmetic mean is typically employed in place of the harmonic mean.   This reflects the fundamental difference in having different observational frequencies across sources. The harmonic mean is oftentimes favoured when rates and ratios are involved, for instance in physics \citep[e.g.][]{ferger1931nature}.  Even for this term alone, intriguingly,  due to its non-monotonicity nature, simply having more sources (even when the bias is zero) does not necessarily translate to enhanced estimation precision.  This is different from the majority if not all of the transfer learning studies which equate increased sources to merely having more independent samples. \cite{cai2023transfer} focused on the high-frequency framework in functional data analysis but assumed the same source observational frequencies, thereby concealing such phenomena. See \Cref{sec-num-sup} for the corresponding numerical illustration.

As for the bias term, instead of upper bounding it using the maximum discrepancy level among all sources \citep[e.g.][]{bastani2021predicting, tian2022transfer, li2022transfer, li2023estimation}, we upper bound it by the arithmetic mean of $n_k^{-1}\|\delta^{(k)}\|_2^2$.  This characterises the estimation error of $\widehat{f}^{[K]}$ or $\widetilde{f}^{[K]}$ without constraining the discrepancy between the source and target datasets.  The arithmetic mean roots in the design of the transferred estimator \eqref{estimator-ora-1}, where the optimisation is taken over the squared $\ell_2$-norm of a residual obtained from an arithmetic mean.  Similar to the fluctuation term, the bias term is not a monotone function, i.e.~without further assumptions, adding a source dataset does not necessarily increase or decrease the estimation error.

In response to the sentiment of not directly using the target data, we propose an alternative estimator and establish its estimation error bound in \Cref{sec:target-multi}.  Its estimation error upper bound also consists of two terms: the fluctuation term involves the arithmetic mean of the number of observations, and the bias term is the weighted mean of the discrepancy level.  See \Cref{sec:target-multi} for more discussions and comparisons.

\subsection{Estimation with multisource selection} \label{sec-det-estimators}

It is shown in \Cref{sec-ora-estimators} that both $\ell_1$- and $\ell_0$-penalised estimators achieve an estimation error upper bound of the order
\begin{equation} \label{eq-upper-bound}
    \frac{(s_{0}+1)  \big\{1+\log \big(n_{0}/(s_0+1)  \big) \big\}}{K^2 \big(\sum_{k =1}^{K}n_k^{-1}\big)^{-1}}  + \frac{1}{K}\sum_{k =1}^{K}\frac{\|\delta^{(k)}\|_2^2}{n_k}.    
\end{equation}
To discuss the optimality, with $K \geq 1$ source datasets in hand, one would like to seek an estimator of~$\mathcal{A}^*$, where
\[
    \mathcal{A}^* \in \argmin_{\mathcal{A} \subset [K]} \bigg\{\frac{(s_{0}+1)  \big\{ 1+\log \big(n_{0}/(s_0+1)  \big) \big\}}{|\mathcal{A}|^2 \big(\sum_{k \in \mathcal{A}} n_k^{-1}\big)^{-1}}  \vee \frac{1}{|\mathcal{A}|}\sum_{k \in \mathcal{A}} \frac{\|\delta^{(k)}\|_2^2}{n_k}\bigg\}.
\]
A consistent estimation of $\mathcal{A}^*$ relies on a consistent estimation of $\|\delta^{(k)}\|_2$.  

Note that, in our framework \eqref{model-target} and \eqref{model-aux}, the target signals possess piecewise-constant patterns, while the source signals do not necessarily. Our knowledge of the difference vectors, $\delta^{(k)}$, is thus limited to their dimensionality, which is $n_k$. The estimation error associated with  $\|\delta^{(k)}\|_2^2$ has an order of $n_k \log(n_k)$, and therefore dominates the fluctuation term in \eqref{eq-upper-bound}.  This prohibits a consistent estimation of $\mathcal{A}^*$.  As a resort, we present a direct consequence of \Cref{prop-ora}.

\begin{corollary}\label{cor-l0l1}
Let the target data $\{y_i\}_{i= 1}^{n_0}$ be from \eqref{model-target} and multisource datasets $\{y_i^{(k)}\}_{i=1, k = 1}^{n_k, K}$ be from~\eqref{model-aux} with $K \in \mathbb{N}^*$ and $\min_{k \in [K]}n_k \geq  n_0$.  Assume that $\{ \epsilon_i^{(k)}\}_{i=1, k = 1}^{n_k, K}$ are mutually independent mean-zero $C_{\sigma}$-sub-Gaussian distributed with an absolute constant $C_{\sigma} >0$. 

For any $h > 0$, let 
\begin{align}\label{def-mathcal-A-h}
   \mathcal{A}_h = \big\{ k \in [K]: \, n_k^{-1/2} \| \delta^{(k)}\|_2 \leq h   \big\}.
\end{align} 
If $\mathcal{A}_h \neq \emptyset$, then let $\widehat{f}^{\mathcal{A}_h}$ and $\widetilde{f}^{\mathcal{A}_h}$ denote the estimators defined in \eqref{estimator-ora-1} and \eqref{estimator-ora-0}, with tuning parameters 
\[
    \lambda =  C_{\lambda} |\mathcal{A}_h|^{-1} \sqrt{\frac{\sum_{k \in \mathcal{A}_h} n_k^{-1}}{s_0 + 1}} \quad \mbox{and} \quad \widetilde{\lambda} =  C_{\widetilde{\lambda}}\frac{  1 + \log\big(n_{0} / (s_{0}+1) \big)}{|\mathcal{A}_h|^2 \big(\sum_{k \in \mathcal{A}_h} n_k^{-1}\big)^{-1}}, 
\]
respectively, with $C_{\lambda}, C_{\tilde{\lambda}} > 0$ being absolute constants.  For $\{\delta^{(k)}\}_{k \in [K]}$ defined in \eqref{def-delta} and absolute constants $C, c>0$, it holds with probability at least $1 - n_0^{-c}$ that 
\[
    \big\|\widetilde{f}^{\mathcal{A}_h} - f \big\|_{1/n_0}^2 \leq C \Bigg\{\frac{(s_{0}+1)  \big\{ 1 + \log \big(n_{0}/(s_0+1)  \big) \big\}}{|\mathcal{A}_h|^2 \big(\sum_{k \in \mathcal{A}_h } n_k^{-1}\big)^{-1}}  + h \Bigg\};
\]
if additionally \Cref{ass-fusedlasso} holds, then it holds with probability at least $1 - n_0^{-c}$ that 
\[
    \big\|\widehat{f}^{\mathcal{A}_h} - f \big\|_{1/n_0}^2 \leq C \Bigg\{\frac{(s_{0}+1)  \big\{ 1 + \log \big(n_{0}/(s_0+1)  \big) \big\}}{|\mathcal{A}_h|^2 \big(\sum_{k \in \mathcal{A}_h } n_k^{-1}\big)^{-1}}  + h \Bigg\}.
\]
\end{corollary}

\Cref{cor-l0l1} provides the estimation error bounds for the estimators $\widehat{f}^{\mathcal{A}_h}$ and $\widetilde{f}^{\mathcal{A}_h}$, which utilise sources from the set $\mathcal{A}_h$ instead of all. 
In the existing literature \citep[e.g.][]{bastani2021predicting, li2022transfer, tian2022transfer}, as discussed after \Cref{prop-ora}, the fluctuation term decreases with more sources added in.  A common practice is to choose $h$ as the estimation error obtained by only using the target dataset, then choose the corresponding $\mathcal{A}_h$.  Motivated by \cite{li2022transfer}, we propose the following informative set detection algorithm detailed in \Cref{alg_isd}.

\begin{algorithm}[ht]
\caption{Informative set detection algorithm} 
    \begin{algorithmic}
        \INPUT{Target data $y \in \R^{n_0}$, source data $y^{(k)} \in \R^{n_k}, k \in [K]$, screening width $\widehat{t}^{k} \in [n_k],  k \in [K]$ and thresholds $\{ \tau_k \}_{k=1}^K \subset  \R$}
        \For{$k \in [K]$ }
        \State {$\widehat{\Delta}^{(k)} \leftarrow n_k^{-1/2} y^{(k)} - n_k^{-1/2}  P^{n_k, n_0} y $} \Comment{See \eqref{def-P} for $P^{n_k, n_0}$}
        \State{$\widehat{T}_k \leftarrow \Big\{ i \in [n_k]\colon \big\vert \widehat{\Delta}_{i}^{(k)}\big\vert \mbox{ is among the first } \widehat{t}_{k} \mbox{ largest of } \{|\widehat{\Delta}_j^{(k)}|\}_{j \in [n_k]}\Big\}$}
                 \EndFor
       \State{ $\widehat{\mathcal{A}} \leftarrow \big\{k \in [K]\colon \big\| \big( \widehat{\Delta}^{(k)} \big)_{\widehat{T}_k} \big\|_2^2  \leq \tau_k \big\}$}
        \OUTPUT{$\widehat{\mathcal{A}}$}
    \end{algorithmic}\label{alg_isd}
\end{algorithm}

The core of \Cref{alg_isd} lies in executing sure screening \citep{fan2008sure} on the normalised deviation vectors between the target and source data to reduce the magnitude of noise. With a sequence of predetermined tuning parameters $\{ \tau_{k} \}_{k=1}^K$, a source is identified as informative if the squared $\ell_2$-norm of the corresponding screened version statistic does not exceed its assigned threshold.  The computational cost of \Cref{alg_isd} is of order $O\big(\sum_{k =1}^K n_k\big)$. 

For a certain $h > 0$, a consistent estimate of $\mathcal{A}_h$ relies on some identifiability condition of $\mathcal{A}_h$, e.g.~\Cref{ass_detection}.  Under \Cref{ass_detection}, we show that with properly chosen tuning parameters and high probability, \Cref{alg_isd} outputs $\widehat{\mathcal{A}} = \mathcal{A}_h$.

\begin{assumption}[Identifiability of $\mathcal{A}_{h^*}$]\label{ass_detection}
Assume that there exists 
\begin{equation}\label{h_upper_bound}
    h^* \leq \sqrt{C_{\mathcal{A}} \frac{( s_{0} +1 )  \big\{ 1 + \log\big(n_0/(s_0 +1)\big) \big\}}{n_0}},
\end{equation}  
where $C_{\mathcal{A}} > 0$ is an absolute constant.  Let $\mathcal{A}_{h^*}$ be the corresponding set defined in \eqref{def-mathcal-A-h}.

If $[K]\setminus \mathcal{A}_{h^*} \neq \emptyset$, then for any $k \in [K]\setminus \mathcal{A}_{h^*}$, assume that 
\[
   n_k^{-1} \big\| \delta^{(k)}_{\mathcal{H}_{k}} \big\|_2^2\geq C_{\mathcal{A}^{c}} \frac{  ( s_{0} +1) \big\{ 1+ \log \big( n_0/(s_0+1) \big) \big\} +  \log(n_0 \vee n_k ) }{n_0}, 
\]
where $\mathcal{H}_k = \{j \in [n_k]:\, |\delta^{(k)}_j| > 4 \sqrt{\log(n_0 \vee n_k)}\} \neq \emptyset$, $\delta^{(k)}$ is defined in \eqref{def-delta} and $C_{\mathcal{A}^{c}}>0$ is an absolute constant satisfying $C_{\mathcal{A}^{c}} \geq  4 C_{\mathcal{A}}$.
\end{assumption}

It follows directly from \Cref{cor-l0l1} that $h^*$ in \eqref{h_upper_bound} and its corresponding $\mathcal{A}_{h^*}$ lead to estimation error bounds for both $\ell_1$- and $\ell_0$-penalised estimators of the order
\[
    \frac{(s_{0}+1)  \big\{ 1 + \log \big(n_{0}/(s_0+1)  \big) \big\}}{|\mathcal{A}_{h^{*}}|^2 \big(\sum_{k \in \mathcal{A}_{h^{*}} } n_k^{-1}\big)^{-1}} \vee  \bigg\{(h^*)^2 \wedge \frac{( s_{0} +1 )  \big\{ 1 + \log\big(n_0/(s_0 +1)\big) \big\}}{n_0}\bigg\}.
\]
Together with the assumption that $\min_{k \in [K]} n_k \geq n_0$, the above rate is always sharper than the optimal estimation rate when only using the target dataset \citep[e.g.][]{fan2018approximate}.  

As we discussed before, without additional assumptions, the estimation error of $\delta_k$'s always dominates the optimal estimation rate when only using the target dataset.  \Cref{ass_detection} imposes a separation, in the sense that 
\[
    n_k^{-1}\|\delta^{(k)}\|^2_2
    \begin{cases}
        \lesssim (h^*)^2 \wedge \frac{( s_{0} +1 )  \{ 1 + \log(n_0/(s_0 +1)\}}{n_0}, & k \in \mathcal{A}_{h^*}, \\ 
        \gtrsim \frac{  ( s_{0} +1) \{ 1+ \log ( n_0/(s_0+1) )\} +  \log(n_0 \vee n_k ) }{n_0}, & k \notin \mathcal{A}_{h^*}.
    \end{cases}
\]
We acknowledge that between these two cases, there is a gap, which vanishes provided a mild condition holds that
\[
    ( s_{0} +1) \big\{ 1+ \log ( n_0/(s_0+1) )\big\} \gtrsim \log(n_0 \vee n_k).
\]
\Cref{ass_detection} further assumes that for $k \notin \mathcal{A}_{h^*}$, there exists a sub-vector such that each entry of $\delta^{(k)}_{\mathcal{H}_k}$ is, in magnitude, large enough - larger than a high-probability upper bound on mean-zero sub-Gaussian noise.  This level guarantees that entrywise screening is sufficient to detect such deviance.

\begin{theorem}\label{theorem_detection_consistency}
Let $\widehat{\mathcal{A}}$ be the output of \Cref{alg_isd}, with the following inputs:
\begin{itemize}
    \item the target dataset $\{y_i\}_{i=1}^{n_0}$ satisfying \eqref{model-target},
    \item the source datasets $\{y_i^{(k)}\}_{i=1, k = 1}^{n_k, K}$ from \eqref{model-aux} satisfying  \Cref{ass_detection},
    \item the index sequence $\{\widehat{t}_{k}\}_{k=1}^K$ and the threshold sequence $\{\tau_k\}_{k=1}^K$ satisfying
    \begin{align}\label{hat_t_k-def}
        \widehat{t}_{k} = C_{\widehat{\mathcal{A}}} \frac{ n_k}{8n_0}  \bigg\{\frac{ (s_{0}+1)  \big\{1+\log\big(n_0/(s_0+1) \big)\big\}}{\log(n_0 \vee n_k)} + 1\bigg\}
    \end{align}
    and    
    \begin{align}\label{tau-def}    
        \tau_{k} = C_{\widehat{\mathcal{A}}} \frac{  (s_{0} +1) \big\{1+ \log\big(n_0/(s_0+1) \big) \big\} +\log(n_0 \vee n_k)}{n_0}, 
    \end{align}    
    where $C_{\widehat{\mathcal{A}}} >0$ is an absolute constant satisfying $ 2 C_{\mathcal{A}} \leq C_{\widehat{\mathcal{A}}}  \leq C_{\mathcal{A}^{c}}/2$, with absolute constants $C_{\mathcal{A}^{c}}, C_{\mathcal{A}} >0$ introduced in \Cref{ass_detection}.
\end{itemize}

Assume that $\{\epsilon_i\}_{i=1}^{n_0} \cup \{\epsilon_i^{(k)}\}_{i=1, k = 1}^{n_k, K}$ are mutually independent mean-zero $C_{\sigma}$-sub-Gaussian distributed with an absolute constant $C_{\sigma} >0$.  It holds that 
\[
    \mathbb{P} \{\widehat{\mathcal{A}} = \mathcal{A}_{h^*}\} \geq  1 - K  \{n_0\vee (\min_{k \in [K]}n_k)\}^{-c},
\]
where $c>0$ is an absolute constant and $\mathcal{A}_{h^*}$ is defined in \Cref{ass_detection}.
\end{theorem}

Note that \Cref{theorem_detection_consistency} holds for nonempty or empty $\mathcal{A}_{h^*}$.  \Cref{theorem_detection_consistency} presents a non-asymptotic result with the proof deferred to \Cref{app_3_det}. To the best of our knowledge, all existing theoretical results on the informative source detection algorithm are asymptotic  \citep[e.g.][]{li2022transfer,tian2022transfer}. Despite these exciting results, it is important to recognise that \Cref{theorem_detection_consistency} depends on the selection of the tuning parameters.  The screening sizes sequence $\{\widehat{t}_{k}\}_{k=1}^K$ controls the errors from additive noise.  The choice of thresholds $\{\tau_k\}_{k=1}^K$ in \eqref{tau-def}, serves as an upper bound on the maximum squared $\ell_2$-norm of the corresponding screened version statistic when the sources are informative, and as a lower bound when the sources are not informative, as demonstrated in \Cref{app_3_det}. Theoretical selections for both sequences depend on the number of change points in target signals $s_0$. Practical guidance for selecting these tuning parameters can be found in \Cref{sec-num}.

Combining \Cref{cor-l0l1} and \Cref{theorem_detection_consistency}, we immediately have the following.

\begin{corollary}\label{cor-a-hat-l0l1}
Let the target data $\{y_i\}_{i= 1}^{n_0}$ be from \eqref{model-target} and the source datasets $\{y_i^{(k)}\}_{i=1, k = 1}^{n_k, K}$ be from~\eqref{model-aux}, satisfying \Cref{ass_detection}, with $K \in \mathbb{N}^*$ and $\min_{k \in [K]}n_k \geq n_0$.  Assume that $\{\epsilon_i\}_{i = 1}^{n_0} \cup \{ \epsilon_i^{(k)}\}_{i=1, k = 1}^{n_k, K}$ are mutually independent mean-zero $C_{\sigma}$-sub-Gaussian distributed with an absolute constant $C_{\sigma} >0$. 

Let  $\widehat{\mathcal{A}}$ be the output of  \Cref{alg_isd} with the index sequence $\{ \widehat{t}_{k}\}_{k=1}^{K}$ and the threshold sequence $\{\tau_k\}_{k=1}^K$ chosen as \Cref{theorem_detection_consistency}.  If $\widehat{\mathcal{A}} \neq \emptyset$, then let $\widehat{f}^{\widehat{\mathcal{A}}}$ and $\widetilde{f}^{\widehat{\mathcal{A}}}$ denote the estimators defined in \eqref{estimator-ora-1} and \eqref{estimator-ora-0}, with tuning parameters 
\[
    \lambda =  C_{\lambda} |\widehat{\mathcal{A}}|^{-1} \sqrt{\frac{\sum_{k \in \widehat{\mathcal{A}}} n_k^{-1}}{s_0 + 1}} \quad \mbox{and} \quad \widetilde{\lambda} =  C_{\widetilde{\lambda}}\frac{  1 + \log\big(n_{0} / (s_{0}+1) \big)}{|\widehat{\mathcal{A}}|^2 \big(\sum_{k \in \widehat{\mathcal{A}}} n_k^{-1}\big)^{-1}}, 
\]
respectively, with $C_{\lambda}, C_{\tilde{\lambda}} > 0$ being absolute constants.  With $\{\delta^{(k)}\}_{k \in [K]}$ defined in \eqref{def-delta} and absolute constants $C, c>0$, it holds with probability at least $1 - 2^K n_0^{-c}$ that 
\[
    \big\|\widetilde{f}^{\widehat{\mathcal{A}}} - f\big\|_{1/n_0}^2 \leq C \bigg\{\frac{(s_0+1)  \big\{1 + \log \big(n_{0}/ (s_0+1)) \big) \big\}}{ \vert\mathcal{A}_{h^*} \vert^2 \big(\sum_{k \in \mathcal{A}_{h^*}} n_k^{-1}\big)^{-1}} + (h^*)^2 \wedge \frac{(s_0+1)\big\{ 1 + \log\big(n_0/(s_0+1)\big) \big\}}{n_0} \bigg\};
\]
if additionally \Cref{ass-fusedlasso} holds, then it holds with probability at least $1 - 2^K n_0^{-c}$ that 
\[
    \big\|\widehat{f}^{\widehat{\mathcal{A}}} - f\big\|_{1/n_0}^2 \leq C \bigg\{\frac{(s_0+1)  \big\{1 + \log \big(n_{0}/ (s_0+1)) \big) \big\}}{ \vert\mathcal{A}_{h^*} \vert^2 \big(\sum_{k \in \mathcal{A}_{h^*}} n_k^{-1}\big)^{-1}} + (h^*)^2 \wedge \frac{(s_0+1)\big\{ 1 + \log\big(n_0/(s_0+1)\big) \big\}}{n_0}   \bigg\}.
\]
\end{corollary}

\Cref{cor-a-hat-l0l1} shows that, with the source datasets selected via \Cref{alg_isd}, the $\ell_0$- and $\ell_1$-penalised transferred estimators always improve upon only using the target dataset.  In addition, with $\vert\mathcal{A}_{h^*} \vert^2 \big(\sum_{k \in \mathcal{A}_{h^*}} n_k^{-1}\big)^{-1}$ stay unchanged, the smaller $h^*$ for the identifiable $\mathcal{A}_{h^*}$ defined in \Cref{ass_detection}, the more improvement.

\begin{remark}\label{remark-empty-A}
    \Cref{cor-a-hat-l0l1} derives an estimation error upper bound when the collection of selected informative sets is nonempty.  Recall that the selection is consistent even when $\mathcal{A}_{h^*} = \emptyset$ as shown in \Cref{theorem_detection_consistency}.  We remark that when $\widehat{\mathcal{A}} = \emptyset$, we deploy estimators only using target datasets, i.e.
    \[
        \widetilde{f}^{\mathrm{target}} = \argmin_{ \theta \in \R^{n_0}} \bigg\{  \frac{1}{2n_0} \Big\|y - \theta \Big\|_2^2   +\lambda \|D   \theta \|_0  \bigg\}  \mbox{ and }  \widehat{f}^{\mathrm{target}} = \argmin_{ \theta \in \R^{n_0}} \bigg\{  \frac{1}{2n_0} \Big\|y - \theta \Big\|_2^2   +\lambda \|D   \theta \|_1  \bigg\}.
    \]
    It holds with large probability that
    \[
        \max\big\{\big\|\widetilde{f}^{\mathrm{target}} - f\big\|_{1/n_0}^2, \,  \big\|\widehat{f}^{\mathrm{target}} - f\big\|_{1/n_0}^2\big\} \lesssim \frac{(s_0+1)\big\{ 1 + \log\big(n_0/(s_0+1)\big) \big\}}{n_0},
    \]
    which directly follows from the proofs of Theorems~\ref{theorem_l_1-1} and \ref{theorem_l_0-1}.
\end{remark}

\subsubsection{An optional selection step}\label{sec-optional-selction}

Due to the potentially different observational frequencies across multisources, as we have discussed, even if all $\delta^{(k)}$ equal zero, simply having more source datasets does not necessarily decrease the fluctuation term, unlike existing literature \citep[e.g.][]{bastani2021predicting, tian2022transfer}.  \Cref{alg_isd} provides a consistent estimator $\widehat{\mathcal{A}}$ of $\mathcal{A}_{h^*}$ under the identifiability condition \Cref{ass_detection}.  With~$\widehat{\mathcal{A}}$, \Cref{cor-a-hat-l0l1} provides an estimation error bounds of the proposed transferred estimators.  To minimise the fluctuation term, one may adopt an optional step and choose
\begin{equation}\label{eq-widetilde-a}
    \widetilde{\mathcal{A}} \in \argmin_{\emptyset \neq \mathcal{A} \subset \widehat{\mathcal{A}}} \frac{\sum_{k \in \mathcal{A}} n_k^{-1}}{|\mathcal{A}|^2}.
\end{equation}
The computational cost of \eqref{eq-widetilde-a} is of order $O(2^{|\widehat{\mathcal{A}}|})$.  A direct consequence of \Cref{cor-a-hat-l0l1} is that one can improve the estimation rates to
\[
   \frac{(s_0+1)  \big\{1 + \log \big(n_{0}/ (s_0+1)) \big) \big\}}{  \max_{\emptyset \neq \mathcal{A} \subset \mathcal{A}_{h^*} } \big\{\vert\mathcal{A} \vert^2 (\sum_{k \in \mathcal{A}} n_k^{-1})^{-1} \big\} } + (h^*)^2 \wedge \frac{(s_0+1)\big\{ 1 + \log\big(n_0/(s_0+1)\big) \big\}}{n_0}.
\]
Despite being a potential improvement upon \Cref{cor-a-hat-l0l1}, we would like to point out that, when the frequencies are roughly of the same order, such improvement may not be materialised due to the unspecified constants involved.  We would, therefore, focus on estimators based on $\widehat{\mathcal{A}}$ in the sequel, with a numerical example designed for $\widetilde{\mathcal{A}}$ presented in \Cref{sec-num-sup}.

\subsection{Minimax optimality}\label{sec-minimax} 

In this subsection, we investigate the minimax lower bound on the estimation of target signals within the framework of transfer learning. This analysis underscores the minimax optimality of the $\widehat{\mathcal{A}}$-transferred $\ell_1$- and $\ell_0$-penalised estimators.

\begin{theorem}\label{theorem_minimax}
Let the target data $\{y_i\}_{i= 1}^{n_0}$ be from \eqref{model-target} and the source datasets $\{y_i^{(k)}\}_{i=1, k = 1}^{n_k, K}$ be from~\eqref{model-aux}, with $K \in \mathbb{N}^*$ and $\min_{k \in [K]} n_k \geq  n_0$.  Assume that $\{\epsilon_i\}_{i = 1}^{n_0} \cup \{ \epsilon_i^{(k)}\}_{i=1, k = 1}^{n_k, K}$ are mutually independent mean-zero $C_{\sigma}$-sub-Gaussian distributed with an absolute constant $C_{\sigma} >0$. 

For any $h > 0$, let its associated $\mathcal{A}_h = \{ k_1, \dots, k_a\}$ be defined in \eqref{def-mathcal-A-h} with $|\mathcal{A}_h| = a$.  Define the parameter space as
\[
     \Theta_{s_0, \mathcal{A}_h} = \bigg\{\theta = \big(f^{\top}, (f^{(k_1)})^{\top}, \dots,  (f^{(k_a)})^{\top} \big)^{\top} \colon  \|D f\|_{0} \leq s_0\bigg\},
\]
with $\delta^{(k)}$ defined in  \eqref{def-delta}.
It holds that 
\begin{align}\label{minimax_bound}
    \inf_{\widehat{f} \in \R^{n_0}} \sup_{\theta \in  \Theta_{s_0, \mathcal{A}_h}} \P \bigg\{ \| \widehat{f} - f\|^2_{1/n_0}  \geq  C \bigg(  \frac{s_0 \log(n_0 /s_0)}{\sum_{k \in \mathcal{A}_h} n_k}  + h^2 \wedge \frac{s_0 \log(n_0 /s_0)}{n_0}  \bigg) \bigg\} \geq \frac{1}{2},
\end{align}
with an absolute constant $C>0$.
\end{theorem}

\begin{remark}\label{remark-minimax}
    In \Cref{theorem_minimax}, we assume that $\min_{k \in [K]} n_k \geq  n_0$. In fact, this condition is not necessary for unisource scenarios.  Following an almost identical proof, for cases of both $n_1 \geq n_0$ and $n_1 < n_0$, one can derive a minimax lower bound in unisource scenarios:
    \[
    \frac{s_0 \log(n_0 /s_0)}{n_1 + n_0}  +  \frac{\| \delta \|^2}{n_1} \wedge \frac{s_0 \log(n_0 /s_0)}{n_0},
    \]
    where $\delta$ is defined in \eqref{delta-unisource}.
\end{remark}

The term $ s_0 \log(n_0 /s_0) / \big( \sum_{k \in \mathcal{A}_h} n_k)$ arises from the ideal scenario where $f^{(k)} = P^{n_k, n_0} f$ holds true for any $k \in \mathcal{A}_h$. This scenario leads to the representation of this term as the minimax optimal convergence rate.  The other term, $h^2 \wedge \{ s_0/n_0 \log(n_0 /s_0)\}$, is derived from the minimax optimal convergence rate corresponding to the worst-case scenario, where for any $k \in \mathcal{A}_h$, $f^{(k)} = 0$ and $f$ satisfies $ n_k^{-1}\|\delta^{(k)}\|^2 \leq h^2 \wedge \{s_0/n_0 \log(n_0 /s_0)\}$.  

To further understand \Cref{theorem_minimax}, we compare it with the minimax convergence rate only using the target dataset $s_0/n_0\log(n_0/s_0)$ \citep[e.g.][]{fan2018approximate}, which is larger than the one in \Cref{theorem_minimax}, when $\sum_{k \in \mathcal{A}_h} n_k \geq n_0$.  Comparing \Cref{theorem_minimax} with the minimax rates established in the existing transfer learning literature, our minimax lower bound follows a similar dual-term pattern \citep[][]{tian2022transfer, li2022transfer, cai2022transfer}, involving the minimax optimal estimation rate resulting from multisources, and the minimum between the minimax optimal estimation rate only using the target dataset and the contrasts between the target and source datasets.

We acknowledge that there is a gap between the minimax lower bound and upper bounds achieved by $\widehat{\mathcal{A}}$-transferred $\ell_1$- and $\ell_0$-penalised estimators, as shown in \Cref{cor-a-hat-l0l1}.  To be specific, the upper bound involves 
\[
\big(\vert \mathcal{A}\vert \times \mbox{the harmonic mean of source observations in }\mathcal{A}\big)^{-1},
\]
while the lower bound has
\[
\big(\vert \mathcal{A}\vert \times \mbox{the arithmetic mean of suorce observations in }\mathcal{A}\big)^{-1}.
\]  
Since the harmonic mean is no larger than the arithmetic mean, in general, when the harmonic and arithmetic means are of the same order, our proposed estimators are minimax rate-optimal up to constants.   These two different means are different in rates only if the frequencies are highly unbalanced, where we conjecture that the lower bound should be improved.  Some numerical demonstration of this can be found in \Cref{sec-num-sup}.

\section{Extensions}\label{sec-extensions}

In this section, we discuss two extensions, focusing exclusively on $\ell_0$-penalised estimators. Note that $\ell_1$-penalised estimators yield the same results under the minimal length condition \Cref{ass-fusedlasso}. In \Cref{sec-aff}, we allow for general affine transformations instead of the alignment operator $P^{n_1, n_0}$ defined in \eqref{def-P}, when describing the deviance between the source and target.  More direct usage of the target is studied in \Cref{sec:target-uni} and Appendix \ref{sec:target-multi}, despite the minimax optimality obtained already in Sections~\ref{sec-one-source} and \ref{sec-multiple-source}.

\subsection{Affine transformation}\label{sec-aff}

To allow for more flexibility in leveraging additional information, we consider the unisource scenario with the discrepancy between the source and target measured through 
\begin{equation}\label{def-delta-A}
 \delta^A = f^{(1)} - A f  \in \R^{n_1},
\end{equation}
for any matrix $A \in \R^{n_1 \times n_0}$ with a left inverse, meaning that there exists a matrix $\widetilde{A} \in \R^{n_0 \times n_1}$ such that $\widetilde{A}A = I_{n_0}$. The corresponding estimator is defined as 
\begin{equation}\label{estimator_l_0-affine}
    \widetilde{f}^{\widetilde{A}} = \widetilde{f}^{\widetilde{A}}(\widetilde{\lambda}_{\widetilde{A}}) = \argmin_{\theta \in \R^{n_0}} \bigg\{  \frac{1}{2n_0} \Big\|\widetilde{A} y^{(1)} - \theta \Big\|_2^2   +  \widetilde{\lambda}_{\widetilde{A}}  \|D   \theta \|_0  \bigg\},  
\end{equation}
where $\widetilde{\lambda}_{\widetilde{A}} > 0$ is a tuning parameter, and $D \in \R^{(n_0 - 1) \times n_0}$ is defined in \eqref{def-D}. 

The theoretical guarantees for $ \widetilde{f}^{\widetilde{A}}$ are derived below.

\begin{proposition}\label{theorem_l_0-affine}
Let the target data $\{y_i\}_{i= 1}^{n_0}$ be from \eqref{model-target} and  unisource data $\{y_i^{(1)}\}_{i=1}^{n_1}$ be from \eqref{model-aux}.  Assume that $ \{ \epsilon_i^{(1)}\}_{i=1}^{n_1}$ are mutually independent mean-zero $C_{\sigma}$-sub-Gaussian distributed with an absolute constant $C_{\sigma} >0$. Let $A \in \R^{n_1 \times n_0}$ and assume that there exists a matrix $\widetilde{A} \in \R^{n_0 \times n_1}$ such that $\widetilde{A}A = I_{n_0}$. Let $\widetilde{f}^{\widetilde{A}}$ be defined in \eqref{estimator_l_0-affine}, with tuning parameter 
\begin{align}\label{tuning-parameter-affine-0}
  \widetilde{\lambda}_{\widetilde{A}} = C_{\widetilde{\lambda}}\frac{ 1 + \log \big(n_0/(s_0+1) \big)}{n_0 / \|\widetilde{A}\|^2 },  
 \end{align}
where $C_{\widetilde{\lambda}}> 0$ is an absolute constant.  It holds with probability at least $1 - n_0^{-c}$ that 
 \[
  \big\|  \widetilde{f}^{\widetilde{A}} - f  \big\|_{1/n_0}^2
    \leq    C \frac{  (s_{0}+1) \big\{1+\log \big(n_0/(s_0+1) \big) \big\} + \|\delta^{A}\|_2^2}{n_0 / \|\widetilde{A}\|^2 } 
 \]
 where $\delta^A \in \R^{n_1}$ is defined in \eqref{def-delta-A}, and $C, c>0$ are absolute constants. 
\end{proposition} 

We observe that the existence of the left inverse implies that $n_1 \geq n_0$, which is assumed in \Cref{theorem_l_0-1}.  By \Cref{lemma-P}, we can see that \Cref{theorem_l_0-affine} is a generalisation of \Cref{theorem_l_0-1}, where $A = P^{n_1, n_0}$.

\subsection{Using target data for transfer learning in unisource scenarios}\label{sec:target-uni}

In \Cref{sec-one-source}, we assume $n_1 \geq n_0$ and the target data is not directly used.  The results and methods therein provide the cornerstone for an minimax optimal procedure presented in \Cref{sec-multiple-source}.  To use the target data more directly and more importantly, to cover the case when $n_1 < n_0$, we introduce the target-unisource-transferred $\ell_0$-penalised estimator
\begin{equation}\label{estimator_l_0-1_all}
    \widetilde{f}^{\{ 0, 1\}} = \widetilde{f}^{\{ 0, 1\}}(\widetilde{\lambda}) = \argmin_{\theta \in \R^{n_0}} \bigg\{  \frac{1}{2n_0} \Big\|\widetilde{P}^{n_0, n_1+n_0} \widetilde{y} - \theta \Big\|_2^2   + \widetilde{\lambda} \|D   \theta \|_0  \bigg\},  
\end{equation}
where $\widetilde{P}^{n_0, n_1+n_0} \in \R^{n_0 \times (n_1+n_0)}$ is defined in \eqref{def-P-alt},  $\widetilde{y} \in \R^{n_1 +n_0}$ with for any $i \in [n_1 + n_0]$,
\[
\widetilde{y}_i = 
\begin{cases} 
y_j & \mbox{if } i = \lceil j n_1/n_0 \rceil + j \mbox{ for some } j \in [n_0], \\
y^{(1)}_{i - \vert\{ \lceil j n_1/n_0 \rceil + j\colon  j \in [n_0] \} \cap [i] \vert} & \mbox{otherwise},
\end{cases}
\]
the quantity $\widetilde{\lambda} > 0$ is a tuning parameter and $D \in \R^{(n_0 - 1) \times n_0}$ is defined in \eqref{def-D}.  Compared to the estimator \eqref{estimator_l_0-1}, the core of the estimator \eqref{estimator_l_0-1_all} is to combine the source and target data as a $(n_1 + n_0)$-dimensional vector.  The theoretical guarantees for $\widetilde{f}^{\{ 0, 1\}}$ are collected below.

\begin{proposition}\label{theorem_l_0-1_all}
Let the target data $\{y_i\}_{i= 1}^{n_0}$ be from \eqref{model-target} and  unisource data $\{y_i^{(1)}\}_{i=1}^{n_1}$ be from \eqref{model-aux}.  Assume that $  \{ \epsilon_i\}_{i=1}^{n_0}  \cup \{ \epsilon_i^{(1)}\}_{i=1}^{n_1}$ are mutually independent mean-zero $C_{\sigma}$-sub-Gaussian distributed with an absolute constant $C_{\sigma} >0$.  Let $\widetilde{f}^{\{ 0, 1\}}$ be defined in \eqref{estimator_l_0-1_all}, with tuning parameter 
\begin{align}\label{tuning-parameter-0_all} 
  \widetilde{\lambda} = C_{\widetilde{\lambda}}\frac{ 1 + \log \big(n_0/(s_0+1) \big)}{n_1 +n_0},  
 \end{align}
where $C_{\widetilde{\lambda}}> 0$ is an absolute constant.  It holds with probability at least $1 - n_0^{-c}$ that 
 \[
    \big\|\widetilde{f}^{\{ 0, 1\}} - f  \big\|_{1/n_0}^2  
     \leq     C\frac{  (s_{0}+1) \big\{1+\log \big(n_0/(s_0+1) \big) \big\}+\|\delta\|_2^2 }{n_1 +n_0},
 \]
 where $\delta \in \R^{n_1}$ is defined in \eqref{delta-unisource}, and $C, c>0$ are absolute constants. 
\end{proposition}

\Cref{theorem_l_0-1_all} shows that the estimation error bounds in \eqref{estimator_l_0-1_all} achieve minimax optimality when $(n_1 + n_0)^{-1}\| \delta\|_2^2 \leq (s_0+1) \log\big(n_0/(s_0+1)\big) /n_0 $, supported by  \Cref{theorem_minimax} and \Cref{remark-minimax}.  The proof of \Cref{theorem_l_0-1_all} is given in \Cref{app-tl-1-all}.  Comparing \Cref{theorem_l_0-1_all} with \Cref{theorem_l_0-1}, we see that \Cref{theorem_l_0-1_all} notably allows $n_1 < n_0$.

\section{Numerical experiments}\label{sec-num}

In this section, we conduct numerical experiments to support our theoretical findings.  Simulated and real data analysis are in Sections~\ref{sec-simulation} and \ref{sec-real-data}, respectively.  The code and datasets are available online \footnote{\url{https://github.com/chrisfanwang/transferlearning}}.

\subsection{Simulation studies}\label{sec-simulation}

We evaluate the performance of our proposed methods for piecewise-constant mean estimation and compare them with existing methods.  

\medskip
\noindent \textbf{Estimators.}  The estimators considered include:
\begin{itemize}
    \item $\ell_1$-penalised estimator ($\ell_1$),
    \item $\ell_0$-penalised estimator ($\ell_0$),
    \item unisource-transferred $\ell_1$-penalised estimator ($\ell_1$-T-$1$), i.e.~\eqref{estimator_l_1-1},
    \item  unisource-transferred $\ell_0$-penalised estimator ($\ell_0$-T-$1$), i.e.~\eqref{estimator_l_0-1},
    \item multisource-transferred $\ell_1$-penalised estimator with known informative multisources ($\ell_1$-T-$\mathcal{A}$), studied in \Cref{cor-l0l1},
    \item multisource-transferred $\ell_0$-penalised estimator with known informative multisources ($\ell_0$-T-$\mathcal{A}$), studied in \Cref{cor-l0l1},
    \item multisource-transferred $\ell_1$-penalised estimator with informative sources learned by \Cref{alg_isd} ($\ell_1$-T-$\widehat{\mathcal{A}}$), studied in \Cref{cor-a-hat-l0l1}, and
    \item multisource-transferred $\ell_0$-penalised estimator with informative sources learned by \Cref{alg_isd}  ($\ell_0$-T-$\widehat{\mathcal{A}}$), studied in \Cref{cor-a-hat-l0l1}.
\end{itemize}   
For the small-scale simulation shown in \Cref{fig:all}, we consider two additional estimators:
\begin{itemize}
    \item all-source-transferred $\ell_1$-penalised estimator ($\ell_1$-$T$-$[K]$), i.e. \eqref{estimator-ora-1}, and
     \item all-source-transferred $\ell_0$-penalised estimator ($\ell_0$-$T$-$[K]$), i.e. \eqref{estimator-ora-0}.
\end{itemize}
R \citep{R} packages \texttt{genlasso} \citep{genlasso} and \texttt{changepoints} \citep{changepoints} are used for $\ell_1$- and $\ell_0$-penalised optimisations. 

\medskip 
\noindent \textbf{Evaluation.}  We report the mean squared estimation errors $\| f^{\mathrm{est}} - f \|_{1/n_0}^2$ in the form of mean and standard errors, where $ f^{\mathrm{est}} \in \R^{n_0}$ denotes an estimated target mean vector.

\medskip 
\noindent \textbf{Simulation setup.}
Two simulation scenarios are examined, each with $N=100$ Monte Carlo trials.  The number of source datasets $K=10$, the target dataset size $n_0 = 200$ and the source dataset sizes $n_k = 2n_0$ for all $k \in [K]$, stay fixed for all cases.  Additional simulation studies are conducted in \Cref{app-simu} for varying $n_0$.  Varying source frequencies cases are investigated in \Cref{sec-num-sup}.  We adopt uniform observational frequencies across multisources. The noise random variables $\{\epsilon_i\}_{i=1}^{n_0} \cup \{\epsilon_i^{(k)}\}_{i=1, k = 1}^{n_k, K} \overset{\mbox{i.i.d.}}{\sim} \mathcal{N}(0, \sigma^2)$ with $\sigma = 0.5$.  Simulations on dependent noise variables are presented in \Cref{app-simu}.

Two types of target signals $f \in \R^{n_0}$ are considered, with \textbf{Scenario 1} equally-spaced change points $\mathcal{S} = \{25, 50, \ldots,175\}$ and \textbf{Scenario 2} unequally-spaced change points $\mathcal{S} = \{20,  40,  50,   120,$ $134,   160,  176\}$.  Corresponding signal magnitudes at the change points are $\{ f_{t_1}, \ldots, f_{t_{8}}\} = \{2\gamma, 4\gamma, \gamma$, $5\gamma, 7\gamma, 8\gamma, 2\gamma, \gamma\}$, with $t_8 = n_0$ and $\gamma \in \{0.25, 0.5, 0.75, 1 \}$.

As for the signals in the source datasets, given an informative set $\mathcal{A}$ with $|\mathcal{A}| = a \in \{2, 4, 6, 8\}$, for $k \in [K]$, let $f^{(k)}_j = \big( P^{n_k, n_0} f \big)_j + \delta_{k,j} \mathbbm{1}_{\{j \in \mathcal{H}_k \}}$, with $\mathcal{H}_k = [Hn_k]$ and $H \in \{ 0.075, 0.15, 0.225, 0.3\}$.

We further consider two configurations.  
\begin{itemize}
\item  \textbf{Configuration 1.}  For $k \in [K]$ and $j \in [n_k]$, let $\delta_{k,j} = \alpha \mathbbm{1}_{\{k \in \mathcal{A}\}} + \widetilde{\alpha} \mathbbm{1}_{\{k \notin \mathcal{A}\}}$, with $\alpha \in \{0.1, 0.2, 0.3, 0.4\}$ and $\widetilde{\alpha} = 2$.

\item \textbf{Configuration 2.}  For $k \in [K]$ and $j \in [n_k]$, let $\delta_{k, j}$ be drawn independently and identically from the distribution $\mathcal{N}(0, \kappa)\mathbbm{1}_{\{k \in \mathcal{A}\}} + \mathcal{N}(0, \widetilde{\kappa}) \mathbbm{1}_{\{k \notin \mathcal{A}\}}$, with $\kappa  \in \{0.1, 0.2, 0.3, 0.4\}$ and $\widetilde{\kappa} = 5$.
We also conduct simulations with dependence across the discrepancy vector between the target data and sources, see \Cref{app-simu}.
\end{itemize}

In \textbf{Configuration 1}, there is a deterministic discrepancy between the source and target signal vectors, while \textbf{Configuration 2} considers a random discrepancy.

\medskip 
\noindent \textbf{Tuning parameters.}
The penalisation tuning parameters across all the estimators are selected via a $5$-fold cross-validation.  For $\ell_1$-T-$1$ and $\ell_0$-T-$1$ estimators, the first source in the informative set serves as the unisource data.  For $\ell_1$-T-$\widehat{\mathcal{A}}$ and $\ell_0$-T-$\widehat{\mathcal{A}}$ methods, the informative set is estimated through \Cref{alg_isd}. Tuning parameters in \Cref{alg_isd} include the screening size $\widehat{t}_k = 50$ and the threshold values $\tau_k = \tau$, for each $k\in [K]$.  The threshold level $\tau$ is determined via a permutation-based algorithm shown in \Cref{app-permu-alg}. 
A sensitivity study on the screening size is collected in \Cref{app-simu}.

\medskip
\noindent \textbf{Results.}  
The simulation results for \textbf{Scenarios 1} and \textbf{2} are collected in Figures~\ref{Fig_simulation_se_1} and \ref{Fig_simulation_se_2} (in Appendix \ref{app-simu}), respectively.  Some remarks are in order.
In both scenarios, there is a clear ranking in estimation performance. Estimators only using the target data are the worst. Transfer learning estimators utilising unisource data enhance the estimation performance with further improvement using estimated informative multisources. The best estimation performance is achieved when using predefined informative multisources. The performance of estimators utilising estimated informative multisourecs from Algorithm \ref{alg_isd} is comparable to cases where the informative set is predefined, thereby showing the resilience of the informative set detection algorithm in Algorithm \ref{alg_isd}.  We also see that $\ell_0$-penalised estimators outperform their $\ell_1$ counterparts, consistent with our theoretical results.

In both Figures~\ref{Fig_simulation_se_1} and \ref{Fig_simulation_se_2} (in Appendix \ref{app-simu}), panels (A), (B), (C) and (D) show that an increase in discrepancy levels (represented by $\alpha$ in \textbf{Configuration 1} and $\kappa$ in \textbf{Configuration 2}), or the changing frequencies of difference vectors $(H)$, intensifies the contrast between source signals and target signals. This greater contrast leads to increased estimation errors across all transfer learning methods. This finding echoes our theoretical results. In panels (E) and (F), we show that, when observational frequencies are uniform, there is a negative correlation between $a$ the cardinality of the informative set and the estimation errors of transfer learning estimators using multisources. This correlation aligns with the discussion in \Cref{sec-multiple-source}. In panels (G) and (H), we observe that as $\gamma$ the magnitude of change increases, the estimation errors of $\ell_1$-penalised estimators increase. This finding coincides with the expectation that higher variability leads to less stable estimation. Another noticeable fact in panels (G) and (H) is that estimation errors of $\ell_0$-penalised estimators do not show a similar increasing pattern.
We conjecture that $\ell_0$-penalised methods are inherently more robust to larger changes because $\ell_0$-penalisation directly penalises the number of change points rather than their magnitudes. This allows larger changes to help reduce the influence of noise and minor fluctuations, leading to more stable or even reduced estimation errors.  Lastly, we see that there is no significant difference between Figures~\ref{Fig_simulation_se_1} and \ref{Fig_simulation_se_2}, demonstrating the robustness of our methods against unbalanced change points.

\begin{figure}[t] 
\centering 
\includegraphics[width=0.8 \textwidth]{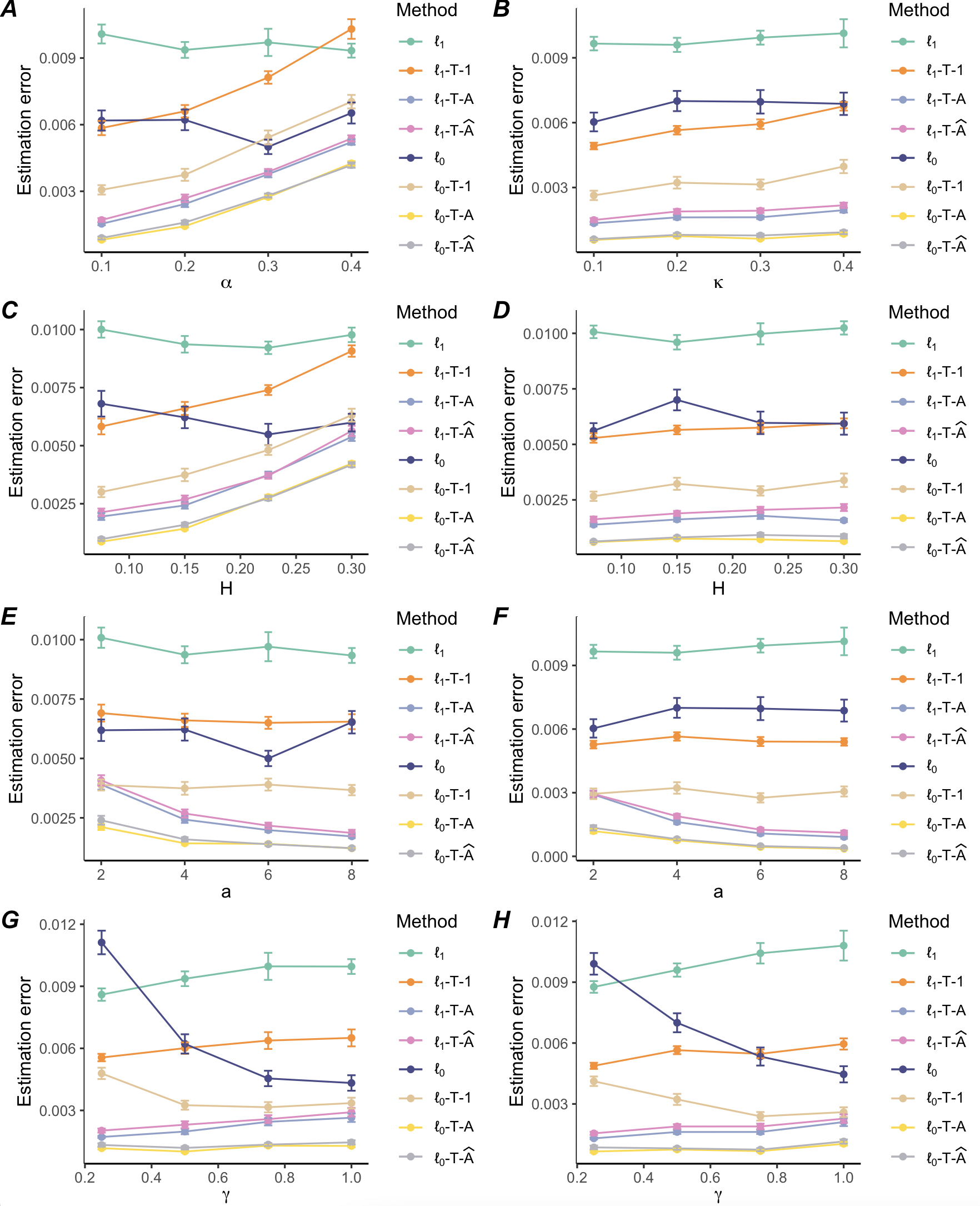}
\caption{Estimation results in Scenario 1. From left to right: Configurations 1 and 2.  From top to bottom: performances with varying discrepancy levels ($\alpha$ and $\kappa$), difference vector changing frequencies ($H$), cardinalities of the informative set ($a$) and change magnitudes~($\gamma$).} \label{Fig_simulation_se_1}
\end{figure}

\subsubsection{Simulation studies with varied source frequencies}\label{sec-num-sup}

To examine the performance of our proposed methods under varying source frequencies, we construct $10$ sources, each associated with a unique observational frequency, i.e.
\begin{equation}\label{eq_order}
    n_k = 200 \times (11-k), \quad k \in [10].
\end{equation}
We vary the number of observed sources $K$ from $1$ to $10$ and depict in \Cref{fig_add} the following quantity: 
\begin{equation}\label{eq-fig-2-quantity}
    K \times \mbox{the harmonic mean of source observations in } [K].
\end{equation}
      
Let the target data be constructed as \textbf{Scenario 1}, with the parameter $\gamma = 0.5$. For $k \in [K]$, let the $k$th source data follow \textbf{Configuration 1} with the specified parameters $\mathcal{A} = [K]$, $H = 0.15$ and $\alpha = 0.2$. The estimators considered include multisource-transferred $\ell_1$- and $\ell_0$-penalised estimators which use
\begin{itemize}
    \item all sources in $[K]$, i.e.~\eqref{estimator-ora-1} and \eqref{estimator-ora-0}, and
    \item a set of sources identified in the selection step, $\widetilde{\mathcal{A}}$, i.e.~\eqref{eq-widetilde-a} with $\widehat{\mathcal{A}} = [K]$. 
\end{itemize} 
Evaluations and choices of tuning parameters remain the same as introduced before.  The simulation results are depicted in \Cref{fig_add}.

\begin{figure}[t] 
\centering 
\includegraphics[width=0.8 \textwidth]{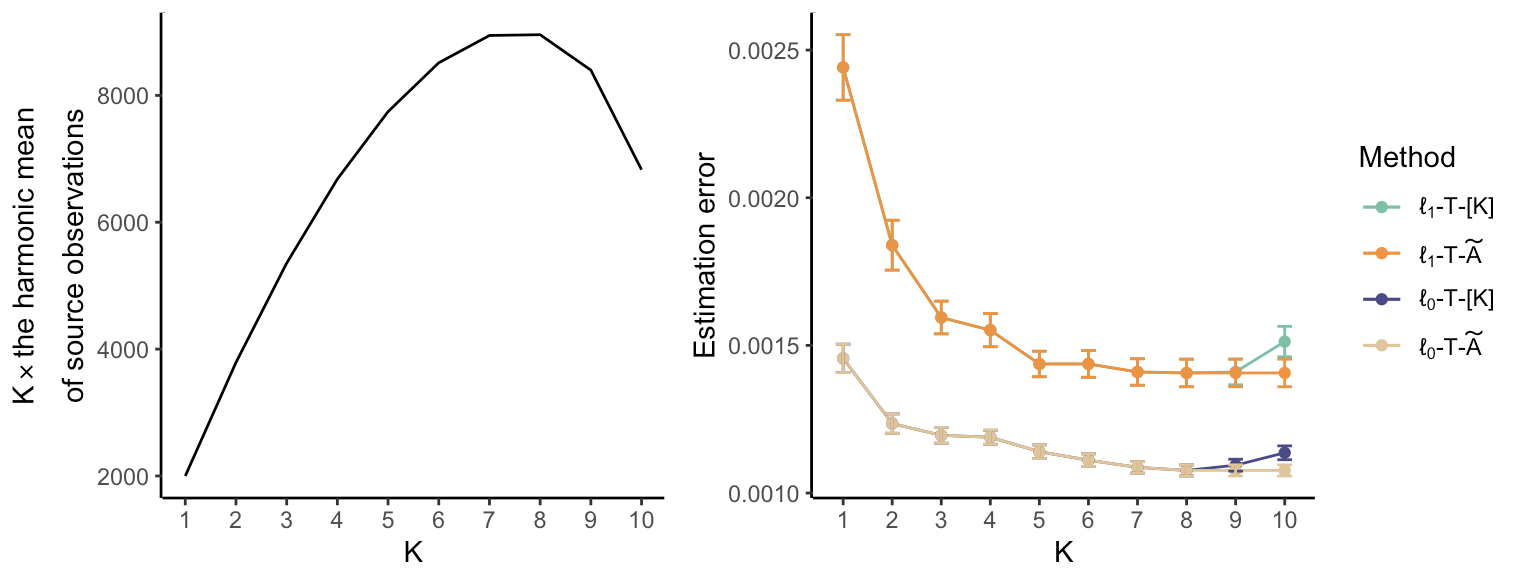}
\caption{Results of \Cref{sec-num-sup}. Left panel:  Relationship between $K$ the number of sources and the quantity in \eqref{eq-fig-2-quantity}, with sample sizes detailed in \eqref{eq_order}. Right panel: Estimation results as $K$ varies. }\label{fig_add}
\end{figure}

The left panel in \Cref{fig_add} shows that increasing the number of source datasets~$K$ does not necessarily raise the value $K^2 (\sum_{k \in [K]}{n_k}^{-1})^{-1}$, which is the denominator in the fluctuation term in our estimation error bounds. This observation aligns with the discussion in \Cref{sec-optional-selction}.  We, again, highlight that our focus is not simply on increased source sample sizes, but on increased source observational frequencies and guarantee adaptability to varying observational frequencies across multiple sources. The study by \cite{cai2023transfer} in the high-frequency functional data analysis assumed a uniform observation across multisources, overlooking the interesting features.

From the right panel in \Cref{fig_add}, we observe that transferred estimators using all sources in~$[K]$ show a turning point at $K = 8$, where estimation errors start to increase. This turning point matches the one in the left panel, where the term $K^2 (\sum_{k \in [K]}{n_k}^{-1} )^{-1}$ shifts from an increasing to decreasing trend.  Comparing the performances of transferred estimators that utilise all sources to those with selected sources, we see that an additional source selection step, as shown in \Cref{sec-optional-selction},  ensures the precision of transferred estimators remains non-decreasing as the number of sources grows.  In the varied source frequency framework, hence, simply adding more beneficial sources does not necessarily improve precision.

\subsubsection{Simulation studies with incorporating target data}\label{sec-num-targrt}

To assess the performance of the estimators proposed in \Cref{sec-extensions}, we conduct simulations comparing transfer learning estimators
\begin{itemize}
    \item for unisource scenarios, i.e.~\eqref{estimator_l_1-1} and \eqref{estimator_l_0-1}, and
    \item for non-selective multisource scenarios, i.e.~\eqref{estimator-ora-1} and \eqref{estimator-ora-0}, 
\end{itemize}
 with those incorporating target data 
 \begin{itemize}
     \item for unisource scenario studied in \Cref{sec:target-uni}, and
     \item for non-selective multisource scenarios studied in \Cref{sec:target-multi}.
 \end{itemize}

The target data are generated according to \textbf{Scenario 1}, with the parameter $\gamma = 0.5$. For $k \in [K]$, let the $k$th source data follow \textbf{Configuration 1} with the specified parameters $\mathcal{A} = [K]$, $H = 0.15$ and $\alpha = 0.2$. The procedure for evaluating and selecting tuning parameters remains as previously described.   The simulation results are provided in \Cref{fig:extension}, where we observe that numerically incorporating target data is beneficial for both unisource and non-selective multisource scenarios.

\begin{figure}[t]
        \centering
        \includegraphics[width=0.45 \textwidth]{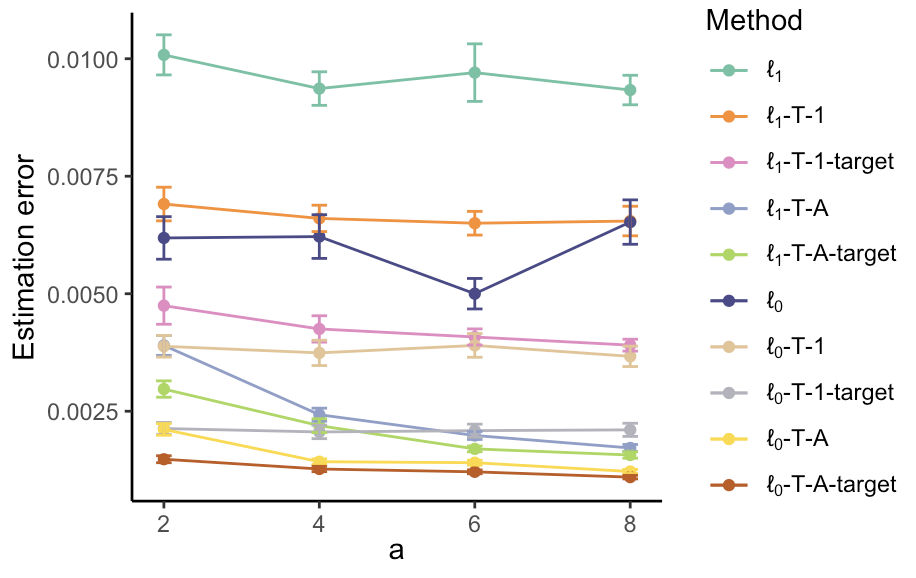}
        \caption{Estimation results for Configuration 1 and Scenario 1 in \Cref{sec-simulation} with estimators defined in \Cref{sec-num-targrt}, the number of source datasets  $K = 10$ and the cardinality of the informative set $a \in \{2, 4, 6, 8\}$.}
        \label{fig:extension}
\end{figure}

\subsection{Real data analysis}\label{sec-real-data}
Consider three real datasets: the U.S.~electric power operations dataset \citep{U.S.Energy},  the GDP \citep{gdpdata} \& IP \citep{ipdata} dataset, and the air quality dataset \citep{airdata}.  
This section focuses exclusively on the analysis of the U.S. electric power operations dataset and the GDP \& IP dataset. The
analysis of the air quality dataset can be found in \Cref{app-real-data}. 
All methods listed in \Cref{sec-simulation} are considered, except those that incorporate known informative multisources. Tuning parameters involved are chosen following \Cref{sec-simulation}.  To evaluate different estimators, we split the target dataset into training and test datasets. Estimators are derived using the training dataset, while the mean squared errors are computed using the test dataset.

\medskip
\noindent\textbf{The U.S.~electric power operations dataset} includes daily records of the electrical power demand of the various regions and sub-balancing authorities within the U.S.~electricity market. Our study specifically focuses on the New York Independent System Operator, which consists of $11$ distinct sub-regions:~Capital, Central, Dunwoodie, Genesee, Hudson Valley, Long Island, Millwood, Mohawk Valley, New York City, North and West.

We conduct two separate analyses, using data collected every Saturday from 2nd July 2020 to 1st July 2023 ($156$ days) for New York City and Central sub-regions as the target datasets.  For both analyses, daily observations from other sub-regions within the same time frame ($1092$ days) serve as source data.  For transfer learning estimators utilising unisource data, the Dunwoodie sub-region is chosen when New York City is the target, due to their similar urban characteristics. For Central being the target, the Mohawk Valley sub-region is selected, given their geographical alignment. 
We also conducted additional analyses using different sub-regions as unisource with results in \Cref{app-real-data}. We split the target dataset with the sample size $n_0 = 156$ into even-week Saturdays as training $y^{\mathrm{train}}$ and odd-week as test $y^{\mathrm{test}}$ datasets.  Using  $y^{\mathrm{train}}$, we obtain the estimated mean vector $f^{\mathrm{est}}$ and then report the mean squared prediction errors $\| f^{\mathrm{est}} - y^{\mathrm{train}} \|_{2/n_0}^2$. Results are shown in \Cref{table:estimation_results}.

\Cref{table:estimation_results} indicates that transfer learning methods, especially those using estimated informative multisources, outperform traditional methods when estimating electricity consumption in both target regions. This emphasises the advantage of leveraging source data to enhance estimation precision.  It is worth mentioning that simply including all multisources without selection for transfer might not improve estimation performance. The results could even be worse than those obtained only using the target data. This highlights the necessity of the informative set detection algorithm, as shown in \Cref{alg_isd}.  The estimated informative sets through \Cref{alg_isd}, consist of sub-regions with electricity consumption patterns similar to those of the target regions. For instance, when New York City is the target, the estimated informative set features Long Island, likely because of their shared urban characteristics and geographic proximity. 

\medskip
\noindent \textbf{The GDP \& IP  dataset} includes quarterly records of GDP, presented as the percentage change compared to the same period in the previous year, and monthly records of the IP data, measured as an index based on the reference period, from various countries worldwide.

As illustrated in  \Cref{sec-intro}, we use data collected from  Q1-2000 to Q4-2022 (92 quarters) for Hungary as the target dataset. Monthly records of IP data from  12 different countries (Bulgaria, Croatia, Czechia, Estonia, Greece, Hungary, Latvia, Lithuania, Poland, Romania, Slovak Republic and Slovenia)  within the same duration (276 months) serve as the source data. For the source data, each entry is adjusted to represent the change compared to the same period in the previous year.
For transfer learning estimators utilising unisource data, Hungary's IP dataset is selected as the source. The target dataset is split into $y^{\mathrm{train}}$ even quarters as training dataset and $y^{\mathrm{test}}$ odd quarters as test dataset.  Using  $y^{\mathrm{train}}$, we obtain the estimated mean vector $f^{\mathrm{est}}$ and then report the mean squared prediction errors $\| f^{\mathrm{est}} - y^{\mathrm{train}} \|_{2/n_0}^2$. Results are presented in \Cref{table:estimation_results}.

From \Cref{table:estimation_results}, we observe that in Hungary's GDP trend estimation, transfer learning methods show a significant advantage over traditional methods. The performance of using all multisource methods and informative multisource methods, however, is comparable. We conjecture that this is due to all the source data coming from East European countries, which likely share some similarities such as economic structures and industrial patterns.

\begin{table}[t]
\centering
\caption{Results for New York City and Central sub-regions in the U.S.~electric power operations dataset, and Hungary's GDP in the GDP \& IP dataset.} 
\begin{tabular}{lcccccccc}
\hline\hline
 & $\ell_1$ & $\ell_1$-T-1 & $\ell_1$-T-$\widehat{\mathcal{A}}$ & $\ell_1$-T-$[K]$ & $\ell_0$ & $\ell_0$-T-1 & $\ell_0$-T-$\widehat{\mathcal{A}}$ & $\ell_0$-T-$[K]$ \\
\hline
New York City & 0.3734 & 0.1952 & 0.1955 & 0.4586 & 0.4131 & 0.3052 & 0.3061 & 0.5813 \\
Central & 0.6390 & 0.3026 & 0.2789 & 0.4025 & 0.7638 & 0.4298 & 0.4521 & 0.5588 \\
Hungry & 0.9260 & 0.5782 & 0.5433 & 0.5022 & 0.7823 & 0.7279 & 0.6787 & 0.6968 \\
\hline
\end{tabular}
\label{table:estimation_results}
\end{table}

\section{Conclusions}\label{sec-conclusion}

In this paper, we study transfer learning for the estimation of piecewise-constant signals, which is the first time seen in the literature. Our approaches leverage higher observational frequencies and accommodate diverse observational frequencies across multiple sources.  We consider both $\ell_1$- and $\ell_0$-penalisation. The theoretical advantages of the transferred $\ell_0$-penalised estimator include its independence from the minimal length condition and its reduced reliance on unknown parameters when selecting tuning parameters, compared to its $\ell_1$-penalised counterpart. 

The current work offers several interesting directions for future studies. Firstly, the foundational frameworks and methodologies used for the source and target models in this study can be generalised to transfer learning for piecewise-polynomial mean estimation.  Specifically, we can investigate the trend filtering method \citep[e.g.][]{tibshirani2014adaptive, ortelli2019prediction, guntuboyina2020adaptive}, which incorporates the $r$th order difference operator $D^{(r)}$, a generalisation of the difference operator $D$ in \eqref{def-D}.    The associated challenges lie in the extension of the alignment operator $P^{n_k, n_0}$, defined in \eqref{def-P}, and the characterisation of its associated eigenvalue spectrum,  like \Cref{lemma-P}.  Secondly, an extension to transfer learning for high-dimensional linear regression models with general designs and piecewise-constant regression coefficients \citep{wang2022denoising, xu2019iterative}, can be studied. The main challenges are to establish measures for the discrepancies between the source and the target regression coefficients and the covariance matrices of their respective covariates.

\section*{Acknowledgements}

Wang is supported by Chancellor's International Scholarship, University of Warwick. Yu is partially supported by the EPSRC and Leverhulme Trust.

\clearpage

\bibliographystyle{plainnat}
\bibliography{references.bib}

\appendix

\section*{Appendices}\label{appendices}

All technical details of this paper can be found in the Appendices. The proofs of \Cref{theorem_l_1-1}, \Cref{theorem_l_0-1}, and the theoretical results in \Cref{sec-multiple-source} and \Cref{sec-extensions} are presented in Appendices \ref{app_1}, \ref{app_2}, \ref{app_3} and \ref{app_4} respectively.  An extension to methodologies in \Cref{sec:target-uni} is presented in \Cref{sec:target-multi}.  Details and results in \Cref{sec-num} are collected in \Cref{sec-add-num}.

\section[]{Proof of Theorem \ref{theorem_l_1-1}}\label{app_1}

The proof of \Cref{theorem_l_1-1} can be found in  \Cref{app_1_proof}.  Relevant notation is provided in \Cref{app_1_notation} and all necessary auxiliary results are in \Cref{app_1_lemma}.

\subsection{Additional notation}\label{app_1_notation}

Let us consider $\mathcal{S} \subseteq [n_0-1]$ as defined in \eqref{def-S} with cardinality $s_0$. 
When  $s_0 >0$, denote  $\mathcal{S} = \{ t_1, \dots, t_{s_0}\}$ and let the sign vector $ q \in \R^{s_0}$ be defined as 
\begin{equation}\label{def-q}
   q_{i} = \mathrm{sign} \big( ( D f )_{t_i}\big), \quad \mbox{for each }  i \in [s_0] 
\end{equation} 
with $D$ defined in \eqref{def-D}, and the set $\mathcal{S}_{\pm} \subseteq [s_0 +1 ]$ be defined as
\begin{equation}\label{def-S_pm}
    \mathcal{S}_{\pm} = \{ i \in \{2, \dots, s_{0} \}\colon q_{i} q_{i-1} =-1 \} \cup \{1,  s_{0}+1 \}.
\end{equation}
Using the difference operator $D$ defined in \eqref{def-D}, if $s_0 < n_0-1$,  let $\Psi^{-\mathcal{S}}  \in \R^{n_0 \times (n_0-1-s_0)}$ be defined as   
\begin{align}\label{def-psi}
\Psi^{-\mathcal{S}} = D _{- \mathcal{S}}^{\top}( D _{- \mathcal{S}} D_{- \mathcal{S}}^{\top})^{-1},
\end{align}  
which, as investigated by \cite{ortelli2019prediction},  is well-defined.

We then introduce the definition of effective sparsity for target signals $f$, as proposed by \cite{ortelli2019prediction}. This can be seen as a substitute for sparsity and then can be used to establish upper bounds.

\begin{definition}[Effective sparsity] \label{def-effective sparsity}
Let the set $\mathcal{S} \subseteq [n_0-1]$  with cardinality $s_0$ be defined in \eqref{def-S}.  If $s_0 >0$, let the sign vector $q \in \R^{s_0}$  be defined in \eqref{def-q}. 
The effective sparsity of target signals $f$,  denoted as $\Gamma_\mathcal{S}^2$, is defined as
\begin{align}
	\Gamma_{\mathcal{S}}^2 = \begin{cases}
        \big\{  \max \big\{   \sum_{j=1}^{n_0-1-s_0} \big\vert\big( 1- w^{- \mathcal{S}}_j \big) \big((D\theta )_{-\mathcal{S}}\big)_j \big\vert\colon \|\theta\|_{1/n_0} =1\big\} \big\}^2,  &\mbox{if } s_0 = 0, \\
	  \big\{  \max \big\{ q^{\top}  (D\theta)_{\mathcal{S}}\colon \|\theta\|_{1/n_0} =1 \big\} \big\}^2,   &  \mbox{if }  s_0 = n_0 - 1, \\
    \big\{  \max \big\{ q^{\top}  (D\theta)_{\mathcal{S}}   - \sum_{j=1}^{n_0-1-s_0} \big \vert\big( 1- w^{- \mathcal{S}}_j \big) \big((D\theta )_{-\mathcal{S}}\big)_j \big\vert\colon \|\theta\|_{1/n_0} =1 \big\} \big\}^2,    & \mbox{otherwise}.
	\end{cases}
  \nonumber
\end{align}
Here, if $s_0 < n_0 - 1$, the vector $ w^{- \mathcal{S}} \in [0, 1]^{n_0-1-s_0}$ is defined as
\begin{align}
	 w^{- \mathcal{S}}_j =  \big\|\Psi^{ - \mathcal{S}}_{, j} \big\|_{1/n_0} \Big( \max_{j \in [n_0-1-s_0]} \big\|\Psi^{ - \mathcal{S}}_{, j} \big\|_{1/n_0} \Big)^{-1}, \quad  j \in [n_0-1-s_0], \nonumber
\end{align} 
with $j$th column of  $\Psi^{ - \mathcal{S}}$ denoted as  $\Psi^{ - \mathcal{S}}_{, j}$.  
\end{definition}

\subsection{Proof of Theorem \ref{theorem_l_1-1}}\label{app_1_proof}

\begin{proof}[Proof of Theorem \ref{theorem_l_1-1}]
This proof consists of four steps. In \textbf{Step 1}, we decompose our target quantity into several terms.
In \textbf{Step 2} and \textbf{Step 3}, we deal with these terms separately.  In \textbf{Step 4}, we gather all the pieces and complete the proof.

\medskip
\noindent\textbf{Step 1.}
    It directly follows from the definition of  $\widehat{f}$ that
\begin{align}
&  \frac{1}{2n_0} \big\|\widetilde{P}^{n_0, n_1} y^{(1)} -  \widehat{f}  \big\|_2^2 +\lambda \|D \widehat{f} \|_1    \leq \frac{1}{2n_0} \big\|\widetilde{P}^{n_0, n_1} y^{(1)}   -   f  \big\|_2^2 +\lambda \|D  f \|_1.  \nonumber
\end{align}
Given that $y^{(1)} = f^{(1)} + \epsilon^{(1)}$ with $f^{(1)} = P^{n_1, n_0} f + \delta$, we obtain that 
\begin{align}
  \frac{1}{2n_{0}} \big\|\widehat{f} - \widetilde{P}^{n_0, n_1} P^{n_1, n_0} f  \big\|_2^2  
\leq  &    \frac{1}{2n_{0}} \big\|f - \widetilde{P}^{n_0, n_1} P^{n_1, n_0} f  \big\|_2^2  +   \frac{1}{n_0}\widetilde{\epsilon}^{\top}\big( \widehat{f}  - f \big)   + \lambda \|D   f \|_1 
\nonumber\\
& \hspace{0.5cm}-  \lambda   \|D   \widehat{f} \|_1 + \frac{1}{n_0}\big(   \widetilde{P}^{n_0, n_1} \delta\big)^{\top} \big( \widehat{f}  - f \big), \nonumber 
\end{align}
with $\widetilde{\epsilon} =\widetilde{P}^{n_0, n_1} \epsilon^{(1)} \in \R^{n_0}$. By $n_1 \geq n_0$ and \Cref{lemma-P}, it holds that
\begin{align}\label{theorem1_1}
 \frac{1}{2n_{0}} \big\|\widehat{f} - f  \big\|_2^2   
\leq  &   \frac{1}{n_0}\widetilde{\epsilon}^{\top} \big( \widehat{f}  - f \big) 
  + \lambda \|D   f \|_1 - \lambda   \|D   \widehat{f} \|_1
+ \frac{1}{n_0}\big(   \widetilde{P}^{n_0, n_1} \delta\big)^{\top} \big( \widehat{f}  -   f\big) \nonumber\\
= &  (I.1) + (I.2) + (I. 3) + (II) =(I) + (II).
\end{align}

\medskip
\noindent\textbf{Step 2.} In this step, we deal with  the term $(I)$ in \eqref{theorem1_1}.
We claim that if 
\begin{align}\label{theorem1_lambda}
    \lambda = C_{\lambda}\big(n_{\max}^{(0)}/(n_1 n_0)\big)^{1/2},    
\end{align}
with $n_{\max}^{(0)}$  defined in \Cref{ass-fusedlasso} and an absolute constant $C_{\lambda} >0$,  then it holds that 
\begin{align}\label{theorem1_2}
 \P  \Bigg\{  (I)
    \leq &   \frac{3}{8n_0} \big\|  \widehat{f} - f  \big\|_2^2  + C_{\mathcal{E}} \frac{n_{\max}^{(0)} /\widetilde{n}_{\min}^{(0)} (s_0+1) \big(1 +\log(n_{\max}^{(0)}) \big)}{n_1}  \Bigg\} \geq 1 - n_0^{-c_{\mathcal{E}}},
 \end{align}
with $\widetilde{n}_{\min}^{(0)}$ defined in \Cref{remark_l_1_one} and absolute constants $C_{\mathcal{E}}, c_{\mathcal{E}} >0$. Then we prove this claim under two scenarios $s_0 = n_0 - 1$ and $s_0 < n_0 - 1$ in \textbf{Step 2.1} and \textbf{Step 2.2}, respectively. Before proving the claim, by \Cref{lemma-alignment-error} and the assumption $n_1 \geq n_0$, we have that 
\begin{equation}\label{theorem1_3}
\{ \widetilde{\epsilon}_i\}_{i=1}^{n_0} \overset{\mbox{ind.}}{\sim}  \mbox{mean-zero }C_{\sigma}\big(2 n_0/ n_1\big)^{1/2}\mbox{-sub-Gaussian}.
\end{equation}

\medskip
\noindent\textbf{Step 2.1.} In this step, we prove the claim stated in \eqref{theorem1_2} when $s_0 = n_0 - 1$.
By \eqref{theorem1_3} and general Hoeffding inequality \citep[e.g. Theorem 2.6.3 in][]{vershynin2018high},  we can conclude that  there exists an absolute constant  $c_0 >0$ such that
\begin{align}\label{theorem1_4}
\P \big\{ \mathcal{E}_1 \big\} \geq 1 - \exp\{ -c_0 n_0 \} \quad \mbox{with} \quad 
\mathcal{E}_1 =  \big\{ (I) \leq n_1^{-1/2} \big\|  \widehat{f} - f  \big\|_2 + \lambda \|D   f \|_1 - \lambda   \|D   \widehat{f} \|_1 \big\}.
\end{align}
From now on, we assume that the event $\mathcal{E}_1 $ holds in this sub-step.
By applying Cauchy--Schwartz inequality and the fact that $\vert ab \vert \leq a^2+ b^2/4$, we obtain that 
\begin{align}\label{theorem1_5}
(I) \leq  \frac{1}{4n_0}\big\|  \widehat{f} - f  \big\|_2^2  + \frac{n_0}{n_1} +  \lambda \|D   f \|_1 - \lambda   \|D \widehat{f} \|_1.
 \end{align}
 With the sign vector $q \in \R^{s_0}$ defined in \eqref{def-q}, we have that 
\begin{align}\label{theorem1_6}
\|Df  \|_1  = \big\|(D   f)_{\mathcal{S}} \big\|_1  = q^{\top}(D   f)_{\mathcal{S}}  \quad \mbox{and}  \quad \|D \widehat{f} \|_1  = \big\|(D \widehat{f})_{\mathcal{S}}  \big\|_1  \geq q^{\top}(D   \widehat{f})_{\mathcal{S}}.
\end{align}
Combining \eqref{theorem1_5} and  \eqref{theorem1_6}, we have that 
 \begin{align}\label{theorem1_7}
    (I) \leq  & \frac{1}{4n_0}\big\|  \widehat{f} - f  \big\|_2^2  + \frac{n_0}{n_1} + \lambda  q^{\top}\big(D  (f- \widehat{f}) \big)_{\mathcal{S}} 
    \leq  \frac{1}{4n_0}\big\|  \widehat{f} - f  \big\|_2^2  + \frac{n_0}{n_1} + \lambda \Gamma_{\mathcal{S}} \big\|  \widehat{f} - f  \big\|_{1/n_0}  \nonumber\\
    \leq & \frac{3}{8n_0}\big\|  \widehat{f} - f  \big\|_2^2  + \frac{n_0}{n_1} + 2 \lambda^2 \Gamma_{\mathcal{S}}^2
    \leq \frac{3}{8n_0}\big\|  \widehat{f} - f  \big\|_2^2 +\big(C_{\lambda}^2 C_{\Gamma}+1\big) \frac{(s_0+1)} {n_1 }, 
 \end{align} 
 where \begin{itemize}
     \item  the second inequality follows from the definition of the effective sparsity $\Gamma_{\mathcal{S}}$ in \Cref{def-effective sparsity},
     \item  the third inequality is based on the fact that $\vert ab \vert \leq 2a^2+ b^2/8$,
     \item and the final inequality follows from the choice of $\lambda$ in \eqref{theorem1_lambda}, \Cref{lemma-Gamma} and for any $i \in [s_0+1]$, $n_{i}^{(0)} = 1$ and $n_{\max}^{(0)} = 1$.
 \end{itemize}
Since when $s_0 = n_0 - 1$, we have $\widetilde{n}_{\min}^{(0)} = n_{\max}^{(0)} = 1$, then combining \eqref{theorem1_4} and \eqref{theorem1_7}, it holds  with an absolute constant $c_1 >0$ that
\begin{align}
     \P  \Bigg\{  (I)
    \leq &   \frac{3}{8n_0} \big\|  \widehat{f} - f  \big\|_2^2  + \big(C_{\lambda}^2 C_{\Gamma}+1\big) \frac{n_{\max}^{(0)} /\widetilde{n}_{\min}^{(0)} (s_0+1) \big(1 +\log(n_{\max}^{(0)}) \big)}{n_1}  \Bigg\} \geq 1 - n_0^{-c_{1}}, \nonumber
\end{align}
which proves \eqref{theorem1_2} when $s_0 = n_0 -1$.

\medskip
\noindent\textbf{Step 2.2.} In this step, we prove the claim stated in \eqref{theorem1_2} when $s_0 < n_0 - 1$.

By \eqref{theorem1_3}  and \Cref{theorm-error-fusedlasso}, we obtain  that    $\P \{ \mathcal{E}_2 \} \geq 1 - \exp \big\{-c_2 (s_0+1) \log\big(n_0/(s_0+1) \big) \big\}$  with 
\begin{align}
  \mathcal{E}_2 =    \Bigg\{  (I.1)
    \leq & \frac{1}{4n_0} \big\|  \widehat{f} - f  \big\|_2^2 + C_2  \frac{(s_0+1) \log( n_{\max}^{(0)})}{n_1}  + \lambda 
 \sum_{i=1}^{n_0-1-s_0} \big\vert w^{-\mathcal{S}}_i \big( ( D\widehat{f} - D f)_{-\mathcal{S}} \big)_{i}  \big\vert \Bigg\}, \nonumber
 \end{align}
where $n_{\max}^{(0)}$ is defined in \Cref{ass-fusedlasso} and $C_2, c_2 > 0 $ are absolute constants, if  $\lambda$ satisfies 
 \begin{align}\label{theorem1_lambda_check}
 \lambda \geq   \lambda_{\mathcal{S}} \quad \mbox{with} \quad 
 \lambda_{\mathcal{S}} = C_{\lambda_0}  \max_{j \in [n_{0} - 1- s_{0}]} \big\|\Psi^{ -\mathcal{S}}_{,j} \big\|_{1/n_{0}}    n_1^{-1/2} 
  \leq  C_{\lambda_0} \sqrt{ \frac{ n_{\max}^{(0)}}{2 n_1 n_0}},  
\end{align}
where the last inequality follows from \Cref{lemma-Psi} and $ C_{\lambda_0}  >0$ is an absolute constant.
Note that the function
\[
    s \mapsto  -c_2 (s +1) \log ( n_{0}/ (s + 1)))
\]
is convex, so its maximum over $[ 0:(n_{0}-2)]$ is attained at either $s =0$ or $s = n_{0} - 2$.  Thus, it holds with an absolute constant $c_3>0$ that
\begin{align}\label{theorem1_8}
\P \{ \mathcal{E}_2 \} \geq 1 - \max \big\{  \exp \{ - c_{2}\log ( n_{0} ) \},  \exp \{  -c_{2} (n_{0} -1)\log ( n_{0}/(n_{0}-1) ) \} \big\} \geq 1- n_{0}^{-c_3}.  
\end{align}

 From now on we assume that the event $\mathcal{E}_2$ holds in this sub-step.   Note that 
\begin{align}\label{theorem1_9}
&  \lambda \sum_{i=1}^{n_0-1-s_0} \big\vert w^{-\mathcal{S}}_i \big(( D   \widehat{f} - D f  )_{-\mathcal{S}} \big)_{i}  \big\vert   + \lambda \|D   f \|_1 - \lambda \|D   \widehat{f} \|_1   \nonumber \\
= &  - \lambda \sum_{i=1}^{n_0-1-s_0} \big( 1- w^{-\mathcal{S}}_i \big) \big\vert \big(( D   \widehat{f} - D f  )_{-\mathcal{S}} \big)_{i}  \big\vert   + \lambda  \big\|( D   \widehat{f} - D f  )_{-\mathcal{S}} \big\|_1  + \lambda \|D   f \|_1 - \lambda \|D   \widehat{f} \|_1   \nonumber \\
\leq&    - \lambda \sum_{i=1}^{n_0-1-s_0} \big( 1- w^{-\mathcal{S}}_i \big) \big\vert \big(( D   \widehat{f} - D f  )_{-\mathcal{S}} \big)_{i}  \big\vert   +\lambda \big\|(D   f)_{\mathcal{S}} \|_1 - \lambda \big\| (D   \widehat{f})_{\mathcal{S}}\|_1 +  2\lambda  \big\| ( D f )_{-\mathcal{S}} \big\|_1   \nonumber\\
\leq&   - \lambda \sum_{i=1}^{n_0-1-s_0} \big( 1- w^{-\mathcal{S}}_i \big) \big\vert \big(( D   \widehat{f} - D f  )_{-\mathcal{S}} \big)_{i}  \big\vert   + \lambda q^{\top} \big(D (f - \widehat{f})_{\mathcal{S}} \big)+  2\lambda  \big\| ( D f )_{-\mathcal{S}} \big\|_1   \nonumber\\
\leq&  \lambda \Gamma_\mathcal{S} \big\| \widehat{f} - f  \big\|_{1/n_0}+  2\lambda  \big\| ( D f )_{-\mathcal{S}} \big\|_1 
\leq \frac{1}{8n_{0}}\big\| \widehat{f} - f  \big\|_2^2 +    2\lambda^2 \Gamma_\mathcal{S}^2 +  2\lambda  \big\| ( D f )_{-\mathcal{S}} \big\|_1,  
 \end{align}
where
\begin{itemize}
    \item the first equality is based on the fact that $w^{-\mathcal{S}} \in [0, 1]^{n_0-1-s_0}$ defined in \Cref{def-effective sparsity}, 
    \item  the first inequality follows from the reverse triangle inequality,
    \item  the second inequality is based on the fact that $ \big\|(D   f)_{\mathcal{S}} \big\|_1  = q^{\top}(D   f)_{\mathcal{S}}$ and  $  \big\|(D \widehat{f})_\mathcal{S}  \big\|_1  \geq q^{\top}(D   \widehat{f})_{\mathcal{S}}$. Specifically, the sign vector $q$ is defined in \eqref{def-q} when $s_0 >0 $ and is set to  $q = 0$ for $s_0 = 0$,
    \item  the third inequality follows from the definition of the effective sparsity $\Gamma_\mathcal{S}$ in \Cref{def-effective sparsity},
     \item  the final inequality is based on the fact $\vert ab \vert \leq 2a^2+ b^2/8$.
\end{itemize}
By the construction of $\mathcal{S}$ in \eqref{def-S}, it holds that   
\begin{align}\label{theorem1_10}
       \| ( D f )_{-\mathcal{S}} \|_1   = 0.
\end{align}
By \Cref{lemma-Gamma}, we obtain a deterministic result with an absolute  constant $C_{\Gamma}>0$ as follows
\begin{align}
   \Gamma_\mathcal{S}^2 \leq  \begin{cases}
   C_{\Gamma}  \log(n_{0}), \quad &\mbox{if } s_0 = 0, \\
        C_{\Gamma}  n_{0} \bigg( \sum_{i \in \mathcal{S}_{\pm}}\frac{1+\log ( n^{(0)}_i )}{n^{(0)}_i } +  \sum_{i \in  S \setminus  \mathcal{S}_{\pm}} \frac{1+\log (n^{(0))}_i)}{n^{(0)}_{\max}  } \bigg)  , \quad & \mbox{otherwise}.
           \end{cases} \nonumber
 \end{align}
 Then it holds with $n_{\max}^{(0)}$ and $\widetilde{n}_{\min}^{(0)}$ defined in \Cref{ass-fusedlasso} and \Cref{remark_l_1_one}, respectively, that 
 \begin{align}\label{theorem1_11}
       \Gamma_\mathcal{S}^2 \leq 
       C_{\Gamma} \frac{n_{0} (s_{0}+1) \big(1 +\log(n_{\max}^{(0)}) \big)}{ \widetilde{n}_{\min}^{(0)} }.
 \end{align}
Then combining \eqref{theorem1_8}, \eqref{theorem1_9}, \eqref{theorem1_10} and \eqref{theorem1_11}, with $n_{\max}^{(0)} \geq n_0/(s_0+1)$ and the choice of $\lambda$ in \eqref{theorem1_lambda} which satisfies \eqref{theorem1_lambda_check}, it holds with an absolute constant $C_3 > 0$ that
 \begin{align}
 \P \Bigg\{(I) \leq &     \frac{3}{8n_0} \big\|  \widehat{f} - f  \big\|_2^2  + C_3 \frac{n_{\max}^{(0)} /\widetilde{n}_{\min}^{(0)} (s_0+1) \big(1 +\log(n_{\max}^{(0)}) \big)}{n_1}  \Bigg\} \geq 1 - n_0^{-c_{3}}, \nonumber
 \end{align}
which proves \eqref{theorem1_2}  when $s_0 < n_0 - 1$.

\medskip
\noindent\textbf{Step 3.} We consider the term $(II)$ in   \eqref{theorem1_1}. Note that by applying the Cauchy-Schwartz inequality and utilising the fact that $\vert ab \vert \leq 4a^2 + b^2/16$, we can establish  that 
 \begin{align}\label{theorem1_12}
   (II) \leq   &    \frac{4\| \widetilde{P}^{n_0, n_1} \delta\|_2^2 }{n_{0}}  + \frac{1}{16n_{0}} \big\| \widehat{f} - f   \big\|_2^2   \leq     \frac{ 8 \| \delta\|_2^2}{ n_1}  + \frac{1}{16n_{0}} \big\| \widehat{f} - f   \big\|_2^2, 
\end{align}
where the last inequality follows from $n_1 \geq n_0$ and \Cref{lemma-alignment-delta}.

\medskip
\noindent\textbf{Step 4.} Choosing  $\lambda$ as \eqref{theorem1_lambda}, and combining \eqref{theorem1_1}, \eqref{theorem1_2} and \eqref{theorem1_12}, we have with an absolute constant $C_4 >0$ that  
\begin{align}
 \P  \Bigg\{  \big\|  \widehat{f} - f  \big\|_{1/n_0}^2
    \leq &    C_4  \frac{n_{\max}^{(0)} /\widetilde{n}_{\min}^{(0)} (s_0+1) \big(1 +\log(n_{\max}^{(0)}) \big) +  \|\delta \|_2^2}{n_1}   \Bigg\} \geq 1 - n_0^{-c_{\mathcal{E}}}. \nonumber
 \end{align}
Under \Cref{ass-fusedlasso},  if $\lambda = C_{\lambda}\big((s_0+1)n_1\big)^{-1/2}$, then it holds with an absolute constant $C_5 >0$  that  
 \begin{align}
 \P  \Bigg\{  \big\|  \widehat{f} - f  \big\|_{1/n_0}^2
    \leq &    C_5  \frac{ (s_0+1) \big\{1 +\log \big(n_0 /( s_0 +1) \big) \big\} +  \|\delta \|_2^2}{n_1}   \Bigg\} \geq 1 - n_0^{-c_{\mathcal{E}}}, \nonumber
 \end{align}
 completing the proof.
\end{proof}

\subsection{Additional lemmas}\label{app_1_lemma}

\begin{lemma}\label{lemma-P}  
For any $n, m \in \mathbb{N}^{*}$, let the alignment operator $P^{n, m} \in \R^{n \times m}$ be defined in \eqref{def-P}. We have that 
\begin{align}\label{lemma-P-reuslt-1}
        \sqrt{  \lceil n/m\rceil  - 1} \leq  \sigma_{m}(  P^{n, m} )  \leq     \sigma_{1}(  P^{n, m} ) \leq  \sqrt{ \lceil n/m \rceil},
\end{align}
and if $n=m$, $P^{n, m} = I_n$.

For any $n, m \in \mathbb{N}^{*}$ with $n \leq m$, let the alignment operator $\widetilde{P}^{n, m} \in \R^{n \times m}$ be defined in \eqref{def-P-alt}. If $m = n $, then $\widetilde{P}^{n, m} = I_n$. If $m >n$ then we have that 
\begin{align}\label{lemma-P-reuslt-2}
   \big(  \lceil m/n\rceil  \big)^{-1/2} \leq  \sigma_{n} (  \widetilde{P}^{n, m} )  \leq     \sigma_{1} (  \widetilde{P}^{n, m} ) \leq  \big( \lceil m/n \rceil -1  \big)^{-1/2},
\end{align}
and $\widetilde{P}^{n, m} P^{m, n} = I_{n}$.
\end{lemma}

\begin{proof}
We first consider $P^{n, m} \in \R^{n \times m}$ defined in \eqref{def-P}.  Based on the definition of $P^{n, m}$ in \eqref{def-P}, for any $i, j \in [m]$,  we have that 
\begin{align}
  \big((P^{n, m})^{\top} P^{n, m} \big)_{i, j} = &  \sum_{l=1}^n   \mathbbm{1}_{ \{ \lceil(i-1)  n/ m \rceil +1 \leq l  \leq \lceil i n / m \rceil \}} \mathbbm{1}_{ \{ \lceil(j-1)  n/ m \rceil +1 \leq l  \leq \lceil j n / m \rceil \}  }  \nonumber\\
 = & \begin{cases}
       \lceil j n/m \rceil - \lceil (j-1) n/m \rceil, & \mbox{if } i=j, \\
       0,  &\mbox{otherwise}. \nonumber
   \end{cases}
\end{align}
and if $n=m$,  $P^{n, m} = I_n$.
Since for any $j \in [m]$,
\[ 
    \lceil (j-1)n/m\rceil + \lceil n/m\rceil - 1\leq   \lceil jn/m \rceil \leq \lceil (j-1)n/m\rceil + \lceil n/m\rceil,
\]
it holds with  $\lambda_{1} \big((P^{n, m})^{\top} P^{n, m} \big)\geq \cdots \geq \lambda_{m} \big((P^{n, m})^{\top} P^{n, m} \big))$ being the eigenvalues of $(P^{n, m})^{\top} P^{n, m}$ that 
\[
    \lceil n/m\rceil - 1 \leq     \lambda_{m} \big((P^{n, m})^{\top} P^{n, m} \big) \leq  \lambda_{1} \big(P^{n, m})^{\top} P^{n, m} \big) \leq  \lceil n/m\rceil,
\]
which proves \eqref{lemma-P-reuslt-1}.

Next, we consider  $\widetilde{P}^{n, m} \in \R^{n \times m}$ defined in \eqref{def-P-alt} with $m \geq  n$. By the definition of $\widetilde{P}^{n, m} $ in \eqref{def-P-alt}, we have that for any $i, j \in [n]$,
\begin{align}
    \big( \widetilde{P}^{n, m}  (\widetilde{P}^{n, m} )^{\top}\big)_{i, j}  
    = & \sum_{l=1}^m \frac{\mathbbm{1}_{ \{ \lceil(i-1)m/n  \rceil  + 1 \leq l \leq \lceil im/n \rceil  \} }}{\lceil im/n \rceil - \lceil(i-1)m/n\rceil}  \frac{\mathbbm{1}_{ \{ \lceil(j-1)m/n  \rceil  + 1 \leq l \leq \lceil jm/n \rceil  \} }}{\lceil jm/n \rceil - \lceil(j-1)m/n\rceil} \nonumber \\
    = & \begin{cases}
       \big( \lceil jm/n \rceil - \lceil (j-1)m/n\rceil \big)^{-1}, & \mbox{if } i=j, \\
       0,  &\mbox{otherwise},\nonumber
   \end{cases}
\end{align}
and if $n=m$, $\widetilde{P}^{n, m} = I_n$.
Since
\[ 
    \lceil (j-1)m/n\rceil + \lceil m/n\rceil - 1\leq   \lceil jm/n \rceil \leq \lceil (j-1)m/n\rceil + \lceil m/n\rceil,
\]
for any $m >n $, it holds with $\lambda_{1} \big( \widetilde{P}^{n, m}  (\widetilde{P}^{n, m})^{\top}\big) \geq \cdots \geq \lambda_{n} \big( \widetilde{P}^{n, m}  (\widetilde{P}^{n, m} )^{\top}\big)$ being the eigenvalues of $ \widetilde{P}^{n, m}  (\widetilde{P}^{n, m} )^{\top}$  that 
\[
    \big(\lceil m/n\rceil  \big)^{-1} \leq     \lambda_{n}\big( P^{n, m} (P^{n, m})^{\top}\big) \leq  \lambda_{1} \big( P^{n, m} (P^{n, m})^{\top}\big) \leq  \big( \lceil m/n \rceil -1 \big)^{-1},
\]
which proves \eqref{lemma-P-reuslt-2}.
Additionally, by the definition of $\widetilde{P}^{n, m}$ and $P^{m, n}$ in \eqref{def-P-alt} and \eqref{def-P}, respectively, under the assumption $m>n$, it holds that for any $i, j \in [n]$ 
\begin{align}
  \big(\widetilde{P}^{n, m} P^{m, n} \big)_{i, j} 
  & =  \sum_{l=1}^{m} \frac{\mathbbm{1}_{ \{ \lceil(i-1)m/n  \rceil  + 1 \leq l \leq \lceil im/n \rceil  \} }}{\lceil im/n \rceil - \lceil(i-1)m/n\rceil}
  \mathbbm{1}_{ \{ \lceil(j-1)  m/ n \rceil +1 \leq l  \leq \lceil j m / n \rceil \}} \nonumber \\
 = &  \begin{cases}
       1, & \mbox{if } i=j, \\
       0,  &\mbox{otherwise},
   \end{cases} \nonumber
\end{align}
which proves $\widetilde{P}^{n, m} P^{m, n} = I_{n}$, completing the proof.
\end{proof}

\begin{lemma}\label{lemma-alignment-error}
For any $n, m \in \mathbb{N}^{*}$ with $m \geq n$, let the alignment operator $\widetilde{P}^{n, m} \in \R^{n \times m}$ be defined in \eqref{def-P-alt}, and assume that $\{ \widetilde{\epsilon}_{i} \}_{i=1}^{m}$ are mutually independent mean-zero $C_{\sigma}$-sub-Gaussian variables. We have that 
$\{ (\widetilde{P}^{n, m} \widetilde{\epsilon})_j \}_{j=1}^n$ are mutually independent mean-zero  $C_{\sigma}(2 n/ m)^{1/2}$-sub-Gaussian variables.
\end{lemma}
\begin{proof}
Note that by the definition of the alignment operator $\widetilde{P}^{n, m}$ in \eqref{def-P-alt} with $m \geq n$, we have that for any $j \in [n]$,  
\begin{align}
 ( \widetilde{P}^{n, m} \widetilde{\epsilon} )_j = \frac{1}{ \lceil j (m/n) \rceil - \lceil (j-1) m/n\rceil} \sum_{i = \lceil (j-1) (m/n)\rceil+1}^{ \lceil j m/n \rceil}\widetilde{\epsilon}_i. \nonumber
\end{align}
Since  for each $j \in [n]$, there is  no overlapping of independent random variables $\{\widetilde{\epsilon}_i\}_{i=1}^{m}$, we have that $\{ ( \widetilde{P}^{n, m} \widetilde{\epsilon} )_j\}_{j=1}^{n}$ are independent.
According to Proposition 2.6.1 in \cite{vershynin2018high}, we derive that for any $j \in [n]$,
\begin{align}\label{lemma-error-1}
 (\widetilde{P}^{n, m} \widetilde{\epsilon})_j \sim
   \mbox{mean-zero } C_{\sigma} \big( \lceil  j(m/n) \rceil - \lceil (j-1) (m/n)\rceil \big)^{-1/2}\mbox{-sub-Gaussian}. 
\end{align}
If $m = n$, then it holds that 
\begin{align}\label{lemma-error-2}
 \big( \lceil  j(m/n) \rceil - \lceil (j-1) (m/n)\rceil \big)^{-1/2} = \big( \lceil  j \rceil - \lceil (j-1) \rceil \big)^{-1/2} = 1 < (2n/m)^{1/2}.
\end{align}
Since for $m > n$
\begin{align}
   ( \lceil j (m/n) \rceil - \lceil (j-1) (m/n)\rceil )^{-1/2} \leq   (\lceil m/n \rceil - 1)^{-1/2},  \nonumber
\end{align}  
if $n < m < 2 n$, then it holds that 
\begin{align}\label{lemma-error-3}
( \lceil j (m/n) \rceil - \lceil (j-1) (m/n)\rceil \big)^{-1/2} \leq 1 < (2 n/ m)^{1/2}, 
\end{align}
and if $ m \geq 2 n$, then it holds that 
\begin{align}\label{lemma-error-4}
    ( \lceil j (m/n) \rceil - \lceil (j-1) (m/n)\rceil )^{-1/2} \leq  \{ n/(m - n) \}^{1/2} \leq  (2 n/ m )^{1/2}.
\end{align}
Combining \eqref{lemma-error-1}, \eqref{lemma-error-2} \eqref{lemma-error-3} and \eqref{lemma-error-4}  we can conclude that $\{ (\widetilde{P}^{n, m} \widetilde{\epsilon})_j \}_{j=1}^n$ are mutually independent mean-zero  $C_{\sigma}(2 n/ m)^{1/2}$-sub-Gaussian variables, which completes the proof.
\end{proof}

\begin{lemma}\label{lemma-alignment-delta}
For any $n, m \in \mathbb{N}^{*}$ with $m \geq n$, let the alignment operator $\widetilde{P}^{n, m} \in \R^{n \times m}$ be defined in \eqref{def-P-alt}.  For any vector $v \in \R^{m}$, it holds that 
\[
  \| \widetilde{P}^{n, m} v\|_2^2  \leq    \frac{2 \| v\|_2^2}{ m/n}.
\]
\end{lemma}
\begin{proof}
Note that if $m =  n$, by \Cref{lemma-P},  we have that 
 \begin{align}\label{lemma-delta-1}
 \| \widetilde{P}^{n, m} v \|_2^2 = \| v \|_2^2 \leq   \frac{2\| v \|_2^2}{m/n};
\end{align}
if $n < m < 2n$, by \Cref{lemma-P}  we have that 
\begin{align}\label{lemma-delta-2}  
  \| \widetilde{P}^{n, m} v \|_2^2  \leq     \frac{  \| v \|_2^2}{   \lceil m/n \rceil -1 } =    \| v\|_2^2 < \frac{ 2 \|v\|_2^2}{ m/n};
\end{align}
and if $m \geq  2n$, by \Cref{lemma-P} we have that 
\begin{align}\label{lemma-delta-3}   
 \| \widetilde{P}^{n, m} v \|_2^2  \leq     \frac{ \| v \|_2^2}{  \lceil m/n \rceil -1 } 
     \leq    \frac{  \| v \|_2^2}{ (m -  n)/n} \leq  \frac{ 2 \| v \|_2^2}{ m/n}.
\end{align}
Combining \eqref{lemma-delta-1}, \eqref{lemma-delta-2} and \eqref{lemma-delta-3}, if $n \leq m$, it holds that 
\[
  \| \widetilde{P}^{n, m} v\|_2^2  \leq    \frac{2 \| v\|_2^2}{ m/n},
\]
which concludes the proof.
\end{proof}

\begin{theorem}[\citealp{vandegeer2020logistic}]\label{theorm-error-fusedlasso}
 Let  $\mathcal{S} \subseteq  [n_0-1]$ be  defined in \eqref{def-S} with cardinality $s_0$.
Assume that $s_0 < n_0 -1$, then let $n_{\max}^{(0)}$,  the matrix  $\Psi^{-\mathcal{S}} \in \R^{n_0\times (n_0-1-s_0)}$ and the vector $w^{-\mathcal{S}} \in \R^{n_0 - 1 -s_0}$ be defined in  \Cref{ass-fusedlasso}, \eqref{def-psi} and \Cref{def-effective sparsity}, respectively. Furthermore, assume that $\{ \widetilde{\epsilon}_{i} \}_{i=1}^{n_0}$ are independent mean-zero $\sigma$-sub-Gaussian variables. If there exists $\lambda$ satisfies 
 \begin{align}
 \lambda \geq   \lambda_{\mathcal{S}} \quad \mbox{with} \quad 
 \lambda_{\mathcal{S}} = C_{\lambda_{\mathcal{S}}}  \sigma \max_{j \in [n_{0} - 1- s_{0}]} \big\|\Psi^{ -\mathcal{S}}_{,j} \big\|_{1/n_{0}}    n_0^{-1/2},   \nonumber
\end{align}
with the $j$th column of  $\Psi^{ - \mathcal{S}}$  denoted as  $\Psi^{ - \mathcal{S}}_{, j}$ and an absolute constant $ C_{\lambda_{\mathcal{S}}}  >0$, then it holds with probability at least $1 - \exp \big\{-c_{\epsilon} (s_0+1) \log\big(n_0/(s_0+1)\big) \big\}$ that 
\begin{align}
   \frac{\widetilde{\epsilon}^{\top} \theta}{n_0} 
    \leq &  \frac{1}{4n_0} \| \theta\|_2^2 +  C_{\epsilon} \sigma^2 \frac{(s_0+1) \log ( n_{\max}^{(0}))}{n_0} + \lambda 
 \sum_{i=1}^{n_0-1-s_0} \big\vert w^{-\mathcal{S}}_i \big( ( D\theta)_{-\mathcal{S}} \big)_{i}  \big\vert,  \nonumber
 \end{align}
 for any $\theta \in R^{n_0}$.
\end{theorem}

\begin{lemma}[\citealp{ortelli2019prediction}]\label{lemma-Psi} 
 Let  $\mathcal{S} \subseteq  [n_0-1]$ be  defined in \eqref{def-S} with cardinality $s_0$. Assume that $s_0 < n_0 -1$, then let $\Psi^{-\mathcal{S}}$ and $n_{\max}^{(0)}$ be defined in \eqref{def-psi} and \Cref{ass-fusedlasso}, respectively. Denote the $j$th column of  $\Psi^{ - \mathcal{S}}$ as  $\Psi^{ - \mathcal{S}}_{, j}$, then it holds that
    \[
      \max_{j \in [n_0-1-s]} \big\|\Psi^{-\mathcal{S}}_{, j} \big\|_{1/n_0}^2  \leq \frac{n_{\max}^{(0)}}{2n_0}.
    \]
\end{lemma}

\begin{lemma}[\citealp{vandegeer2020logistic}]\label{lemma-Gamma}
 Let  $\mathcal{S} \subseteq  [n_0-1]$ be  defined in \eqref{def-S} with cardinality $s_0$ and effective sparsity $\Gamma_{\mathcal{S}}^2$ be defined in \Cref{def-effective sparsity}. For $s_0 > 0$, let $\{n_i^{(0)}\}_{i=1}^{s_0+1}$, $n_{\max}^{(0)}$ and $\mathcal{S}_{\pm}$ be defined in \Cref{ass-fusedlasso} and \Cref{def-S_pm}.
It holds that 
 \[
    \Gamma_{\mathcal{S}}^2 \leq
    \begin{cases}
        C_{\Gamma} \log(n_0)   \quad &\mbox{if } s_0 = 0,\\
        C_{\Gamma} n_0  \Big( \sum_{i \in \mathcal{S}_{\pm}} \frac{1+\log (n^{(0)}_i)}{n^{(0)}_i } +  \sum_{i \in \mathcal{S} \setminus \mathcal{S}_{\pm}} \frac{1+\log (n_i^{(0)})}{n_{\max}^{(0)}} \Big), \quad &\mbox{otherwise},
    \end{cases}
 \]
where $C_{\Gamma}>0$ is an absolute constant .
\end{lemma}

\section[]{Proof of Theorem \ref{theorem_l_0-1}}\label{app_2}
The proof of \Cref{theorem_l_0-1} can be found in \Cref{app_2_proof} with the necessary auxiliary results in \Cref{app_2_lemma}.

\subsection{Proof of Theorem \ref{theorem_l_0-1}}\label{app_2_proof}
\begin{proof}[Proof of \Cref{theorem_l_0-1}]
This proof consists of four steps. In \textbf{Step 1}, we decompose our target quantity into several terms. We then deal with these terms individually in \textbf{Step 2} and \textbf{Step 3}. In \textbf{Step 4}, we gather all the pieces and conclude the proof.

\medskip
\noindent\textbf{Step 1.}
    It directly follows from the definition of  $\widetilde{f}$ that
\begin{align}
&  \frac{1}{2n_0} \big\|\widetilde{P}^{n_0, n_1} y^{(1)} -  \widetilde{f}  \big\|_2^2 +\widetilde{\lambda} \|D \widetilde{f} \|_0    \leq \frac{1}{2n_0} \big\|\widetilde{P}^{n_0, n_1} y^{(1)}   -   f  \big\|_2^2 +\widetilde{\lambda} \|D  f \|_0.  \nonumber
\end{align}
Given that $y^{(1)} = f^{(1)} + \epsilon^{(1)}$ with $f^{(1)} = P^{n_1, n_0} f + \delta$, we derive that 
\begin{align}
  \frac{1}{2n_{0}} \big\|\widetilde{f} - \widetilde{P}^{n_0, n_1} P^{n_1, n_0} f  \big\|_2^2  
\leq  &    \frac{1}{2n_{0}} \big\|f - \widetilde{P}^{n_0, n_1} P^{n_1, n_0} f  \big\|_2^2  +   \frac{1}{n_0}\widetilde{\epsilon}^{\top}\big( \widetilde{f}  - f \big)   + \widetilde{\lambda} \|D   f \|_0 
\nonumber\\
& \hspace{0.5cm}-  \widetilde{\lambda}   \|D   \widetilde{f} \|_0 + \frac{1}{n_0}\big(   \widetilde{P}^{n_0, n_1} \delta\big)^{\top} \big( \widetilde{f}  - f \big), \nonumber 
\end{align}
with $\widetilde{\epsilon} =\widetilde{P}^{n_0, n_1} \epsilon^{(1)} \in \R^{n_0}$. By $n_1 \geq n_0$ and \Cref{lemma-P}, it holds that
\begin{align}\label{theorem2_1}
 \frac{1}{2n_{0}} \big\|\widetilde{f} - f  \big\|_2^2   
\leq  &   \frac{1}{n_0}\widetilde{\epsilon}^{\top} \big( \widetilde{f}  - f \big) 
  + \lambda \|D   f \|_0 - \lambda   \|D   \widetilde{f} \|_0
+ \frac{1}{n_0}\big(   \widetilde{P}^{n_0, n_1} \delta\big)^{\top} \big( \widetilde{f}  -   f\big) \nonumber\\
= &  (I.1) + (I.2) + (I. 3) + (II) =(I) + (II).
\end{align}

\medskip
\noindent\textbf{Step 2.} In this step, we consider the term $(I)$ in \eqref{theorem2_1}. Note that by \Cref{lemma-alignment-error} and the assumption $n_1 \geq  n_0$, we obtain that
\begin{align}\label{theorem2_3}
\{ \widetilde{\epsilon}_i\}_{i=1}^{n_0} \overset{\mbox{ind.}}{\sim} 
  \mbox{mean-zero }C_{\sigma}\big(2 n_0/ n_1 \big)^{1/2}\mbox{-sub-Gaussian}.
\end{align}
Let the set $\mathcal{S}$ be defined in \eqref{def-S}  with  cardinality $s_0$ and the set $\widetilde{\mathcal{S}}$ be defined as
\begin{align}\label{theorem2_4}
\widetilde{\mathcal{S}} = \big\{i \in [n_{0}-1]\colon \widetilde{f}_i \neq \widetilde{f}_{i+1}  \big\}  = \big\{i \in [n_{0}-1]\colon (D\widetilde{f})_i \neq 0  \big\}. 
\end{align}
Let the orthogonal projection operator $P^{\widetilde{\mathcal{S}} \cup \mathcal{S}}$ be defined in \Cref{lemm_l_0}, then we have that 
\begin{align}\label{theorem2_5}
   (I.1) = &  \frac{1}{n_{0}} \widetilde{\epsilon}^{\top} \big( P^{\widetilde{\mathcal{S}} \cup \mathcal{S}} ( \widetilde{f} - f  )  \big)
       =       \frac{1}{n_{0}} \big(P^{ \widetilde{\mathcal{S}} \cup \mathcal{S}} \widetilde{\epsilon} \big)^{\top} \big( \widetilde{f} - f  \big)   \nonumber\\
       \leq    &    \frac{1}{n_{0}}  \big\|P^{\widetilde{\mathcal{S}} \cup \mathcal{S}} \widetilde{\epsilon} \big\|_2
       \big\|\widetilde{f} - f  \big\|_2
         \leq          \frac{1}{n_{0}}  \big\|P^{\widetilde{\mathcal{S}} \cup \mathcal{S}}\widetilde{\epsilon} \big\|_2^2+  \frac{1}{4n_{0}} 
       \big\| \widetilde{f} - f    \big\|_2^2, 
\end{align}
where the first inequality follows from Cauchy--Schwartz inequality and the last inequality is based on the fact that $\vert ab \vert \leq a^2 + b^2/4$.
By \eqref{theorem2_3}, \eqref{theorem2_5} and \Cref{lemm_l_0}, we can conclude that  that $\mathbb{P}\{\mathcal{E}\} \geq 1- n_{0}^{-c_{\epsilon}}$ with 
\[
 \quad \mathcal{E} = \bigg\{  (I.1) \leq \frac{1}{4n_{0}} 
       \big\| \widetilde{f} - f    \big\|_2^2+ C_{\epsilon}  \frac{  \big( \vert \widetilde{\mathcal{S}} \cup \mathcal{S}\vert +1 \big) \big\{1 + \log\big( n_{0}/ (\vert\widetilde{\mathcal{S}} \cup \mathcal{S}\vert+1) \big) \big\} }{n_1 }\bigg\},
\]
where $C_{\epsilon}, c_{\epsilon} >0$ are absolute constants. From now on we assume that the event $\mathcal{E}$ holds. Then it holds that 
\begin{align}\label{theorem2_7}
   (I)
 \leq & \frac{1}{4n_{0}} 
       \big\| \widetilde{f} - f    \big\|_2^2 +  C_{\epsilon}  \frac{  \big( \vert \widetilde{\mathcal{S}} \cup \mathcal{S}\vert +1 \big) \big\{1 + \log\big( n_{0}/ (\vert\widetilde{\mathcal{S}} \cup \mathcal{S}\vert+1) \big) \big\} }{n_1 }+  \widetilde{\lambda} \|D   f \|_0 - \widetilde{\lambda}   \|D   \widetilde{f}\|_0\nonumber\\
= &  \frac{1}{4n_{0}} 
       \big\| \widetilde{f} - f    \big\|_2^2 +  C_{\epsilon}  \frac{  \big( \vert \widetilde{\mathcal{S}} \cup \mathcal{S}\vert +1 \big) \big\{1 + \log\big( n_{0}/ (\vert\widetilde{\mathcal{S}} \cup \mathcal{S}\vert+1) \big) \big\} }{n_1 }+  \widetilde{\lambda} \big(s_0 - \vert \widetilde{\mathcal{S}} \vert \big) \nonumber\\
\leq & \frac{1}{4n_{0}} 
       \big\| \widetilde{f} - f    \big\|_2^2  + 2 \widetilde{\lambda}   (s_0+1) \nonumber\\
= &\frac{1}{4n_{0}} 
       \big\| \widetilde{f} - f    \big\|_2^2  + 2 C_{\widetilde{\lambda}}\frac{( s_0+1) \big\{ 1+ \log\big(n_{0}/(s_0+1)\big) \big\} }{n_1 },
\end{align}
where 
\begin{itemize}
\item the first equality follows from the definitions of $\mathcal{S}$ and $\widetilde{\mathcal{S}}$ in \eqref{def-S} and \eqref{theorem2_4},
\item and the second inequality and the last equality are due to the choice of $\widetilde{\lambda}$ in \eqref{tuning-parameter-0} and $C_{\widetilde{\lambda}} >0$ is a large enough absolute constant.
\end{itemize}

\medskip
\noindent\textbf{Step 3.} In this step, we consider the term $(II)$ in  \eqref{theorem2_1}. Note that by applying the Cauchy--Schwartz inequality and utilising the fact that $\vert ab \vert \leq 2a^2 + b^2/8$, we can establish  that 
 \begin{align}\label{theorem2_8}
   (II) \leq     \frac{2\| \widetilde{P}^{n_0, n_1} \delta\|_2^2 }{n_{0}}  + \frac{1}{8n_{0}} \big\| \widetilde{f} - f   \big\|_2^2  \leq       \frac{ 4 \| \delta\|_2^2}{ n_1}  + \frac{1}{8n_{0}} \big\| \widetilde{f} - f   \big\|_2^2,
\end{align}
where the last inequality follows from $n_1 \geq n_0$ and \Cref{lemma-alignment-delta}.

\medskip
\noindent\textbf{Step 4.} Choosing  $\widetilde{\lambda}$ as \eqref{tuning-parameter-0}, and combining \eqref{theorem2_1}, \eqref{theorem2_7} and \eqref{theorem2_8}, we have with an absolute $C_1 > 0$ that  
\[
 \P  \Bigg\{  \big\|  \widetilde{f} - f  \big\|_{1/n_0}^2
    \leq    C_1\frac{  (s_{0}+1) \big\{1+\log \big(n_1/(s_0+1) \big) \big\}+\|\delta\|_2^2 }{n_1} \Bigg\} \geq 1 - n_0^{-c_{\mathcal{E}}},
\]
completing the proof.
\end{proof}

\subsection{Additional lemmas}\label{app_2_lemma}
\begin{lemma}\label{lemm_l_0}
For any $\mathcal{M} \subseteq [n_0-1]$, if $\vert \mathcal{M} \vert > 0$, denote it as $\mathcal{M} = \{ t_1^\mathcal{M}, \dots, t_{\vert \mathcal{M} \vert}^{\mathcal{M}} \}$. Let $t_0^{\mathcal{M}} =0$ and $t_{\vert \mathcal{M} \vert+1}^{\mathcal{M}} = n_0$. Let the subspace $\mathcal{K}^{\mathcal{M}} \subset \mathbb{R}^{n}$ be defined as $\theta \in \mathcal{K}^{\mathcal{M}}$ if and only if $\theta$ takes a constant value on $\{t_{i}^{\mathcal{M}}+1, \dots, t^{\mathcal{M}}_{i+1}\}$ for each $i \in  [0: \vert \mathcal{M} \vert]$. Then let  $P^{\mathcal{M}}$ be the orthogonal projection operator from $\R^{n_0}$ to $\mathcal{K}^{\mathcal{M}}$.
 Assume that $\{ \widetilde{\epsilon}_{i} \}_{i=1}^{n_0}$ are mutually independent mean-zero $\sigma$-sub-Gaussian variables. Then there exist absolute constants $C_{\epsilon}, c_{\epsilon} >0$ such that 
\begin{align}
\P \Big\{ \forall \mathcal{M}  \subseteq [n_0-1]\colon \big\| P^{\mathcal{M}}  \widetilde{\epsilon} \big\|_2^2 \leq C_{\epsilon}\sigma^2 ( \vert \mathcal{M} \vert  +1 ) \big\{ 1 \vee \log  \big( n_0/ ( \vert \mathcal{M} \vert +1) \big) \big\}\Big\}  \geq 1-n_0^{-c_{\epsilon}}.  \nonumber
\end{align}

\end{lemma}

\begin{proof}

Fix $\mathcal{M} \subseteq [n_0-1]$. For any $\gamma >0$, a random variable $Z$ is said to be $\gamma$-sub-exponential distributed if $\| Z\|_{\psi_1} = \inf \big\{t>0\colon \E\{\exp\{ \vert Z\vert/t\leq 2 \} \big\} \leq \gamma$.  By Proposition 2.6.1 and Lemma 2.7.6  in\cite{vershynin2018high}, we have with an absolute constant $c_0 >0$ that
\[
     \sum_{j=t_{i}+1}^{t_{i+1}} (P^{\mathcal{M}} \widetilde{\epsilon})_j^2  \overset{\mbox{ind.}}{\sim} c_0 \sigma^2\mbox{-sub-exponential}, \quad i \in [ 0 : \vert \mathcal{M} \vert].
\] 
Note that 
\begin{align}\label{l_0-lemma1-1}
   \mathbb{E} \big\{  \big\| P^{\mathcal{M}}  \widetilde{\epsilon} \big\|_2^2  \big\} 
   = &   \sum_{i = 0}^{\vert \mathcal{M} \vert}   \sum_{j=t_{i}+1}^{t_{i+1}} \mathbb{E} \big\{  (P^{\mathcal{M}} \widetilde{\epsilon})_j^2  \big\}
    =   \sum_{i = 0}^{\vert \mathcal{M} \vert}   \sum_{j=t_{i}+1}^{t_{i+1}} \mathbb{E} \bigg\{ \Big( \sum_{k = 1}^{n_0} P^{\mathcal{M}}_{j, k} \widetilde{\epsilon}_k \Big)^2 \bigg\} \nonumber\\
   = &  \sum_{i = 0}^{\vert \mathcal{M} \vert}   \sum_{j=t_{i}+1}^{t_{i+1}}  \sum_{k = 1}^{n_0}  \mathbb{E} \Big\{\Big(P^{\mathcal{M}}_{j, k} \widetilde{\epsilon}_k \Big)^2 \Big\} 
    \leq  C_0 \sigma^2 ( \vert \mathcal{M} \vert+ 1 ),
\end{align}
where $C_0 > 0$ is an absolute constant and third equality follows from that  $\{ \widetilde{\epsilon}_{i} \}_{i=1}^{n_0}$ are mutually independent.
Then by \eqref{l_0-lemma1-1} and Bernstein's inequality \citep[e.g. Theorem 2.8.1][]{vershynin2010introduction}, it holds with absolute constants $C_1, c_1 >0$ that 
\begin{align}
      & \P \Big[  \big\| P^{\mathcal{M}}  \widetilde{\epsilon} \big\|_2^2 >  C_1 \sigma^2 ( \vert \mathcal{M} \vert+ 1 ) \big\{1\vee \log \big( n_0/ (\vert \mathcal{M} \vert+1)\big) \big\} \Big] \nonumber\\
       \leq & \exp \Big[ -c_1 ( \vert \mathcal{M} \vert+1 ) \big\{ 1\vee \log \big( n_0/ (\vert \mathcal{M} \vert+1))\big)  \big\} \Big]. \nonumber
\end{align}
By a union bound argument, we derive that 
\begin{align}\label{l_0-lemma1-2}
& \P \Big[ \exists \mathcal{M}  \subseteq [n_0-1]\colon \big\| P^{\mathcal{M}}  \widetilde{\epsilon} \big\|_2^2 > C_{1}\sigma^2 ( \vert \mathcal{M} \vert +1 ) \big\{ 1 \vee \log \big( n_0/ ( \vert \mathcal{M} \vert +1) \big) \big\} \Big] \nonumber \\
\leq &  \sum_{ \mathcal{M} \subseteq [n_0 - 1] } \P \Big[ \big\| P^{\mathcal{M}}  \widetilde{\epsilon} \big\|_2^2 > C_{1}\sigma^2 ( \vert \mathcal{M} \vert +1 ) \big\{ 1 \vee \log \big( n_0/ (\vert \mathcal{M} \vert +1) \big) \big\} \Big] \nonumber \\
\leq & \sum_{m=0}^{n_0-1} \sum_{\substack{  \mathcal{M} \subseteq [n_0 - 1] \\ \mbox{with }\vert \mathcal{M}\vert = m} } \exp \Big[ -c_1 ( m+1 ) \big\{1\vee \log \big( n_0/ (m+1)\big)\big\} \Big]  \nonumber\\
\leq & \sum_{m = 0}^{n_0-1} \binom{n_0 - 1}{m} \exp \Big[ -c_1 ( m+1 ) \big\{1\vee \log \big( n_0/ (m+1) \big)\big\} \Big]  \nonumber  \\
\leq & n_0^{-c_1} +  \sum_{m = 1}^{n_0-1}  \exp \Big[ m  \log\big( e(n_0-1)/m\big)  -c_1 ( m+1 ) \big\{1\vee \log \big( n_0/ (m+1) \big)\big\} \Big]  \nonumber\\
\leq & n_0^{-c_1} +\sum_{m = 1}^{n_0-1}  \exp \Big[  -c_2 ( m+1 ) \big\{1\vee \log \big( n_0/ (m+1) \big)\big\} \Big],  
\end{align}
where $c_2>0$ is an absolute constant and the fourth inequality is based on the fact that for any $m_1 \in \mathbb{N}^{*}$ and $m_2 \in [m_1]$ 
\[
   \binom{m_1}{m_2} \leq \Big(\frac{em_1}{m_2}\Big)^{m_2}. 
\]
The function
\[
    m \mapsto  - c_2  (m+1) \log \big( n_0/ (m+1) \big) 
\]
is convex, so its maximum over $m \in  [n_0-1]$ is attained at either $m=1$ or $m=n_0-1$. Thus, we have with an absolute constant $c_3>0$ that
\begin{align}\label{l_0-lemma1-3}
 & \sum_{m = 1}^{n_0-1} \exp \big\{ -c_2 ( m+1 ) \big(1\vee \log ( n_0/ (m+1))\big) \big\}  \nonumber \\
\leq &  (n_0-1) \max \Big\{  \exp \big\{ -2c_2 \log ( n_0/2 ) \big\},  \exp \big\{  -c_2 n_0  \big\}\Big\} \leq n_0^{-c_{3}}.
\end{align}
Combining \eqref{l_0-lemma1-2} and \eqref{l_0-lemma1-3}, it holds with an absolute constant $c_4 >0$ that
\begin{align}
& \P \Big[ \forall \mathcal{M}  \subseteq [n_0-1]\colon \big\| P^{\mathcal{M}}  \widetilde{\epsilon} \big\|_2^2 \leq C_{2}\sigma^2 ( \vert \mathcal{M} \vert +1 ) \big\{ 1 \vee \log \big( n_0/ ( \vert \mathcal{M} \vert +1) \big) \big\} \Big]  \geq 1- n_0^{-c_4}, \nonumber
\end{align}
completing the proof.
\end{proof}

\section[]{Technical details of results in \Cref{sec-multiple-source}}\label{app_3}

The proofs of \Cref{prop-ora}, \Cref{theorem_detection_consistency},  \Cref{cor-a-hat-l0l1} and \Cref{theorem_minimax} can be found in Appendices \ref{app_3_proposition1}, \ref{app_3_det}, \ref{app_3_coro}  and \ref{app_3_minimax}, respectively. 

\subsection{Proof of Proposition \ref{prop-ora}}\label{app_3_proposition1}  

The proof of \Cref{app_3_proposition1}  can be found in \Cref{app_2_proof} .

 \begin{proof}[Proof of \Cref{{prop-ora}}]
This proof consists of two steps. In \textbf{Step 1}, we focus on establishing \eqref{upper_bound_l_1_ora}, and in \textbf{Step 2}, we provide the proof of \eqref{upper_bound_l_0_ora}.

\medskip
\noindent\textbf{Step 1.}
Our proof in this step consists of four sub-steps. In \textbf{Step 1.1}, we decompose our target quantity into several terms. 
We deal with these terms individually in \textbf{Step 1.2} and \textbf{Step 1.3}. Finally, in \textbf{Step 1.4}, we gather all the pieces and conclude the proof of \eqref{upper_bound_l_1_ora}.

\medskip
\noindent\textbf{Step 1.1.}
    It directly follows from the definition of  $\widehat{f}^{[K]}$ that
\begin{align}
  &   \frac{1}{2n_{0}} \bigg\|\frac{1}{K} \sum_{k=1}^K \widetilde{P}^{n_0, n_k} y^{(k)}  - \widehat{f}^{[K]} \bigg\|_2^2   +\lambda \|D \widehat{f}^{[K]}\|_1
  \leq  \frac{1}{2n_{0}} \bigg\|\frac{1}{K} \sum_{k=1}^K \widetilde{P}^{n_0, n_k} y^{(k)}  - f \bigg\|_2^2   +\lambda \|D f  \|_1. \nonumber
\end{align}
Since for any $k \in [K]$, $y^{(k)} = f^{(k)} + \epsilon^{(k)}$ with $f^{(k)} = P^{n_k, n_0} f + \delta^{(k)}$, we obtain that 
\begin{align}
\frac{1}{2n_{0}} \bigg\|\widehat{f}^{[K]} - \frac{1}{K} \sum_{k=1}^K \widetilde{P}^{n_0, n_k} P^{n_k, n_0} f  \bigg\|_2^2   
\leq  &  \frac{1}{2n_{0}} \bigg\|f - \frac{1}{K} \sum_{k=1}^K \widetilde{P}^{n_0, n_k} P^{n_k, n_0} f  \bigg\|_2^2  +  \frac{1}{n_0}\big( \epsilon^{K}\big)^{\top} \big( \widehat{f}^{[K]}  -  f\big)   \nonumber\\
 & \hspace{0.5cm} + \lambda  \|D   f \|_1 - \lambda   \|D  \widehat{f}^{[K]} \|_1
+ \frac{1}{n_0}\big( \delta^{K}  \big)^{\top} \big( \widehat{f}^{[K]}  -  f \big), \nonumber 
\end{align}
with $\epsilon^{K} = K^{-1} \sum_{k=1}^K \widetilde{P}^{n_0, n_k} \epsilon^{(k)}$ and $\delta^{K} = K^{-1} \sum_{k=1}^K \widetilde{P}^{n_0, n_k} \delta^{(k)}$. By $\min_{k \in [K]} n_k \geq n_0$ and \Cref{lemma-P}, it holds that
\begin{align}\label{prop1_1}
  \frac{1}{2n_{0}} \big\|\widehat{f}^{[K]} -  f  \big\|_2^2 
\leq  &   \frac{1}{n_0}\big( \epsilon^{K}\big)^{\top} \big( \widehat{f}^{[K]}  -  f\big)  
  + \lambda  \|D   f \|_1 - \lambda   \|D\widehat{f}^{[K]}\|_1
+ \frac{1}{n_0}\big(  \delta^{K}\big)^{\top} \big( \widehat{f}^{[K]}  -  f\big) \nonumber \\
= & (I.1)+(I.2)+(I.3)+(II) = (I) + (II).
\end{align}

\medskip
\noindent\textbf{Step 1.2.}
In this sub-step, we deal with the term $(I)$ in \eqref{prop1_1}. 
We claim that if 
\begin{align}\label{prop1_lambda}
\lambda =  C_{\lambda}  K^{-1} 
 \Big(  n_{\max}^{(0)}/ n_0 \sum_{k = 1}^K n_k^{-1} \Big)^{1/2}
\end{align}
with $n_{\max}^{(0)}$  defined in \Cref{ass-fusedlasso}, $\widetilde{n}_{\min}^{(0)}$ defined in \Cref{remark_l_1_one} and an absolute constant $C_{\lambda} >0$,  then it holds that 
\begin{align}\label{prop1_2}
 \P  \Bigg\{  (I)
    \leq &   \frac{3}{8n_0} \big\|  \widehat{f}^{[K]} - f  \big\|_2^2  + C_{\mathcal{E}} \frac{n_{\max}^{(0)} /\widetilde{n}_{\min}^{(0)} (s_0+1) \big(1 +\log(n_{\max}^{(0)}) \big)}{K^2  \big( \sum_{k =1}^K  n_k^{-1} \big)^{-1} }  \Bigg\} \geq 1 - n_0^{-c_{\mathcal{E}}},
 \end{align}
with absolute constants $C_{\mathcal{E}}, c_{\mathcal{E}} >0$. Then we prove this claim in two scenarios $s_0 = n_0 - 1$ and $s_0 < n_0 - 1$ in \textbf{Step 1.2.1} and \textbf{Step 1.2.2.}, respectively. Before proving the claim,
note that by $\min_{k \in [K]} n_k \geq n_0$ and \Cref{lemma-alignment-error}, we obtain that for any  $k \in [K]$
\begin{align}
\{(\widetilde{P}^{n_0, n_k} \epsilon^{(k)})_i\}_{i=1}^{n_0} \overset{\mbox{ind.}}{\sim} 
  \mbox{mean-zero }C_{\sigma}  
( 2n_0 n_k^{-1})^{1/2}\mbox{-sub-Gaussian}. \nonumber
\end{align}
By $\epsilon^{K} = K^{-1} \sum_{k=1}^K \widetilde{P}^{n_0, n_k} \epsilon^{(k)}$ and  Proposition 2.6.1 in \cite{vershynin2018high}, it holds that 
\begin{align}\label{prop1_3}
\{\epsilon^{K}_i\}_{i=1}^{n_0} \overset{\mbox{ind.}}{\sim} 
  \mbox{mean-zero }C_{\sigma} K^{-1} 
 \Big( 2n_0\sum_{k =1}^K n_k^{-1} \Big)^{1/2}\mbox{-sub-Gaussian}.
\end{align}

\medskip
\noindent\textbf{Step 1.2.1.} In this sub-step, we prove the claim stated in \eqref{prop1_2} in the scenario $s_0 = n_0 - 1$.
By \eqref{prop1_3} and general Hoeffding inequality \citep[e.g. Theorem 2.6.3 in][]{vershynin2018high}, it holds  with an absolute constant  $c_0 >0$ that
\begin{align}\label{prop1_4}
\P \big\{ \mathcal{E}_1 \big\} \geq 1 - \exp\{ -c_0 n_0 \} \quad \mbox{with} \quad 
\mathcal{E}_1 =  \bigg\{ (I.1) \leq K^{-1}\big( \sum_{k =1}^K n_k^{-1} \big)^{1/2} \big\|  \widehat{f}^{[K]} - f  \big\|_2  \bigg\}.
\end{align}
From now on we assume that $\mathcal{E}_1 $ holds in this sub-step.
By applying Cauchy--Schwartz inequality and the fact that $\vert ab \vert \leq a^2+ b^2/4$, we derive that 
\begin{align}\label{prop1_5}
(I) \leq  \frac{1}{4n_0}\big\|  \widehat{f}^{[K]} - f  \big\|_2^2  + \frac{n_0}{ K^2 \big( \sum_{k =1}^K n_k^{-1} \big)^{-1} } +  \lambda \|D   f \|_1 - \lambda   \|D   \widehat{f}^{[K]} \|_1.
 \end{align}
Note that with the sign vector $q \in \R^{s_0}$ defined in \eqref{def-q}, we obtain that 
\begin{align}\label{prop1_6}
\|Df  \|_1  = \big\|(D   f)_{\mathcal{S}} \big\|_1  = q^{\top}(D   f)_{\mathcal{S}}  \quad \mbox{and}  \quad \|D \widehat{f}^{[K]} \|_1  = \big\|(D \widehat{f}^{[K]})_{\mathcal{S}}  \big\|_1  \geq q^{\top}(D   \widehat{f}^{[K]})_{\mathcal{S}}.
\end{align}
Combining \eqref{prop1_5} and  \eqref{prop1_6}, we have that 
 \begin{align}\label{prop1_7}
    (I) \leq  & \frac{1}{4n_0}\big\|  \widehat{f}^{[K]} - f  \big\|_2^2  + \frac{n_0}{ K^2 \big( \sum_{k =1}^K n_k^{-1} \big)^{-1} } + \lambda  q^{\top}\big(D  (f- \widehat{f}^{[K]}) \big)_{\mathcal{S}} \nonumber\\
    \leq &  \frac{1}{4n_0}\big\|  \widehat{f}^{[K]} - f  \big\|_2^2  + \frac{n_0}{ K^2 \big( \sum_{k =1}^K n_k^{-1} \big)^{-1} }  + \lambda \Gamma_{\mathcal{S}} \big\|  \widehat{f}^{[K]} - f  \big\|_{1/n_0}  \nonumber\\
    \leq & \frac{3}{8n_0}\big\|  \widehat{f}^{[K]} - f  \big\|_2^2  + \frac{n_0}{ K^2 \big( \sum_{k =1}^K n_k^{-1} \big)^{-1} } + 2 \lambda^2 \Gamma_{\mathcal{S}}^2 \nonumber\\
    \leq & \frac{3}{8n_0}\big\|  \widehat{f}^{[K]} - f  \big\|_2^2 +\big(C_{\lambda}^2 C_{\Gamma} +1\big) \frac{(s_0+1)} {K^2 \big( \sum_{k =1}^K n_k^{-1} \big)^{-1} }, 
 \end{align} 
 where \begin{itemize}
     \item  the second inequality follows from the definition of the effective sparsity $\Gamma_{\mathcal{S}}$ in \Cref{def-effective sparsity},
     \item  the third inequality is based on the fact $\vert ab \vert \leq 2a^2+ b^2/8$,
     \item and the final inequality follows from the choice of $\lambda$ in \eqref{prop1_lambda}, \Cref{lemma-Gamma}, and for any $i \in [s_0+1]$, $n_{i}^{(0)} = 1$ and $n_{\max}^{(0)} = 1$.
 \end{itemize}
Since when $s_0 = n_0 - 1$, we have $\widetilde{n}_{\min}^{(0)} = n_{\max}^{(0)} = 1$, then combining \eqref{prop1_4} and \eqref{prop1_7}, it holds  with an absolute constant $c_1 >0$ that
\begin{align}
     \P  \Bigg\{  (I)
    \leq &   \frac{3}{8n_0} \big\|  \widehat{f}^{[K]} - f  \big\|_2^2  + \big(C_{\lambda}^2 C_{\Gamma} +1\big) \frac{n_{\max}^{(0)} /\widetilde{n}_{\min}^{(0)} (s_0+1) \big(1 +\log(n_{\max}^{(0)}) \big)}{K^2 \big( \sum_{k =1}^K n_k^{-1} \big)^{-1}}  \Bigg\} \geq 1 - n_0^{-c_{1}}, \nonumber
\end{align}
which proves \eqref{prop1_2} in the scenario $s_0 = n_0 -1$.

\medskip
\noindent\textbf{Step 1.2.2.} In this sub-step, we  prove the claim stated in \eqref{prop1_2} in the scenario $s_0 < n_0 - 1$.

By \eqref{prop1_3}  and \Cref{theorm-error-fusedlasso}, we have that    $\P \{ \mathcal{E}_2 \} \geq 1 - \exp \big\{-c_2 (s_0+1) \log\big(n_0/(s_0+1) \big) \big\}$  with 
\begin{align}
  \mathcal{E}_2 =    \Bigg\{  (I.1)
    \leq & \frac{1}{4n_0} \big\|  \widehat{f}^{[K]} - f  \big\|_2^2 + C_2  \frac{(s_0+1) \log (n_{\max}^{(0)})}{K^2 \big( \sum_{k =1}^K n_k^{-1} \big)^{-1}}  + \lambda 
 \sum_{i=1}^{n_0-1-s_0} \big\vert w^{-\mathcal{S}}_i \big( ( D\widehat{f}^{[K]} - D f)_{-\mathcal{S}} \big)_{i}  \big\vert \Bigg\}, \nonumber
 \end{align}
where $ C_2, c_2 > 0 $ are absolute constants, if  $\lambda$ satisfies 
 \begin{align}\label{prop1_lambda_check}
 \lambda \geq   \lambda_{\mathcal{S}} \quad \mbox{with} \quad 
 \lambda_{\mathcal{S}} = C_{\lambda_0}  \frac{\max_{j \in [n_{0} - 1- s_{0}]} \big\|\Psi^{ -\mathcal{S}}_{,j} \big\|_{n_{0}}}{    K \big( \sum_{k =1}^K n_k^{-1} \big)^{-1/2}}  
  \leq  C_{\lambda_0} K^{-1} \Big( n_{\max}^{(0)} /(2n_0) \sum_{k =1}^K n_k^{-1} \Big)^{1/2},
\end{align}
where the last inequality follows from \Cref{lemma-Psi} and $ C_{\lambda_0}  >0$ is an absolute constant.
Note that the function
\[
    s \mapsto  -c_2 (s +1) \log ( n_{0}/ (s + 1)))
\]
is convex, so its maximum over $\{ 0\} \cup [n_{0}-2]$ is attained at either $s =0$ or $s = n_{0} - 2$.  Thus, it holds with an absolute constant $c_3>0$ that
\begin{align}\label{prop1_8}
\P \{ \mathcal{E}_2 \} \geq 1 - \max \big\{  \exp \{ - c_{2}\log ( n_{0} ) \},  \exp \{  -c_{2} (n_{0} -1)\log ( n_{0}/(n_{0}-1) ) \} \big\} \geq 1- n_{0}^{-c_3}.  
\end{align}

 From now on we assume that the event $\mathcal{E}_2$ holds in this sub-step.   Note that 
\begin{align}\label{prop1_9}
&  \lambda \sum_{i=1}^{n_0-1-s_0} \big\vert w^{-\mathcal{S}}_i \big(( D   \widehat{f}^{[K]} - D f  )_{-\mathcal{S}} \big)_{i}  \big\vert   + \lambda\|D   f \|_1 - \lambda \|D   \widehat{f}^{[K]} \|_1   \nonumber \\
= &  - \lambda \sum_{i=1}^{n_0-1-s_0} \big( 1- w^{-\mathcal{S}}_i \big) \big\vert \big(( D   \widehat{f}^{[K]} - D f  )_{-\mathcal{S}} \big)_{i}  \big\vert   + \lambda  \big\|( D   \widehat{f}^{[K]} - D f  )_{-\mathcal{S}} \big\|  + \lambda \|D   f \|_1 - \lambda \|D   \widehat{f}^{[K]} \|_1   \nonumber \\
\leq&    - \lambda \sum_{i=1}^{n_0-1-s_0} \big( 1- w^{-\mathcal{S}}_i \big) \big\vert \big(( D   \widehat{f}^{[K]} - D f  )_{-\mathcal{S}} \big)_{i}  \big\vert   + \lambda \big\|(D   f)_{\mathcal{S}} \|_1 - \lambda \big\| (D   \widehat{f}^{[K]})_{\mathcal{S}}\|_1 + 2\lambda  \big\| ( D f )_{-\mathcal{S}} \big\|_1   \nonumber\\
\leq&   - \lambda \sum_{i=1}^{n_0-1-s_0} \big( 1- w^{-\mathcal{S}}_i \big) \big\vert \big(( D   \widehat{f}^{[K]} - D f  )_{-\mathcal{S}} \big)_{i}  \big\vert  + \lambda q^{\top} \big(D (f - \widehat{f}^{[K]})_{\mathcal{S}} \big) +  2\lambda  \big\| ( D f )_{-\mathcal{S}} \big\|_1   \nonumber\\
\leq&  \lambda \Gamma_{\mathcal{S}}\big\| \widehat{f}^{[K]} - f  \big\|_{1/n_0}+  2\lambda  \big\| ( D f )_{-\mathcal{S}} \big\|_1 
\leq \frac{1}{8n_{0}}\big\| \widehat{f}^{[K]} - f  \big\|_2^2 +    2\lambda^2 \Gamma_{\mathcal{S}}^2 +  2\lambda  \big\| ( D f )_{-\mathcal{S}} \big\|_1,   
 \end{align}
where
\begin{itemize}
    \item the first equality is based on the fact that $w^{-\mathcal{S}} \in [0, 1]^{n_0-1-s_0}$ defined in \Cref{def-effective sparsity}, 
    \item  the first inequality follows from the reverse triangle inequality,
    \item  the second inequality is based on the fact that $ \big\|(D   f)_{\mathcal{S}} \big\|_1  = q^{\top}(D   f)_{\mathcal{S}}$ and  $  \big\|(D \widehat{f}^{[K]})_{\mathcal{S}}  \big\|_1  \geq q^{\top}(D   \widehat{f}^{[K]})_{\mathcal{S}}$. Specifically, the sign vector $q$ is defined in \eqref{def-q} when $s_0 >0$ and is set to  $q = 0$ for $s_0 = 0$ 
    \item  the third inequality follows from the definition of the effective sparsity $\Gamma_{\mathcal{S}}$ in \Cref{def-effective sparsity},
     \item  the final inequality is based on the fact $\vert ab \vert \leq 2a^2+ b^2/8$.
\end{itemize}
By the construction of $\mathcal{S}$ in \eqref{def-S}, it holds that   
\begin{align}\label{prop1_10}
       \| ( D f )_{-\mathcal{S}} \|_1   = 0.
\end{align}
By \Cref{lemma-Gamma}, we obtain a deterministic result with an absolute  constant $C_{\Gamma}>0$ as follows
\begin{align}
   \Gamma_{\mathcal{S}}^2 \leq  \begin{cases}
   C_{\Gamma}  \log(n_{0}), \quad &\mbox{if } s_0 = 0, \\
        C_{\Gamma}  n_{0} \bigg( \sum_{i \in \mathcal{S}_{\pm}}\frac{1+\log ( n^{(0)}_i )}{n^{(0)}_i } +  \sum_{i \in  S \setminus  \mathcal{S}_{\pm}} \frac{1+\log (n^{(0))}_i)}{n^{(0)}_{\max}  } \bigg)  , \quad & \mbox{otherwise}.
           \end{cases} \nonumber
 \end{align}
 Then it holds with $n_{\max}^{(0)}$ and $\widetilde{n}_{\min}^{(0)}$ defined in \Cref{ass-fusedlasso} and \Cref{remark_l_1_one}, respectively,  that 
 \begin{align}\label{prop1_11}
       \Gamma_{\mathcal{S}}^2 \leq 
       C_{\Gamma} \frac{n_{0} (s_{0}+1) \big(1 +\log(n_{\max}^{(0)}) \big)}{ \widetilde{n}_{\min}^{(0)} }.
 \end{align}
Then combining \eqref{prop1_8}, \eqref{prop1_9}, \eqref{prop1_10} and \eqref{prop1_11}, with $n_{\max}^{(0)} \geq n_0/(s_0+1)$ and the choice of $\lambda$ in \eqref{prop1_lambda} which satisfies \eqref{prop1_lambda_check}, it holds with an absolute constant $C_3 > 0$ that
 \begin{align}
 \P \Bigg\{(I) \leq &     \frac{3}{8n_0} \big\|  \widehat{f} - f  \big\|_2^2  + C_3 \frac{n_{\max}^{(0)} /\widetilde{n}_{\min}^{(0)} (s_0+1) \big(1 +\log(n_{\max}^{(0)}) \big)}{K^2 \big( \sum_{k =1}^K n_k^{-1} \big)^{-1}}  \Bigg\} \geq 1 - n_0^{-c_{3}}, \nonumber
 \end{align}
which proves \eqref{prop1_2} when $s_0 < n_0 - 1$

\medskip
\noindent\textbf{Step 1.3.} We consider the term $(II)$ in   \eqref{prop1_1}. Note that by applying the Cauchy-Schwartz inequality and utilising the fact that $\vert ab \vert \leq 4a^2 + b^2/16$, we can establish   that 
 \begin{align}\label{prop1_12} 
    (II) \leq &     \frac{4 \big\|K^{-1} \sum_{k =1}^K \widetilde{P}^{n_0, n_k} \delta^{(k)} \big\|_2^2 }{n_{0}}  + \frac{1}{16n_{0}} \big\| \widehat{f}^{[K]} - f   \big\|_2^2 \nonumber \\
    \leq &     \frac{4 K^{-1}\sum_{k=1}^K \big\| \widetilde{P}^{n_0, n_k} \delta^{(k)} \big\|_2^2 }{n_{0}}  + \frac{1}{16n_{0}} \big\| \widehat{f}^{[K]} - f   \big\|_2^2 \nonumber \\
    \leq &   
     8 K^{-1}\sum_{k=1}^K   \frac{  \big\|  \delta^{(k)} \big\|_2^2 }{n_{k}} + \frac{1}{16n_{0}} \big\| \widehat{f}^{[K]} - f   \big\|_2^2
\end{align}
where the second inequality follows from Cauchy--Schwartz inequality and the last inequality follows from \Cref{lemma-alignment-delta} and $\min_{k \in [K]} n_k \geq n_0$.

\medskip
\noindent\textbf{Step 1.4.} Choosing  $\lambda$ as \eqref{prop1_lambda}, and combining \eqref{prop1_1}, \eqref{prop1_2} and \eqref{prop1_12}, we have with an absolute constant $C_4 >0$ that  
\begin{align}
 \P  \Bigg\{  \big\|  \widehat{f}^{[K]} - f  \big\|_{1/n_0}^2
    \leq &    C_4  \frac{n_{\max}^{(0)} /\widetilde{n}_{\min}^{(0)} (s_0+1) \big(1 +\log(n_{\max}^{(0)}) \big)}{K^2  \big( \sum_{k =1}^K  n_k^{-1} \big)^{-1} }  + K^{-1}\sum_{k=1}^K   n_k^{-1} \big\|  \delta^{(k)} \big\|_2^2  \Bigg\} \geq 1 - n_0^{-c_{\mathcal{E}}}. \nonumber
 \end{align}
 With \Cref{ass-fusedlasso},  if 
 \[
\lambda =  C_{\lambda}
 \Bigg(  (s_0+1) K^{2}  \Big( \sum_{k =1}^K n_k^{-1} \Big)^{-1}\Bigg)^{-1/2},
 \]
 then it holds that  
 \begin{align}
 \P  \Bigg\{  \big\|  \widehat{f}^{[K]} - f  \big\|_{1/n_0}^2
    \leq &    C_4  \frac{ (s_0+1) \big\{1 +\log\big(n_0/(s_0+1)\big)\big\}   }{K^2  \big( \sum_{k =1}^K  n_k^{-1} \big)^{-1} } + K^{-1}\sum_{k=1}^K   n_k^{-1} \big\|  \delta^{(k)} \big\|_2^2 \Bigg\} \geq 1 - n_0^{-c_{\mathcal{E}}}, \nonumber
 \end{align}
 completing the proof of \eqref{upper_bound_l_1_ora}.

\medskip
\noindent\textbf{Step 2.} This step is structured into four sub-steps. In \textbf{Step 2.1}, we decompose our target quantity into several terms. We deal with these terms individually in \textbf{Step 2.2} and \textbf{Step 2.3}. Finally, in \textbf{Step 2.4}, we gather all the pieces and conclude the proof of \eqref{upper_bound_l_0_ora}.

\medskip
\noindent\textbf{Step 2.1.}
    It directly follows from the definition of  $\widetilde{f}^{[K]}$ that
\begin{align}
  &   \frac{1}{2n_{0}} \bigg\|\frac{1}{K} \sum_{k=1}^K \widetilde{P}^{n_0, n_k} y^{(k)}  - \widetilde{f}^{[K]} \bigg\|_2^2   +\widetilde{\lambda} \|D \widetilde{f}^{[K]}\|_0
  \leq  \frac{1}{2n_{0}} \bigg\|\frac{1}{K} \sum_{k=1}^K \widetilde{P}^{n_0, n_k} y^{(k)}  - f \bigg\|_2^2   +\widetilde{\lambda} \|D f  \|_0. \nonumber
\end{align}
Since for any $k \in [K]$, $y^{(k)} = f^{(k)} + \epsilon^{(k)}$ with $f^{(k)} = P^{n_k, n_0} f + \delta^{(k)}$, it holds that  
\begin{align}
\frac{1}{2n_{0}} \bigg\|\widetilde{f}^{[K]} - \frac{1}{K} \sum_{k=1}^K \widetilde{P}^{n_0, n_k} P^{n_k, n_0} f  \bigg\|_2^2   
\leq  &  \frac{1}{2n_{0}} \bigg\|f - \frac{1}{K} \sum_{k=1}^K \widetilde{P}^{n_0, n_k} P^{n_k, n_0} f  \bigg\|_2^2  +  \frac{1}{n_0}\big( \epsilon^{K}\big)^{\top} \big( \widetilde{f}^{[K]}  -  f\big)   \nonumber\\
 & \hspace{0.5cm} + \widetilde{\lambda}  \|D   f \|_0 - \widetilde{\lambda}_{ K }   \|D  \widetilde{f}^{[K]}\|_0
+ \frac{1}{n_0}\big( \delta^{K}  \big)^{\top} \big( \widetilde{f}^{[K]}  -  f \big), \nonumber 
\end{align}
with $\epsilon^{K} =K^{-1} \sum_{k=1}^K \widetilde{P}^{n_0, n_k} \epsilon^{(k)}$ and $\delta^{K} = K^{-1} \sum_{k=1}^K \widetilde{P}^{n_0, n_k} \delta^{(k)}$. By \Cref{lemma-P}, it holds that
\begin{align}\label{prop2_1}
  \frac{1}{2n_{0}} \big\|\widetilde{f}^{[K]} -  f  \big\|_2^2 
\leq  &   \frac{1}{n_0}\big( \epsilon^{K}\big)^{\top} \big( \widetilde{f}^{[K]}  -  f\big)  
  + \widetilde{\lambda}  \|D   f \|_1 - \widetilde{\lambda}   \|D\widetilde{f}^{[K]}\|_1
+ \frac{1}{n_0}\big(  \delta^{K}\big)^{\top} \big( \widetilde{f}^{[K]}  -  f\big) \nonumber \\
= & (I.1)+(I.2)+(I.3)+(II) = (I) + (II).
\end{align}

\medskip
\noindent\textbf{Step 2.2.} In this step, we consider the term $(I)$ in \eqref{prop2_1}.

Let  $\mathcal{S}$ be defined in \eqref{def-S}  with  cardinality $s_0$ and  $\widetilde{\mathcal{S}}^{K}$ be defined as
\begin{align}\label{prop2_4}
\widetilde{\mathcal{S}}^{K} = \big\{i \in [n_{0}-1]\colon \widetilde{f}^{[K]}_i \neq \widetilde{f}^{[K]}_{i+1}  \big\}  =  \big\{i \in [n_{0}-1]\colon \big( D\widetilde{f}^{[K]}\big)_i \neq 0 \big\}. 
\end{align}
Let the orthogonal projection operator $P^{\widetilde{\mathcal{S}}^{K} \cup \mathcal{S}}$ be defined in \Cref{lemm_l_0}, then we have that 
\begin{align}\label{prop2_5}
   (I.1) = &  \frac{1}{n_{0}}  \big(\epsilon^{K} \big)^{\top} \Big( P^{\widetilde{\mathcal{S}}^{K} \cup \mathcal{S}} \big( \widetilde{f}^{[K]} - f  \big)  \Big)
       =       \frac{1}{n_{0}} \Big(P^{ \widetilde{\mathcal{S}}^{K} \cup \mathcal{S}} \epsilon^{K} \big)^{\top} \big( \widetilde{f}^{[K]} - f  \big)   \nonumber\\
       \leq    &    \frac{1}{n_{0}}  \big\|P^{\widetilde{\mathcal{S}}^{K} \cup \mathcal{S}} \epsilon^{K} \big\|_2
       \big\|\widetilde{f}^{[K]} - f  \big\|_2
         \leq          \frac{1}{n_{0}}  \big\|P^{\widetilde{\mathcal{S}}^{K} \cup \mathcal{S}}\epsilon^{K} \big\|_2^2+  \frac{1}{4n_{0}} 
       \big\| \widetilde{f}^{[K]} - f    \big\|_2^2, 
\end{align}
where the first inequality arises from Cauchy--Schwartz inequality and the last inequality is based on the fact that $\vert ab \vert \leq a^2 + b^2/4$.
By \eqref{prop1_3}, \eqref{prop2_5} and \Cref{lemm_l_0}, we have that $\mathbb{P}\{\mathcal{E}\} \geq 1- n_{0}^{-c_{\epsilon}^{\prime}}$ with 
\begin{align}\label{prop2_6}
 \mathcal{E} = \bigg\{  (I.1) \leq \frac{1}{4n_{0}} 
       \big\| \widetilde{f}^{[K]} - f    \big\|_2^2+ C_{\epsilon}^{\prime}  \frac{  \big( \vert \widetilde{\mathcal{S}}^{K} \cup \mathcal{S}\vert +1 \big) \big\{1 + \log\big( n_{0}/ (\vert\widetilde{\mathcal{S}}^{K} \cup \mathcal{S}\vert+1) \big) \big\} }{K^2  \big( \sum_{k =1}^K  n_k^{-1} \big)^{-1} }\bigg\},
\end{align}
where $C_{\epsilon}^{\prime}, c_{\epsilon}^{\prime} >0$ are absolute constants. From now on we assume that the event $\mathcal{E}$ holds. Then it holds that 
\begin{align}\label{prop2_7}
   (I)
 \leq & \frac{1}{4n_{0}} 
       \big\| \widetilde{f}^{[K]} - f    \big\|_2^2 +  C_{\epsilon}^{\prime}  \frac{  \big( \vert \widetilde{\mathcal{S}}^{K} \cup \mathcal{S}\vert +1 \big) \big\{1 + \log\big( n_{0}/ (\vert\widetilde{\mathcal{S}}^{K} \cup \mathcal{S}\vert+1) \big) \big\} }{K^2  \big( \sum_{k =1}^K  n_k^{-1} \big)^{-1} }+  \widetilde{\lambda} \|D   f \|_0 - \widetilde{\lambda}   \|D   \widetilde{f}^{[K]}\|_0\nonumber\\
= &  \frac{1}{4n_{0}} 
       \big\| \widetilde{f}^{[K]} - f    \big\|_2^2 +  C_{\epsilon}^{\prime}  \frac{  \big( \vert \widetilde{\mathcal{S}}^{K} \cup \mathcal{S}\vert +1 \big) \big\{1 + \log\big( n_{0}/ (\vert\widetilde{\mathcal{S}}^{K} \cup \mathcal{S}\vert+1) \big) \big\} }{K^2  \big( \sum_{k =1}^K  n_k^{-1} \big)^{-1} }+  \widetilde{\lambda} \big(s_0 - \vert \widetilde{\mathcal{S}}^{K} \vert \big) \nonumber\\
\leq & \frac{1}{4n_{0}} 
       \big\| \widetilde{f}^{[K]} - f    \big\|_2^2  + 2 \widetilde{\lambda}   (s_0+1) \nonumber\\
 = &\frac{1}{4n_{0}} 
       \big\| \widetilde{f}^{[K]} - f    \big\|_2^2  + 2 C_{\widetilde{\lambda}}\frac{( s_0+1) \big\{ 1+ \log\big(n_{0}/(s_0+1)\big) \big\} }{K^2  \big( \sum_{k =1}^K  n_k^{-1} \big)^{-1} },
\end{align}
where the first equality follows from the definitions of $\mathcal{S}$ and $\widetilde{\mathcal{S}}^{K}$ in \eqref{def-S} and \eqref{prop2_4}, and the second inequality and the last equality follow from the choice of $\widetilde{\lambda}$ in \eqref{tuning-ora} and $C_{\widetilde{\lambda}} >0$ is a large enough absolute constant.

\medskip
\noindent\textbf{Step 2.3.} In this sub-step, we deal with the term $(II)$ in  \eqref{prop2_1}. Note that by applying the Cauchy-Schwartz inequality and utilising the fact that $\vert ab \vert \leq 2a^2 + b^2/8$, we can establish that 
 \begin{align}\label{prop2_8}
    (II) \leq &     \frac{2 \big\|K^{-1} \sum_{k =1}^K \widetilde{P}^{n_0, n_k} \delta^{(k)} \big\|_2^2 }{n_{0}}  + \frac{1}{8n_{0}} \big\| \widetilde{f}^{[K]} - f   \big\|_2^2 \nonumber \\
    \leq &     \frac{2 K^{-1} \sum_{k =1}^K \big\| \widetilde{P}^{n_0, n_k} \delta^{(k)} \big\|_2^2 }{n_{0}}  + \frac{1}{8n_{0}} \big\| \widetilde{f}^{[K]} - f   \big\|_2^2 \nonumber \\
    \leq &    4 K^{-1} \sum_{k =1}^K \frac{  \big\|  \delta^{(k)} \big\|_2^2 }{n_{k}} + \frac{1}{8n_{0}} \big\| \widetilde{f}^{[K]} - f   \big\|_2^2,
\end{align}
where the second inequality follows from Cauchy--Schwartz inequality and the last inequality follows from \Cref{lemma-alignment-delta} and  $\min_{k \in [K]} n_k \geq n_0$.

\medskip
\noindent\textbf{Step 2.4.} Choosing $\widetilde{\lambda}$ as \eqref{tuning-ora} and combining  \eqref{prop2_1}, \eqref{prop2_6}, \eqref{prop2_7} and  \eqref{prop2_8}, we have with an absolute constant $C_5>0$ that 
 \begin{align}
    \P \bigg\{  \big\|\widetilde{f}^{[K]} -  f  \big\|_{1/n_0}^2  
\leq & C_{5} \bigg( \frac{( s_0+1) \big\{1 + \log\big(n_{0}/(s_0+1)\big) \big\} }{ K^2 \big(\sum_{k =1}n_k^{-1}\big)^{-1} }    + K^{-1} \sum_{k =1}^K n_k^{-1} \big\|  \delta^{(k)} \big\|_2^2\bigg) \bigg\} \geq 1 - n_0^{-c_{\epsilon}^{\prime}}, \nonumber
\end{align}
completing the proof of \eqref{upper_bound_l_0_ora}.
\end{proof}

\subsection{Proof of Theorem \ref{theorem_detection_consistency}}\label{app_3_det}

\begin{proof}[{Proof of \Cref{theorem_detection_consistency}}]
The proof consists of four steps. In \textbf{Step 1}, we decompose our target quantity into several terms. We then handle these terms individually in \textbf{Steps 2} and \textbf{3}. In \textbf{Step 4}, we gather all the pieces and conclude the proof.

\medskip
\noindent \textbf{Step 1.}
Recall that for any $k \in [K]$,  
\begin{align}\label{def-delta_k}
    \widehat{\Delta}^{(k)} = n_k^{-1/2} y^{(k)} - n_k^{-1/2}  P^{n_k, n_0} y,
\end{align}
and 
\[
\widehat{T}_k=  \Big\{ i \in [n_k]\colon \big\vert \widehat{\Delta}_{i}^{(k)}\big\vert \mbox{ is among the first } \widehat{t}_{k} \mbox{ largest of all}\Big\},
\]
as defined in \Cref{alg_isd}.  Note that 
\begin{align}\label{detection_1}
   &\P \big\{  \widehat{\mathcal{A}} = \mathcal{A}_{h^{*}} \big\}  =  1-  \P \big\{  \widehat{\mathcal{A}} \neq \mathcal{A}_{h^{*}} \big\}   \nonumber \\
   =  &1- \P \Big\{ \exists k \in \mathcal{A}_{h^{*}}^{c} \mbox{ such that }\Big\| \big( \widehat{\Delta}^{(k)} \big)_{\widehat{T}_k} \Big\|_2^2  \leq   \tau_k  \mbox{ or } \exists k \in \mathcal{A}_{h^{*}} \mbox{ such that }\Big\| \big( \widehat{\Delta}^{(k)} \big)_{\widehat{T}_k} \Big\|_2^2  >  \tau_k  \Big\} \nonumber \\
   \geq &  1- \P \Big\{\exists k \in \mathcal{A}_{h^{*}}^{c} \mbox{ such that }\Big\| \big( \widehat{\Delta}^{(k)} \big)_{\widehat{T}_k} \Big\|_2^2  \leq   \tau_k \Big\}  - \P \Big\{ \exists k \in \mathcal{A}_{h^{*}} \mbox{ such that }\Big\| \big( \widehat{\Delta}^{(k)} \big)_{\widehat{T}_k} \Big\|_2^2  >  \tau_k  \Big\} \nonumber\\
     \geq &  1- \big\vert \mathcal{A}_{h^{*}}^{c}  \big\vert  \max_{k \in  \mathcal{A}_{h^{*}}^{c}}\P \Big\{ \Big\| \big( \widehat{\Delta}^{(k)} \big)_{\widehat{T}_k} \Big\|_2^2  \leq   \tau_k \Big\}  - \vert \mathcal{A}_{h^{*}} \vert  \max_{k \in  \mathcal{A}_{h^{*}}} \P \Big\{ \Big\| \big( \widehat{\Delta}^{(k)} \big)_{\widehat{T}_k} \Big\|_2^2  >  \tau_k  \Big\} \nonumber\\
   =  &  1- (I) -(II),
\end{align}
where the first and second inequalities follow from a union bound argument.

For any $k \in \{ 0\} \cup [K]$ define the event  
\[
    \mathcal{E}_k = \Big\{ \max_{i \in [n_k]} \big\vert  \epsilon^{(k)}_i\big\vert \leq  \sqrt{\log(n_0 \vee n_k)} \mbox{ and } \max_{i \in [n_k]} \big\vert ( P^{n_k, n_0}\epsilon)_i\big\vert \leq  \sqrt{\log(n_0 \vee n_k)} \Big\}.
\]

By a union bound argument, we have with an absolute constant $c_1> 0$ that 
\begin{align}\label{detection_1_2_0}
    \P \big\{ \mathcal{E}_k \big\} \geq & 1 -  \P\Big\{ \max_{i \in [n_k]} \big\vert  \epsilon^{(k)}_i\big\vert \geq \sqrt{\log(n_0 \vee n_k)} \Big\} - \P\Big\{\max_{i \in [n_k]} \big\vert ( P^{n_k, n_0}\epsilon)_i\big\vert\geq  \sqrt{\log(n_0 \vee n_k)} \Big\}  \nonumber\\
    \geq & 1 -  n_k  \max_{i \in [n_k]}  \P \Big\{  \big\vert  \epsilon^{(k)}_i\big\vert \geq  \sqrt{\log(n_0 \vee n_k)} \Big\} - \P\Big\{\max_{i \in [n_k]} \big\vert ( P^{n_k, n_0}\epsilon)_i\big\vert \geq  \sqrt{\log(n_0 \vee n_k)} \Big\}  \nonumber\\
    \geq & 1- (n_0 \vee n_k)^{-c_1} - \P\Big\{\max_{i \in [n_k]} \big\vert ( P^{n_k, n_0}\epsilon)_i\big\vert \geq  \sqrt{\log(n_0 \vee n_k)} \Big\},
\end{align}
where the last inequality follows from the assumption that $ \{\epsilon_i^{(k)}\}_{i = 1, k=1}^{n_k, K}$  are mutually independent mean-zero $C_{\sigma}$-sub-Gaussian distributed and Proposition 2.5.2 in \cite{vershynin2018high}.
We have that 
\begin{align}\label{detection_1_2_1}
    \P\Big\{\max_{i \in [n_k]} \big\vert  P^{n_k, n_0}\epsilon_i\big\vert \geq  \sqrt{\log(n_0 \vee n_k)} \Big\}  \leq & \P\Big\{\max_{i \in [n_0]} \big\vert  \epsilon_i\big\vert \geq  \sqrt{\log(n_0 \vee n_k)} \Big\}  
    \nonumber \\
    \leq & n_0  \max_{i \in [n_0]}  \P \Big\{  \big\vert  \epsilon_i\big\vert \geq  \sqrt{\log(n_0 \vee n_k)} \Big\} 
    \leq  (n_0 \vee n_k)^{-c_1},
\end{align}
where the first inequality follows from the definition of $P^{n_k, n_0}$ in \eqref{def-P}, and the last inequality follows from the assumption that $ \{\epsilon_i\}_{i = 1}^{n_0}$  are mutually independent mean-zero $C_{\sigma}$-sub-Gaussian distributed and Proposition 2.5.2 in \cite{vershynin2018high}. 

Combining \eqref{detection_1_2_0} and \eqref{detection_1_2_1}, we have with an absolute constant $c_2> 0$ that 
\begin{align}\label{detection_1_2}
    \P \big\{ \mathcal{E}_k \big\} \geq  & 1- (n_0 \vee n_k)^{-c_2}. 
\end{align}

\medskip
\noindent \textbf{Step 2.} In this step, we address the term $(I)$ in \eqref{detection_1}. We decompose the term $(I)$ in \eqref{detection_1} into three components, which we then deal with separately in \textbf{Steps 2.1}, \textbf{2.2}, and \textbf{2.3}. In \textbf{Step 2.4}, we gather all the pieces and conclude the proof of this step. 

For any measurable sets $A_1$, $A_2$ and $A_3$, we have that 
\begin{align}
    \P\{ A_1  \} = & \P\big\{ A_1  \cap (A_2 \cup A_3) \big\} + \P\big\{ A_1  \cap (A_2 \cup A_3)^{c} \big\} \nonumber \\
    = & \P\big\{ (A_1  \cap A_2) \cup(A_1  \cap A_3) \big\} + \P\big\{ A_1  \cap A_2^{c} \cap A_3^{c} \big\} \nonumber\\
    \leq & \P\big\{ A_1  \cap A_2 \big\} + \P\big\{A_1  \cap A_3\big\} + \P\big\{ A_1  \cap A_2^{c} \cap A_3^{c} \big\}, \nonumber
\end{align}
where the last inequality follows from a union bound argument.  We have that
\begin{align}\label{detection_2_1}
      (I)  \leq  &   \big\vert \mathcal{A}_{h^{*}}^{c}  \big\vert  \max_{k \in  \mathcal{A}_{h^{*}}^{c}}\P\Big\{  \Big\| \big( \widehat{\Delta}^{(k)} \big)_{\widehat{T}_k} \Big\|_2^2  \leq   \tau_k  \mbox{ and }  \widehat{T}_k \subseteq  \mathcal{H}_k  \Big\}
      \nonumber \\
       & \hspace{0.5cm}
       +  \big\vert \mathcal{A}_{h^{*}}^{c}  \big\vert  \max_{k \in  \mathcal{A}_{h^{*}}^{c}} \P\Big\{  \Big\| \big( \widehat{\Delta}^{(k)} \big)_{\widehat{T}_k} \Big\|_2^2  \leq   \tau_k \mbox{ and } \mathcal{H}_k \subseteq \widehat{T}_k \Big\}  \nonumber \\
       & \hspace{0.5cm} +\big\vert \mathcal{A}_{h^{*}}^{c}  \big\vert  \max_{k \in  \mathcal{A}_{h^{*}}^{c}} \P\Big\{ \Big\| \big( \widehat{\Delta}^{(k)} \big)_{\widehat{T}_k} \Big\|_2^2  \leq   \tau_k, \mathcal{H}_k \nsubseteq \widehat{T}_k \mbox{ and } \widehat{T}_k \nsubseteq  \mathcal{H}_k   \Big\}  \nonumber \\
      = &  (I.1) + (I.2)+ (I.3).
\end{align}

We now focus on the term  $ \min_{i \in \mathcal{H}_k} \vert  \widehat{\Delta}_{i}^{(k)}  \vert $, for each $k \in \mathcal{A}_{h^{*}}^{c}$. For any $k \in \mathcal{A}_{h^{*}}^{c}$, assume that $\mathcal{E}_k$ holds, then we have that 
\begin{align}\label{detection_2_2}
       \min_{i \in \mathcal{H}_k} \big\vert  \widehat{\Delta}_{i}^{(k)}  \big\vert 
        & = n_k^{-1/2} \min_{i \in \mathcal{H}_k} \Big\vert    \big( f_i^{(k)} - ( P^{n_k, n_0} f )_i \big) +  \big( \epsilon_i^{(k)} - ( P^{n_k, n_0} \epsilon )_i\big) \Big\vert \nonumber\\
        & \geq  n_k^{-1/2}  \Big( \min_{i \in \mathcal{H}_k} \big\vert  ( f_i^{(k)} - ( P^{n_k, n_0} f )_i  \big\vert -  \max_{i \in \mathcal{H}_k} \big\vert  \epsilon_i^{(k)} \big\vert  - \max_{i \in \mathcal{H}_k} \big\vert ( P^{n_k, n_0} \epsilon )_i \big\vert \Big)  \nonumber \\ 
         & \geq  n_k^{-1/2}  \Big( \min_{i \in \mathcal{H}_k} \big\vert  ( f_i^{(k)} - ( P^{n_k, n_0} f )_i  \big\vert -  \max_{i \in [n_k]} \big\vert  \epsilon_i^{(k)} \big\vert  -  \max_{i \in [n_k]} \big\vert ( P^{n_k, n_0} \epsilon )_i \big\vert \Big)\nonumber\\   
            & \geq  n_k^{-1/2}  \Big(4 \sqrt{\log(n_0 \vee n_k)} -  \max_{i \in [n_k]} \big\vert  \epsilon_i^{(k)} \big\vert  -  \max_{i \in [n_k]} \big\vert ( P^{n_k, n_0} \epsilon )_i \big\vert \Big), \nonumber \\
        & \geq  2n_k^{-1/2}  \sqrt{\log(n_0 \vee n_k)},
\end{align}
where the first equality follows from the definition of $ \widehat{\Delta}^{(k)}$ in \eqref{def-delta_k}, the third inequality follows from \Cref{ass_detection} and the final inequality is a consequence of the event $\mathcal{E}_k$. 

\medskip
\noindent \textbf{Step 2.1.}
We now consider the term $(I.1)$ in \eqref{detection_2_1}.  For any $k \in \mathcal{A}_{h^{*}}^{c}$, assuming that $\widehat{T}_k \subseteq  \mathcal{H}_k $ and the event $\mathcal{E}_k$ hold, then we have that
\begin{align}\label{detection_2.1_1}
  \Big\| \big( \widehat{\Delta}^{(k)} \big)_{\widehat{T}_k} \Big\|_2^2  \geq  \widehat{t}_{k}   \min_{i \in \mathcal{H}_k} \big\vert  \widehat{\Delta}_{i}^{(k)} \big\vert^2 \geq \frac{C_{\widehat{\mathcal{A}}}}{2}\frac{ (s_{0}+1) \big\{1+\log\big(n_0/(s_0+1)\big) + \log(n_0 \vee n_k)\big\}}{n_0} > \tau_k.
\end{align}
where the second inequality follows from \eqref{detection_2_2} and the definition of $\widehat{t}_{k}$ in \eqref{hat_t_k-def}, and the last inequality follows from the definition of $ \tau_k$ in \eqref{tau-def}.
Note that
\begin{align}\label{detection_2.1_2}
    (I.1) = &  \big\vert \mathcal{A}_{h^{*}}^{c}  \big\vert  \max_{k \in  \mathcal{A}_{h^{*}}^{c}} \Big(   \P \Big\{  \Big\| \big( \widehat{\Delta}^{(k)} \big)_{\widehat{T}_k} \Big\|_2^2  \leq   \tau_k, \widehat{T}_k \subseteq  \mathcal{H}_k  \mbox{ and }\mathcal{E}_k \Big\}
    \nonumber\\
    & \hspace{0.5cm}
    + \P\Big\{  \Big\| \big( \widehat{\Delta}^{(k)} \big)_{\widehat{T}_k} \Big\|_2^2  \leq   \tau_k,\widehat{T}_k \subseteq  \mathcal{H}_k  \mbox{ and } \mathcal{E}_k^{c}   \Big\}\Big)\nonumber \\
   \leq &  \big\vert \mathcal{A}_{h^{*}}^{c}  \big\vert  \max_{k \in  \mathcal{A}_{h^{*}}^{c}}   \P\Big\{   \Big\| \big( \widehat{\Delta}^{(k)} \big)_{\widehat{T}_k} \Big\|_2^2  \leq   \tau_k,  \widehat{T}_k \subseteq  \mathcal{H}_k  \mbox{ and } \mathcal{E}_k  \Big\} +   \big\vert \mathcal{A}_{h^{*}}^{c}  \big\vert  \max_{k \in  \mathcal{A}_{h^{*}}^{c}} \P\big\{ \mathcal{E}_k^{c} \big\}  \nonumber \\
      \leq & \big\vert \mathcal{A}_{h^{*}}^{c} \big\vert \max_{ k \in  \mathcal{A}_{h^{*}}^{c}}
        \P\big\{ \mathcal{E}_k^{c} \big\} 
   \leq \big\vert \mathcal{A}_{h^{*}}^{c} \big\vert \{n_0\vee (\min_{k \in [K]}n_k)\}^{-c_2},
\end{align}
where the first inequality is based on the fact that for any sets $A_1, A_2$ and $A_3$, $\P(A_1 \cap A_2 \cap A_3) \leq \P(A_3)$, the second inequality follows from \eqref{detection_2.1_1} and the last inequality follows from \eqref{detection_1_2}.

\medskip
\noindent \textbf{Step 2.2.} 
We now consider the term $(I.2)$ in \eqref{detection_2_1}. 
For any $k \in \mathcal{A}_{h^{*}}^{c}$, assuming that $  \mathcal{H}_k \subseteq \widehat{T}_k$ and the event $\mathcal{E}_k$ holds, then we have that
\begin{align}\label{detection_2.2_1}
       \Big\| \big( \widehat{\Delta}^{(k)} \big)_{\widehat{T}_k} \Big\|_2^2   
     = & n_k^{-1}\Big\| \big(  f^{(k)} - P^{n_k, n_0} f \big)_{\widehat{T}_k} +  
 \epsilon^{(k)}_{\widehat{T}_k} - \big( P^{n_k, n_0} \epsilon \big)_{\widehat{T}_k}\Big\|_2^2 \nonumber \\
 \geq &  n_k^{-1} \Big( 
     \Big\| \big(  f^{(k)} - P^{n_k, n_0} f \big)_{\widehat{T}_k} \Big\|_2^2 - 
    \Big\|  \epsilon^{(k)}_{\widehat{T}_k} \Big\|_2^2  - \Big\| \big( P^{n_k, n_0} \epsilon  \big)_{\widehat{T}_k}\Big\|_2^2  \Big) \nonumber\\
    \geq &  n_k^{-1} \Big( 
     \Big\| \big(  f^{(k)} - P^{n_k, n_0} f \big)_{\mathcal{H}_k} \Big\|_2^2 -  \max_{ \substack{ T_k \subseteq [n_k] \\ \mbox{with }\vert T_k \vert = \widehat{t}_{k} }}
    \Big\|  \epsilon^{(k)}_{T_k} \Big\|_2^2  -  \max_{ \substack{ T_k \subseteq [n_k] \\ \mbox{with }\vert T_k \vert = \widehat{t}_{k} }} \Big\| \big( P^{n_k, n_0} \epsilon  \big)_{T_k}\Big\|_2^2  \Big) \nonumber \\
        \geq &   n_k^{-1} \Big( 
     \Big\| \big(  f^{(k)} - P^{n_k, n_0} f \big)_{\mathcal{H}_k} \Big\|_2^2 -  \widehat{t}_{k} \max_{ i \in [n_k] }
    \Big\vert  \epsilon^{(k)}_i \Big\vert^2  -  \widehat{t}_{k} \max_{ i \in [n_k] } \big\vert ( P^{n_k, n_0} \epsilon  )_{i}\big \vert^2  \Big) \nonumber\\
    \geq &  \big(  C_{\mathcal{A}^c} - C_{\widehat{\mathcal{A}}} /4 \big) \frac{ (s_{0}+1) \big\{1+\log\big(n_0/(s_0+1)\big) + \log(n_0\vee n_k)\big\}}{n_0} > \tau_k,
\end{align}
where
\begin{itemize}
     \item the first equality is a consequence of the definition of $ \widehat{\Delta}^{(k)} $ shown in \eqref{def-delta_k},
    \item the second inequality follows from the assumption that $ \mathcal{H}_k \subseteq \widehat{T}_k$,
    \item the forth inequality follows from the event $\mathcal{E}_k$,  \Cref{ass_detection} and the definition of $\widehat{t}_{k}$ as shown in  \eqref{hat_t_k-def},
    \item and the final inequality follows from  the definition of $\tau_k$ shown in  \eqref{tau-def} and $C_{\widehat{\mathcal{A}}} \leq C_{\mathcal{A}^c}/2$.
\end{itemize}
Note that
\begin{align}\label{detection_2.2_2}
    (I.2) = &  \big\vert \mathcal{A}_{h^{*}}^{c}  \big\vert  \max_{k \in  \mathcal{A}_{h^{*}}^{c}} \Big[   \P\Big\{  \Big\| \big( \widehat{\Delta}^{(k)} \big)_{\widehat{T}_k} \Big\|_2^2  \leq   \tau_k, \mathcal{H}_k \subseteq \widehat{T}_k \mbox{ and }\mathcal{E}_k \Big\}
        \nonumber\\
    & \hspace{0.5cm}
    + \P\Big\{  \Big\| \big( \widehat{\Delta}^{(k)} \big)_{\widehat{T}_k} \Big\|_2^2  \leq   \tau_k,\mathcal{H}_k \subseteq \widehat{T}_k \mbox{ and } \mathcal{E}_k^{c}   \Big\}\Big]\nonumber \\
   \leq &  \big\vert \mathcal{A}_{h^{*}}^{c}  \big\vert  \max_{k \in  \mathcal{A}_{h^{*}}^{c}}   \P\Big\{   \Big\| \big( \widehat{\Delta}^{(k)} \big)_{\widehat{T}_k} \Big\|_2^2  \leq   \tau_k,  \mathcal{H}_k \subseteq \widehat{T}_k  \mbox{ and } \mathcal{E}_k  \Big\} +   \big\vert \mathcal{A}_{h^{*}}^{c}  \big\vert  \max_{k \in  \mathcal{A}_{h^{*}}^{c}} \P\big\{ \mathcal{E}_k^{c} \big\}  \nonumber \\
      \leq & \big\vert \mathcal{A}_{h^{*}}^{c} \big\vert \max_{ k \in  \mathcal{A}_{h^{*}}^{c}}
        \P\big\{ \mathcal{E}_k^{c} \big\} 
   \leq \big\vert \mathcal{A}_{h^{*}}^{c} \big\vert \{n_0\vee (\min_{k \in [K]}n_k)\}^{-c_2},
\end{align}
where the first inequality is based on the fact that for any sets $A_1, A_2$ and $A_3$, $\P(A_1 \cap A_2 \cap A_3) \leq \P(A_3)$, the second inequality follows from \eqref{detection_2.2_1} and the last inequality follows from \eqref{detection_1_2}.

\medskip
\noindent \textbf{Step 2.3.} 
We now consider the term $(I.3)$ in \eqref{detection_2_1}.
Note that 
\begin{align}
    \Big\{   \mathcal{H}_k \nsubseteq \widehat{T}_k \mbox{ and } \widehat{T}_k \nsubseteq  \mathcal{H}_k   \Big\} \subseteq 
     \Big\{   \min_{i \in \mathcal{H}_k} \big\vert  \widehat{\Delta}_{i}^{(k)} \big\vert \leq \min_{i \in \widehat{T}_k}\big\vert  \widehat{\Delta}_{i}^{(k)} \big\vert \Big\},  \nonumber
\end{align}
then we have that
\begin{align}\label{detection_2.3_1}
    (I.3) \leq &  \big\vert \mathcal{A}_{h^{*}}^{c}  \big\vert  \max_{k \in  \mathcal{A}_{h^{*}}^{c}}  \P\bigg\{  \Big\| \big( \widehat{\Delta}^{(k)} \big)_{\widehat{T}_k} \Big\|_2^2  \leq   \tau_k \mbox{ and }  \min_{i \in \mathcal{H}_k} \big\vert  \widehat{\Delta}_{i}^{(k)} \big\vert \leq \min_{i \in \widehat{T}_k}\big\vert  \widehat{\Delta}_{i}^{(k)} \big\vert  \bigg\}\nonumber \\
    = &  \big\vert \mathcal{A}_{h^{*}}^{c}  \big\vert  \max_{k \in  \mathcal{A}_{h^{*}}^{c}} \bigg[  \P\bigg\{  \Big\| \big( \widehat{\Delta}^{(k)} \big)_{\widehat{T}_k} \Big\|_2^2  \leq   \tau_k, \min_{i \in \mathcal{H}_k} \big\vert  \widehat{\Delta}_{i}^{(k)} \big\vert \leq \min_{i \in \widehat{T}_k}\big\vert  \widehat{\Delta}_{i}^{(k)} \big\vert  \mbox{ and } \mathcal{E}_k \bigg\}  \nonumber\\
    & \hspace{0.5cm} +  \P\bigg\{  \Big\| \big( \widehat{\Delta}^{(k)} \big)_{\widehat{T}_k} \Big\|_2^2  \leq   \tau_k, \min_{i \in \mathcal{H}_k} \big\vert  \widehat{\Delta}_{i}^{(k)} \big\vert \leq \min_{i \in \widehat{T}_k}\big\vert  \widehat{\Delta}_{i}^{(k)} \big\vert  \mbox{ and } \mathcal{E}_k^{c} \bigg\} \bigg] \nonumber\\ 
    \leq &  \big\vert \mathcal{A}_{h^{*}}^{c}  \big\vert  \max_{k \in  \mathcal{A}_{h^{*}}^{c}}  \P\bigg\{  \Big\| \big( \widehat{\Delta}^{(k)} \big)_{\widehat{T}_k} \Big\|_2^2  \leq   \tau_k, \min_{i \in \mathcal{H}_k} \big\vert  \widehat{\Delta}_{i}^{(k)} \big\vert \leq \min_{i \in \widehat{T}_k}\big\vert  \widehat{\Delta}_{i}^{(k)} \big\vert  \mbox{ and } \mathcal{E}_k \bigg\} + \big\vert \mathcal{A}_{h^{*}}^{c}  \big\vert  \max_{k \in  \mathcal{A}_{h^{*}}^{c}} \P \big\{  \mathcal{E}_k^{c}\big\}\nonumber\\
     \leq & \big\vert \mathcal{A}_{h^{*}}^{c}  \big\vert  \max_{k \in  \mathcal{A}_{h^{*}}^{c}}  \P\bigg\{  \Big\| \big( \widehat{\Delta}^{(k)} \big)_{\widehat{T}_k} \Big\|_2^2  \leq   \tau_k, \min_{i \in \mathcal{H}_k} \big\vert  \widehat{\Delta}_{i}^{(k)} \big\vert \leq \min_{i \in \widehat{T}_k}\big\vert  \widehat{\Delta}_{i}^{(k)} \big\vert  \mbox{ and } \mathcal{E}_k \bigg\}   \nonumber \\
      & \hspace{0.5cm}+ \big\vert \mathcal{A}_{h^{*}}^{c} \big\vert \{n_0\vee (\min_{k \in [K]}n_k)\}^{-c_2},
\end{align}
where
\begin{itemize}
    \item the first equality is based on the fact that for any sets $A_1, A_2$ and $A_3$, $\P(A_1 \cap A_2) = \P(A_1 \cap A_2 \cap A_3) + \P(A_1 \cap A_2 \cap A_3^{c})$,
    \item the second inequality is based on the fact that for any sets $A_1, A_2$ and $A_3$, $\P(A_1 \cap A_2 \cap A_3) \leq \P(A_3)$,
     \item and the last inequality follows from \eqref{detection_1_2}.
\end{itemize}

For any $k \in \mathcal{A}_{h^{*}}^{c}$, assuming that  $\min_{i \in \mathcal{H}_k} \big\vert  \widehat{\Delta}_{i}^{(k)} \big\vert \leq \min_{i \in \widehat{T}_k}\big\vert  \widehat{\Delta}_{i}^{(k)} \big\vert $ and the event $\mathcal{E}_k$ hold, then we obtain that
\begin{align}\label{detection_2.3_2}
  \Big\| \big( \widehat{\Delta}^{(k)} \big)_{\widehat{T}_k} \Big\|_2^2  \geq  \widehat{t}_{k}   \min_{i \in \mathcal{H}_k} \big\vert  \widehat{\Delta}_{i}^{(k)} \big\vert^2 \geq \frac{C_{\widehat{\mathcal{A}}}  }{2}\frac{ (s_{0}+1) \big\{1+\log\big(n_0/(s_0+1)\big) + \log(n_0 \vee n_k) \big\}}{n_0} > \tau_k,
\end{align}
where the second inequality follows from \eqref{detection_2_2} and the definition of $\widehat{t}_{k}$ in \eqref{hat_t_k-def}, and the last inequality follows from the definition of $ \tau_k$ in \eqref{tau-def}.
Combining \eqref{detection_2.3_1} and \eqref{detection_2.3_2}, we have that 
\begin{align}\label{detection_2.3.3}
    (I.3) \leq & \big\vert \mathcal{A}_{h^{*}}^{c} \big\vert \{n_0\vee (\min_{k \in [K]}n_k)\}^{-c_2}.
\end{align}

\medskip
\noindent \textbf{Step 2.4.}
Combining \eqref{detection_2_1}, \eqref{detection_2.1_2}, \eqref{detection_2.2_2}, \eqref{detection_2.3.3}, it holds with an absolute constant $c_3 >0$ that
\begin{align}\label{detection_2.4}
     &  (I) \leq \big\vert \mathcal{A}_{h^{*}}^{c} \big\vert \{n_0\vee (\min_{k \in [K]}n_k)\}^{-c_3}.
\end{align}

\medskip
\noindent \textbf{Step 3.} 
In this step, we deal with the term $(II)$ in \eqref{detection_1}.
For any $k \in \mathcal{A}_{h^{*}}$, assuming that the event $\mathcal{E}_k$ holds,  we obtain that
\begin{align}\label{detection_3_1}
     \Big\| \big( \widehat{\Delta}^{(k)} \big)_{\widehat{T}_k} \Big\|_2^2   = & n_k^{-1}
     \Big\| \big(  f^{(k)} - P^{n_k, n_0} f \big)_{\widehat{T}_k} +  
 \epsilon^{(k)}_{\widehat{T}_k} - \big( P^{n_k, n_0} \epsilon \big)_{\widehat{T}_k}\Big\|_2^2\nonumber \\
 \leq & 3 n_k^{-1} \Big\{ 
     \Big\| \big(  f^{(k)} - P^{n_k, n_0} f \big)_{\widehat{T}_k} \Big\|_2^2 +  
    \Big\|  \epsilon^{(k)}_{\widehat{T}_k} \Big\|_2^2  + \Big\| \big( P^{n_k, n_0} \epsilon  \big)_{\widehat{T}_k}\Big\|_2^2  \Big\} \nonumber\\
    \leq & 3 n_k^{-1} \Big\{ 
     n_k(h^{*})^2 +   \max_{ \substack{ T_k \subseteq [n_k] \\ \mbox{with }\vert T_k \vert = \widehat{t}_{k} }}
    \Big\|  \epsilon^{(k)}_{T_k} \Big\|_2^2  +  \max_{ \substack{ T_k \subseteq [n_k] \\ \mbox{with }\vert T_k \vert = \widehat{t}_{k} }} \Big\| \big( P^{n_k, n_0} \epsilon  \big)_{T_k}\Big\|_2^2  \Big\} \nonumber \\
        \leq & 3 n_k^{-1} \Big\{
      n_k(h^{*})^2 +  \widehat{t}_{k} \max_{ i \in [n_k] }
    \Big\vert  \epsilon^{(k)}_i \Big\vert^2  +  \widehat{t}_{k} \max_{ i \in [n_k] } \Big\vert \big( P^{n_k, n_0} \epsilon  \big)_{i}\Big \vert^2  \Big\}\nonumber\\
               \leq & \big( C_{\mathcal{A}}  +C_{\widehat{\mathcal{A}}}/4 \big)  \frac{ (s_{0}+1) \big\{1+\log\big(n_0/(s_0+1)\big) + \log(n_0 \vee n_k)\big\}}{n_0} \leq \tau_k,
\end{align}
where 
\begin{itemize}
    \item the first equality follows from the definition of $ \widehat{\Delta}^{(k)} $ in \eqref{def-delta_k},
    \item the second inequality follows from the definition of $\mathcal{A}_{h^{*}}$ in \eqref{def-mathcal-A-h},
    \item the fourth inequality follows from the event $\mathcal{E}_k$, \Cref{ass_detection} and the definition of $\widehat{t}_{k}$ in   \eqref{hat_t_k-def},
    \item and the final inequality follows from the definition of $\tau_k$ in \eqref{tau-def} and  $ C_{\widehat{\mathcal{A}}} \geq 2C_{\mathcal{A}}$.
\end{itemize}
Note that 
\begin{align}\label{detection_3_3}
      (II) =  &  \big \vert \mathcal{A}_{h^{*}} \big \vert \max_{k \in \mathcal{A}_{h^{*}}} \big\{ \P \Big\{ \Big\| \big( \widehat{\Delta}^{(k)} \big)_{\widehat{T}_k} \Big\|_2^2  >  \tau_k \mbox{ and } \mathcal{E}_k \Big\} +   \P \Big\{ \Big\| \big( \widehat{\Delta}^{(k)} \big)_{\widehat{T}_k} \Big\|_2^2  >  \tau_k \mbox{ and } \mathcal{E}_k^{c}  \Big\} \big\} \nonumber \\
   \leq &  \big \vert \mathcal{A}_{h^{*}} \big \vert \max_{k \in \mathcal{A}_{h^{*}}}  \P \big\{  \mathcal{E}_k^{c}  \big\}  
    \leq  \big\vert \mathcal{A}_{h^{*}}  \big\vert  \{n_0\vee (\min_{k \in [K]}n_k)\}^{-c_2},  
\end{align}
where the first inequality is based on the fact that for any sets $A_1, A_2$ and $A_3$, $\P(A_1 \cap A_2 \cap A_3) \leq \P(A_3)$, the first inequality follows from \eqref{detection_3_1} and the last inequality follows from \eqref{detection_1_2}.

\medskip
\noindent \textbf{Step 4.}
Combining \eqref{detection_1}, \eqref{detection_2.4} and \eqref{detection_3_3}, it holds with an absolute constant $c_4>0$ that 
\begin{align}
   \P \big\{  \widehat{\mathcal{A}} = \mathcal{A}_{h^{*}} \big\}  
   \geq  & 1- \big\vert \mathcal{A}_{h^{*}}^{c} \big\vert   \{n_0\vee (\min_{k \in [K]}n_k)\}^{-c_3}  - \big\vert \mathcal{A}_{h^{*}} \big\vert  \{n_0\vee (\min_{k \in [K]}n_k)\}^{-c_2}  \nonumber\\
   \geq & 1- K  \{n_0\vee (\min_{k \in [K]}n_k)\}^{-c_4}, \nonumber
\end{align}
which completes the proof.

\end{proof}

\subsection{Proof of Corollary \ref{cor-a-hat-l0l1}}\label{app_3_coro} 
\begin{proof}[Proof of \Cref{cor-a-hat-l0l1}]
 This proof consists of two steps. In \textbf{Step 1}, we focus on establishing an estimation upper bound for $\widehat{f}^{\widehat{\mathcal{A}}}$, and in \textbf{Step 2}, we prove a similar estimation upper bound for $\widetilde{f}^{\widehat{\mathcal{A}}}$.

 \medskip
 \noindent \textbf{Step 1.} In this step, we focus on the estimator $\widehat{f}^{\widehat{\mathcal{A}}}$.
 
 For any nonempty set $\widetilde{\mathcal{A}} \subseteq [K]$,  let
 \[
  \widehat{f}^{\widetilde{\mathcal{A}}}= \argmin_{\theta \in \R^{n_0}} \bigg\{  \frac{1}{2n_{0}} \bigg\|\frac{1}{\vert \widetilde{\mathcal{A}} \vert  } \sum_{k \in \widetilde{\mathcal{A}} } \widetilde{P}^{n_0, n_k} y^{(k)}  - \theta \bigg\|_2^2   +\lambda_{\widetilde{\mathcal{A}}}\|D   \theta \|_1  \bigg\}, 
\]
where 
\begin{align}
  \lambda_{\widetilde{\mathcal{A}}} = C_{\lambda} \Big((s_0+1) \vert \widetilde{\mathcal{A}}\vert^2 \big(\sum_{k \in \widetilde{\mathcal{A}}}n_k^{-1}\big)^{-1} \Big)^{-1/2}, \nonumber
\end{align}
with an absolute constant $C_{\lambda} >0$.
By  \Cref{prop-ora}, it holds with absolute constants $C_1, c_1>0$ that 
 \begin{align}
 & \P \bigg\{ \big\|\widehat{f}^{\widetilde{\mathcal{A}}} -f   \big\|_{1/n_0}^2    
 > C_1 \Bigg(\max_{k \in \widetilde{\mathcal{A}}} \frac{\big\|\delta^{(k)}\big\|_2^2}{n_k} + \frac{    (s_{0}+1)  \big\{1+ \log \big(n_{0}/(s_0+1)  \big) \big\}}{\vert \widetilde{\mathcal{A}}\vert^2 \big(\sum_{k \in \widetilde{\mathcal{A}}}n_k^{-1}\big)^{-1}  }   \Bigg)  \bigg\}  \leq n_0^{-c_1}.\nonumber 
\end{align}
Denote 
\begin{align}
    \mathcal{E}_1  = & \Bigg\{ \exists \widetilde{\mathcal{A}} \subseteq [K]\colon  \widetilde{\mathcal{A}} \neq \emptyset \mbox{ and }  \big\|\widehat{f}^{\widetilde{\mathcal{A}}} -f   \big\|_{1/n_0}^2     
> C_1  \Bigg(\max_{k \in \widetilde{\mathcal{A}}} \frac{\big\|\delta^{(k)}\big\|_2^2}{n_k} 
\nonumber\\
& \hspace{8cm}+  
\frac{    (s_{0}+1)  \big\{1+\log \big(n_{0}/(s_0+1)  \big) \big\}}{\vert \widetilde{\mathcal{A}}\vert^2 \big(\sum_{k \in \widetilde{\mathcal{A}}}n_k^{-1}\big)^{-1} }   \Bigg)\Bigg\}. \nonumber
\end{align}
By a union bound argument, we obtain that 
 \begin{align}\label{coro1_1}
\P \big\{   \mathcal{E}_1  \big\}   \leq &
\sum_{ \substack{\widetilde{\mathcal{A}} \subseteq [K] \\   \mbox{with } \widetilde{\mathcal{A}} \neq \emptyset}} \P \bigg\{ \big\|\widehat{f}^{\widetilde{\mathcal{A}}} -f   \big\|_{1/n_0}^2    
 > C_1 \Bigg(\max_{k \in \widetilde{\mathcal{A}}} \frac{\big\|\delta^{(k)}\big\|_2^2}{n_k} + \frac{    (s_{0}+1)  \big\{1+ \log \big(n_{0}/(s_0+1)  \big) \big\}}{\vert \widetilde{\mathcal{A}}\vert^2 \big(\sum_{k \in \widetilde{\mathcal{A}}}n_k^{-1}\big)^{-1}  }   \Bigg)  \bigg\}\nonumber \\
\leq &  \sum_{ \substack{\widetilde{\mathcal{A}} \subseteq [K] \\   \mbox{with } \widetilde{\mathcal{A}} \neq \emptyset}} n_0^{-c_1} = \sum_{a=1}^{K} \sum_{\substack{  \widetilde{\mathcal{A}} \subseteq [K] \\ \mbox{with }\vert \widetilde{\mathcal{A}}  \vert = a }} n_{0}^{-c_1} = \sum_{a=1}^{K} \binom{K}{a} n_{0}^{-c_1} \leq 2^{K} n_{0}^{-c_1},
\end{align}
where the last inequality is based on the binomial formula, wherein for any $n \in \mathbb{N}^{*}$,  $2^n = \sum_{k=0}^{n} \binom{n}{k}$. Denote $\mathcal{E}_2  =  \big\{  \widehat{\mathcal{A}} \neq \mathcal{A}_{h^{*}} \big\}$.
By \Cref{theorem_detection_consistency}, it holds with an absolute constant $c_2 >0$ that 
\begin{align}\label{coro1_2}
  \P \big\{ \mathcal{E}_2 \big\} \leq K n_0^{-c_2} .
\end{align}
Denote
\begin{align}
 \mathcal{E}_3 = & \bigg\{  \big\|\widehat{f}^{\widehat{\mathcal{A}}} -f   
 \big\|_{1/n_0}^2    
\leq  C  \Bigg( \frac{    (s_{0}+1)  \big\{1+\log \big(n_{0}/(s_0+1)  \big) \big\}}{ \vert \mathcal{A}_{h^{*}}\vert^2 \big(\sum_{k \in \mathcal{A}_{h^{*}}}n_k^{-1}\big)^{-1}}
\nonumber\\
& \hspace{7cm}
+ (h^{*})^2   \wedge \frac{(s_0 +1) \big\{1+\log\big(n_0/(s_0+1)\big)\big\} }{n_0}     \Bigg)  \bigg\}. \nonumber
\end{align}  
Note that  $ \mathcal{E}_1^{c} \cap \mathcal{E}_2^{c} \subseteq \mathcal{E}_3$, then we  we can conclude with an absolute constant $c_3>0$ that 
\begin{align}
\P \big\{\mathcal{E}_3   \big\} \geq  & \P \big\{ \mathcal{E}_1^{c} \cap \mathcal{E}_2^{c} \big\} = 1-
\P \big\{ \mathcal{E}_1 \cup \mathcal{E}_2 \big\}  \nonumber\\
\geq & 1-  \P \big\{  \mathcal{E}_1 \big\} -  \P \big\{  \mathcal{E}_2\big\}   
\geq   1 - 2^K n_{0}^{-c_1} - Kn_0^{-c_2} \geq 1- 2^{K}n_0^{-c_3}, \nonumber 
\end{align}  
where the second inequality follows from \eqref{coro1_1} and \eqref{coro1_2}. This completes the proof for $\widehat{f}^{\widehat{\mathcal{A}}}$ .

 \medskip
 \noindent \textbf{Step 2.} In this step, we focus on the estimator $\widetilde{f}^{\widehat{\mathcal{A}}}$.

 For any nonempty set $\widetilde{\mathcal{A}} \subseteq [K]$,  let
 \[
  \widetilde{f}^{\widetilde{\mathcal{A}}}= \argmin_{\theta \in \R^{n_0}} \bigg\{  \frac{1}{2n_{0}} \bigg\|\frac{1}{\vert \widetilde{\mathcal{A}} \vert  } \sum_{k \in \widetilde{\mathcal{A}} } \widetilde{P}^{n_0, n_k} y^{(k)}  - \theta \bigg\|_2^2   +\widetilde{\lambda}_{\widetilde{\mathcal{A}}}\|D   \theta \|_0  \bigg\}, 
\]
where 
\begin{align}
   \widetilde{\lambda}_{\widetilde{\mathcal{A}}} = & C_{\widetilde{\lambda}}\frac{ 1+ \log\big(n_0 / (s_{0}+1)) \big)}{\vert \widetilde{\mathcal{A}}\vert^2 \big(\sum_{k \in \widetilde{\mathcal{A}}}n_k^{-1}\big)^{-1} }, \nonumber
 \end{align}
with an absolute constant $C_{\widetilde{\lambda}} >0$.
By  \Cref{prop-ora}, it holds with absolute constants $C_4, c_4 > 0$ that  
 \begin{align}
 \P \bigg\{ \big\|\widetilde{f}^{\widetilde{\mathcal{A}}} -f   \big\|_{1/n_0}^2    
>  C_4 \Bigg(\max_{k \in \widetilde{\mathcal{A}}} \frac{\big\|\delta^{(k)}\big\|_2^2}{n_k} + \frac{    (s_{0}+1) \big\{1+ \log \big(n_{0}/(s_0+1)  \big)\big\}}{\vert \widetilde{\mathcal{A}}\vert^2 \big(\sum_{k \in \widetilde{\mathcal{A}}}n_k^{-1}\big)^{-1}  }   \Bigg)   \bigg\} \leq n_0^{-c_4}.\nonumber 
\end{align}
Denote 
\begin{align}
    \mathcal{E}_4  = & \Bigg\{ \exists \widetilde{\mathcal{A}} \subseteq [K]\colon~ \widetilde{\mathcal{A}} \neq \emptyset \mbox{ and }\big\|\widetilde{f}^{\widetilde{\mathcal{A}}} -f   \big\|_{1/n_0}^2    
> C_4  \Bigg(\max_{k \in \widetilde{\mathcal{A}}} \frac{\big\|\delta^{(k)}\big\|_2^2}{n_k}
\nonumber\\
& \hspace{8cm}
+ \frac{    (s_{0}+1)  \big\{1+\log \big(n_{0}/(s_0+1)  \big) \big\}}{\vert \widetilde{\mathcal{A}}\vert^2 \big(\sum_{k \in \widetilde{\mathcal{A}}}n_k^{-1}\big)^{-1} }   \Bigg)\Bigg\}. \nonumber
\end{align}
By a union bound argument, we obtain that 
 \begin{align}\label{coro2_1}
\P \big\{   \mathcal{E}_4  \big\}   \leq &
\sum_{\widetilde{\mathcal{A}} \subseteq [K]} \P \bigg\{ \big\|\widetilde{f}^{\widetilde{\mathcal{A}}} -f   \big\|_{1/n_0}^2    
 > C_4 \Bigg(\max_{k \in \widetilde{\mathcal{A}}} \frac{\big\|\delta^{(k)}\big\|_2^2}{n_k} + \frac{    (s_{0}+1)  \big\{1+ \log \big(n_{0}/(s_0+1)  \big) \big\}}{\vert \widetilde{\mathcal{A}}\vert^2 \big(\sum_{k \in \widetilde{\mathcal{A}}}n_k^{-1}\big)^{-1}  }   \Bigg)  \bigg\}\nonumber \\
\leq &  \sum_{ \substack{\widetilde{\mathcal{A}} \subseteq [K] \\   \mbox{with } \widetilde{\mathcal{A}} \neq \emptyset}} n_0^{-c_4} = \sum_{a=1}^{K} \sum_{\substack{  \widetilde{\mathcal{A}} \subseteq [K] \\ \mbox{with }\vert \widetilde{\mathcal{A}}  \vert = a }} n_{0}^{-c_4} = \sum_{a=1}^{K} \binom{K}{a} n_{0}^{-c_4} \leq 2^{K} n_{0}^{-c_4},
\end{align}
where the last inequality is based on the binomial formula, wherein for any $n \in \mathbb{N}$,  $2^n = \sum_{k=0}^{n} \binom{n}{k}$. 

Denote
\begin{align}
 \mathcal{E}_5 = & \bigg\{  \big\|\widetilde{f}^{\widehat{\mathcal{A}}} -f   
 \big\|_{1/n_0}^2    
\leq  C  \Bigg( \frac{    (s_{0}+1)  \big\{1+\log \big(n_{0}/(s_0+1)  \big) \big\}}{ \vert \mathcal{A}_{h^{*}}\vert^2 \big(\sum_{k \in \mathcal{A}_{h^{*}}}n_k^{-1}\big)^{-1}}
\nonumber\\
& \hspace{7cm}
+ (h^{*})^2   \wedge \frac{(s_0 +1) \big\{1+\log\big(n_0/(s_0+1)\big)\big\} }{n_0}     \Bigg)  \bigg\}. \nonumber
\end{align}  
Note that  $ \mathcal{E}_4^{c} \cap \mathcal{E}_2^{c} \subseteq \mathcal{E}_5$, then we can conclude with an absolute constant $c_5>0$ that 
\begin{align}
\P \big\{\mathcal{E}_5   \big\} \geq  & \P \big\{ \mathcal{E}_4^{c} \cap \mathcal{E}_2^{c} \big\} = 1-
\P \big\{ \mathcal{E}_4 \cup \mathcal{E}_2 \big\} 
\nonumber \\
\geq & 1-  \P \big\{  \mathcal{E}_4 \big\} -  \P \big\{  \mathcal{E}_2\big\}   
\geq   1 - 2^K n_{0}^{-c_4} - Kn_0^{-c_2} \geq 1- 2^{K}n_0^{-c_5}, \nonumber 
\end{align}  
where the second inequality follows from \eqref{coro2_1} and \eqref{coro1_2}. This completes the proof for $\widetilde{f}^{\widehat{\mathcal{A}}}$.
\end{proof}

\subsection{Proofs of Theorem \ref{theorem_minimax}}\label{app_3_minimax}
\begin{proof}[Proof of \Cref{theorem_minimax}]  
Let 
\[
     \Theta^{\prime}_{s_0, \mathcal{A}_h} = \bigg\{\theta = \big(f^{\top}, (f^{(k_1)})^{\top}, \dots,  (f^{(k_a)})^{\top} \big)^{\top}\colon \|D f\|_{0} \leq s_0,\,  s_0 \geq 4 \bigg\}.
\]
Since $ \Theta^{\prime}_{s_0, \mathcal{A}_h} \subseteq  \Theta_{s_0, \mathcal{A}_h}$, to prove \eqref{minimax_bound}, we only need to prove that 
\[
    \inf_{\widehat{f} \in \R^{n_0}} \sup_{\theta \in  \Theta^{\prime}_{s_0, \mathcal{A}_h}} \P \bigg\{ \| \widehat{f} - f\|^2_{1/n_0}  \geq  C \bigg(  \frac{s_0 \log(n_0 /s_0)}{\sum_{k \in \mathcal{A}_h} n_k}  + h^2 \wedge \frac{s_0 \log(n_0 /s_0)}{n_0}  \bigg) \bigg\} \geq \frac{1}{2}.
\]

This proof consists of two steps. In \textbf{Step 1}, we examine the scenario where 
\[
 \frac{s_0\log(n_0/s_0)} { \sum_{k \in \mathcal{A}_h} n_k  } \geq   h^2 \wedge \frac{s_0\log(n_0/s_0)} {n_{0}}.
\]
 Subsequently, we address the scenario where 
 \[
 \frac{s_0\log(n_0/s_0)} { \sum_{k \in \mathcal{A}_h} n_k  } <   h^2 \wedge \frac{s_0\log(n_0/s_0)} {n_{0}},
\] 
in \textbf{Step 2}.

Without loss of generality, let $n_0$ be even. 
Let $l$ be the largest nonzero even number such that $ l \leq s_0/2$. Since $s_0 \geq 4$ in $\Theta^{\prime}_{s_0, \mathcal{A}_h}$, such $l$ exists. For any $k \in \mathcal{A}_h$, denote
\[
\lambda_{1} \big((P^{n_k, n_0})^{\top} P^{n_k, n_0} \big) \geq \cdots \geq \lambda_{n_0} \big((P^{n_k, n_0})^{\top} P^{n_k, n_0} \big),
\]
as the eigenvalues of $(P^{n_k, n_0})^{\top} P^{n_k, n_0}$, with the alignment operator $P^{n_k, n_0}$ defined in \eqref{def-P}. 
For any $k \in \mathcal{A}_h$, 
 by \Cref{lemma-P} and $n_k \geq n_0$, we have that 
\begin{align}\label{lemma_eigen_1}
     \lambda_1 \big((P^{n_k, n_0})^{\top} P^{n_k, n_0} \big) \leq \Big\lceil \frac{n_k}{n_0} \Big\rceil 
 \leq \frac{2n_k}{n_0}.
\end{align}
Then we can conclude that for any $\widehat{f} \in \R^{n_0}$, 
 \begin{align}\label{minimax_1}
    \sum_{k \in \mathcal{A}_h} \Big\| P^{n_k, n_0}  \widehat{f} - P^{n_k, n_0} f \Big\|_2^2  \leq \frac{\sum_{k \in \mathcal{A}_h} (2n_k)}{n_0}\big\| \widehat{f} - f \big\|_{2}^2. 
 \end{align}

For any $k \in \mathcal{A}_h$, if $n_k = n_0$, by \Cref{lemma-P}, we have that 
\begin{align}\label{lemma_eigen_2_1}
     \lambda_{n_0} \big((P^{n_k, n_0})^{\top} P^{n_k, n_0} \big) = 1  > \frac{n_k}{2 n_0};
\end{align}
if $n_0 < n_k < 2n_0$, then by \Cref{lemma-P}, we have that 
\begin{align}\label{lemma_eigen_2_2}
     \lambda_{n_0} \big((P^{n_k, n_0})^{\top} P^{n_k, n_0} \big) \geq \Big\lceil \frac{n_k}{n_0}  \Big\rceil -1  =1 > \frac{n_k}{2 n_0};
\end{align}
and if $n_k \geq 2n_0$, then by \Cref{lemma-P}, we have that 
\begin{align}\label{lemma_eigen_2_3}
     \lambda_{n_0} \big((P^{n_k, n_0})^{\top} P^{n_k, n_0} \big) \geq \Big\lceil \frac{n_k}{n_0}  \Big\rceil -1 \geq  \frac{n_k - n_0}{n_0}   \geq \frac{n_k}{2 n_0}.
\end{align}
Combining \eqref{lemma_eigen_2_1}, \eqref{lemma_eigen_2_2} and \eqref{lemma_eigen_2_3}, for any $k \in \mathcal{A}_h$, it holds that 
\begin{align}\label{lemma_eigen_2}
     \lambda_{n_0} \big((P^{n_k, n_0})^{\top} P^{n_k, n_0} \big)  \geq \frac{n_k}{2 n_0}.
\end{align}

\medskip
\noindent \textbf{Step 1.}  
In this step, we consider the scenario where
\[
 \frac{s_0\log(n_0/s_0)} { \sum_{k \in \mathcal{A}_h} n_k   } \geq   h^2 \wedge \frac{s_0\log(n_0/s_0)} {n_{0}},
\]
and construct the ideal case with $f^{(k)} = P^{n_k, n_0} f$, for each $k \in \mathcal{A}_h$. This allows us to derive the lower bound in \eqref{minimax_bound}. 

Remember that $\mathcal{A}_h$ defined in \eqref{def-mathcal-A-h} with cardinality $a$, is denoted as $\mathcal{A}_h  = \{ k_1, \dots, k_a \}$, then define the parameter space  $\Theta_{s_0} $ as 
\begin{align}
     \Theta_{s_0} = \Big\{  \theta = \big(f^{\top}, (P^{n_{k_1}, n_0}f)^{\top}, \dots,  (P^{n_{k_a}, n_0}f)^{\top} \big)^{\top}\colon\|D f\|_{0} \leq s_0\mbox{ and } s_0 \geq 4 \Big\}. \nonumber
\end{align} 
Since $\Theta_{s_0} \subseteq  \Theta^{\prime}_{s_0, \mathcal{A}_h}$, then we have with an absolute constant $C_1 >0$ that
\begin{align}
 &\inf_{\widehat{f} \in \R^{n_0}} \sup_{\theta \in  \Theta^{\prime}_{s_0, \mathcal{A}_h}} \P \bigg\{ \big\| \widehat{f} - f \big\|_{1/n_0}^2  \geq  C_1\frac{s_0 \log(n_0 /s_0)}{\sum_{k \in \mathcal{A}_h} n_k}  \bigg\} \nonumber\\
  \geq &\inf_{\widehat{f} \in \R^{n_0}} \sup_{\theta \in  \Theta_{s_0}} \P \bigg\{ \big\| \widehat{f} - f \big\|_{1/n_0}^2  \geq  C_1\frac{s_0 \log(n_0 /s_0)}{\sum_{k \in \mathcal{A}_h} n_k }  \bigg\} \nonumber\\
  \geq  & \inf_{\widehat{f} \in \R^{n_0}} \sup_{\theta \in  \Theta_{s_0}} \P \bigg\{ \sum_{k \in \mathcal{A}_h} \Big\| P^{n_k, n_0}  \widehat{f} - P^{n_k, n_0} f \Big\|_2^2  \geq  2C_1  s_0 \log(n_0 /s_0) \bigg\} \nonumber\\ 
 \geq  & \inf_{\widehat{f}^{(k)} \in \R^{n_{k}}, \forall k \in \mathcal{A}_h}  \sup_{\theta \in  \Theta_{s_0}} \P \bigg\{ \sum_{k \in \mathcal{A}_h} \big\|  \widehat{f}^{(k)} -  f^{(k)} \big\|_2^2  \geq  2C_1  s_0 \log(n_0 /s_0) \bigg\}, \nonumber
\end{align}
where the second inequality follows from \eqref{minimax_1}.
Thus, to prove \eqref{minimax_bound}, it suffices to prove that 
\begin{align}\label{minimax_1_1}
   \inf_{\widehat{f}^{(k)} \in \R^{n_{k}}, \forall k \in \mathcal{A}_h}  \sup_{\theta \in  \Theta_{s_0}} \P \bigg\{ \sum_{k \in \mathcal{A}_h} \big\|  \widehat{f}^{(k)} -  f^{(k)} \big\|_2^2  \geq 2C_1  s_0 \log(n_0 /s_0) \bigg\} \geq \frac{1}{2}.
\end{align}
Let $l\geq 2$ and 
\begin{align}\label{minimax_1_2}
        \mathcal{B} = \big\{ f \in \{-1, 0, 1 \}^{n_0}\colon\| f\|_0 \leq l \big\}.
    \end{align}
Then by $1 \leq l/2 \leq s_0/4 < n_0/3$ and  Lemma 4 in \cite{raskutti2011minimax}, there exists $\widetilde{\mathcal{B}} \subseteq \mathcal{B}$ such that 
\begin{align}\label{minimax_1_3}
    \log\big(  \vert \widetilde{\mathcal{B}}  \vert  \big ) \geq \frac{l}{2} \log\bigg( \frac{n_0 - l}{l/2}\bigg) \quad \mbox{and} \quad \big\|f^1 -f^2 \big\|_2^2 \geq l/2, \quad \forall f^1\neq f^2 \in \widetilde{\mathcal{B}}. 
\end{align}
Let $\epsilon_1> 0$ be specified later. Define the parameter space $\widetilde{\Theta}_{\epsilon_1, 0}$ as
\begin{align}\label{minimax_1_4}
    \widetilde{\Theta}_{\epsilon_1, 0} = \Big\{ \big(f^{\top}, (P^{n_{k_1}, n_0}f)^{\top}, \dots,  (P^{n_{k_a}, n_0}f)^{\top} \big)^{\top}\colon f \in  2 \sqrt{2n_0/ \big(\sum_{k \in \mathcal{A}_h} n_k + 2n_0\big)} \epsilon_1 \widetilde{\mathcal{B}} \Big\}. 
\end{align}
Combining \eqref{minimax_1_3} and \eqref{minimax_1_4}, we have that 
\begin{align}\label{minimax_1_5}
   \log\big( \vert  \widetilde{\Theta}_{\epsilon_1, 0} \vert \big)  \geq \frac{l}{2} \log\bigg( \frac{n_0 - l}{l/2}\bigg),  
\end{align} 
and for any $\theta^{1}\neq \theta^{2} \in \widetilde{\Theta}_{\epsilon_1, 0}$,
\begin{align}
   \| \theta^{1} - \theta^{2} \|_2^2 \geq \frac{4ln_0\epsilon_1^2}{\sum_{k \in \mathcal{A}_h} n_k + 2n_0} \bigg( 1 + \sum_{k \in \mathcal{A}_h}  \lambda_{n_0} \big((P^{n_k, n_0})^{\top} P^{n_k, n_0} \big)  \bigg)   \geq  2\epsilon_1^2l, \nonumber
\end{align}
where the last inequality follows from \eqref{lemma_eigen_2}.

For any $\theta \in \widetilde{\Theta}_{\epsilon_1, 0}$, we consider comparing the measure $\mathcal{P}_{\theta} = \mathcal{N}\big( \theta, I_{\sum_{k \in \mathcal{A}_h \cup \{ 0\}} n_k} \big)$ against $\mathcal{P}_{0} = \mathcal{N}(0, I_{\sum_{k \in \mathcal{A}_h\cup \{ 0\}}n_k })$. Then we have that 
\begin{align}
    D_{\mathrm{KL}} \big( \mathcal{P}_{\theta}, \mathcal{P}_{0}\big) = \|\theta \|_2^2 \leq \frac{8ln_0\epsilon_1^2}{\sum_{k \in \mathcal{A}_h} n_k + 2n_0} \bigg( 1 + \sum_{k \in \mathcal{A}_h}  \lambda_1 \big((P^{n_k, n_0})^{\top} P^{n_k, n_0} \big)  \bigg) \leq 16\epsilon_1^2 l, \nonumber 
\end{align}
where the first inequality follows from \eqref{minimax_1_2} and \eqref{minimax_1_4}, and the last inequality follows from \eqref{lemma_eigen_1}.
Let $\epsilon_1^2 = \alpha   \log\big( \vert  \widetilde{\Theta}_{\epsilon_1, 0} \vert \big)/ (16 l)$ with $\alpha>0$ to be defined later, then it holds that
\begin{align}
   \frac{1}{\vert  \widetilde{\Theta}_{\epsilon_1, 0}\vert} \sum_{ \theta  \in  \widetilde{\Theta}_{\epsilon_1, 0}} D_{\mathrm{KL}} \big( \mathcal{P}_{\theta},  \mathcal{P}_{0}\big)  \leq \alpha  \log\big( \vert   \widetilde{\Theta}_{\epsilon_1, 0} \vert \big). \nonumber
\end{align}
By Theorem 2.5 in \cite{tsybakov2009introduction}, we have that 
\begin{align}
\inf_{\widehat{f}^{(k)} \in \R^{n_{k}}, \forall k \in \mathcal{A}_h}  \sup_{\theta \in  \widetilde{\Theta}_{\epsilon_1, 0}}  \P \bigg\{ \sum_{k \in \mathcal{A}_h} \big\|  \widehat{f}^{(k)} -  f^{(k)} \big\|_2^2  \geq \epsilon_1^2 l \bigg\} \geq \frac{\sqrt{\vert  \widetilde{\Theta}_{\epsilon_1, 0} \vert }}{ 1+ \sqrt{\vert  \widetilde{\Theta}_{\epsilon_1, 0} \vert}} \bigg( 1 - \alpha - \sqrt{\frac{2\alpha}{\log\big( \vert  \widetilde{\Theta}_{\epsilon_1, 0} \vert \big) }}\bigg). \nonumber
\end{align}
Choosing $\alpha > 0$ to be a small enough constant, by \eqref{minimax_1_5},  we obtain that there exists an absolute constant $C_2 >0$ such that 
\begin{align}
    \epsilon_1^2 l = \alpha \log\big( \vert  \widetilde{\Theta}_{\epsilon_1, 0}\vert \big)/ 16  \geq C_2 s_0   \log(n_0/s_0). \nonumber
\end{align}
Since $\alpha>0$ is a small enough constant, it holds that 
\begin{align}
   & \inf_{\widehat{f}^{(k)} \in \R^{n_{k}}, \forall k \in \mathcal{A}_h}  \sup_{\theta \in  \widetilde{\Theta}_{\epsilon_1, 0}}  \P \Big\{ \sum_{k \in \mathcal{A}_h} \big\|  \widehat{f}^{(k)} -  f^{(k)} \big\|_2^2  \geq  C_2 s_0 \log(n_0 /s_0) \Big\}  \geq \frac{1}{2}, \nonumber
\end{align}
which proves \eqref{minimax_1_1}.

\medskip
\noindent \textbf{Step 2.}
In this step, we deal with the scenario where
 \[
 \frac{s_0\log(n_0/s_0)} { \sum_{k \in \mathcal{A}_h} n_k  } <   h^2 \wedge \frac{s_0\log(n_0/s_0)} {n_{0}}.
\] 
Given a small enough absolute constant $C_h>0$,  we decompose our analysis into two distinct cases and construct the least informative scenario where for any $k \in \mathcal{A}_h$, $f^{(k)} = 0$ to prove the lower bound as shown in \eqref{minimax_bound}. The first case is $h^2 \leq C_h s_0\log(n_0/s_0) / n_{0}$, which is addressed in \textbf{Step 2.1}. Conversely, the second case is $h^2 > C_h s_0\log(n_0/s_0) / n_{0}$, which is examined in \textbf{Step 2.2}. 

\medskip
\noindent \textbf{Step 2.1.} In this step, we consider the scenario where
 \begin{align}\label{minimax_2.1_0}
 \frac{s_0\log(n_0/s_0)} { \sum_{k \in \mathcal{A}_h} n_k  } <   h^2 \wedge \frac{s_0\log(n_0/s_0)} {n_{0}} \quad \mbox{and} \quad h^2 \leq C_h\frac{s_0\log(n_0/s_0)}{n_{0}},
 \end{align}
 with a small enough absolute constant $C_h >0$. 
 
Let $\epsilon_2 = h \sqrt{\alpha n_0 / l}$ with a small enough absolute constant $\alpha>0$. Define the parameter space $\widetilde{\Theta}_{\epsilon_2}$ as
\begin{align}\label{minimax_2.1_1}
     \widetilde{\Theta}_{\epsilon_2}= \Big\{  \theta = \big(f^{\top}, (f^{(k_1)})^{\top}, \dots,  (f^{(k_a)})^{\top} \big)^{\top}\colon  f \in \epsilon_2\widetilde{\mathcal{B}}, f^{(k)} = 0, \forall k \in \mathcal{A}_h \Big\}.
\end{align} 
with $\widetilde{\mathcal{B}}$ defined in \textbf{Step 1}.
Since $ \widetilde{\Theta}_{\epsilon_2} \subseteq  \Theta^{\prime}_{s_0, \mathcal{A}_h}$, we have with an absolute $C_3 >0$ that
\begin{align}
 \inf_{\widehat{f} \in \R^{n_0}} \sup_{\theta \in  \Theta^{\prime}_{s_0, \mathcal{A}_h}} \P \Big\{ \big\| \widehat{f} - f \big\|_{1/n_0}^2  \geq  C_3 h^2  \Big\}  \geq  \inf_{\widehat{f} \in \R^{n_0}} \sup_{\theta \in \widetilde{\Theta}_{\epsilon_2}} \P \Big\{ \big\|\widehat{f} - f  \big\|_2^2  \geq   C_3n_0 h^2 \Big\}. \nonumber
\end{align}
Thus to prove \eqref{minimax_bound}, it suffices to prove that 
\begin{align}\label{minimax_2.1_2}
    \inf_{\widehat{f} \in \R^{n_0}} \sup_{\theta \in \widetilde{\Theta}_{\epsilon_2}} \P \Big\{ \big\|\widehat{f} - f  \big\|_2^2  \geq   C_3n_0 h^2 \Big\}  \geq \frac{1}{2}.
\end{align}

Combining \eqref{minimax_1_3} and \eqref{minimax_2.1_1}, we have that
\begin{align}\label{minimax_2.1_3}
   \log\big( \vert   \widetilde{\Theta}_{\epsilon_2} \vert \big)  \geq \frac{l}{2} \log\bigg( \frac{n_0 - l}{l/2}\bigg) \quad \mbox{and} \quad    \| \theta^{1} - \theta^{2} \|_2^2 \geq  \epsilon_2^2 l /2, \quad  \forall  \theta^{1}\neq \theta_{2} \in \widetilde{\Theta}_{\epsilon_2}.  
\end{align}
For any $\theta \in \widetilde{\Theta}_{\epsilon_2}$, we consider comparing  the measure $\mathcal{P}_{\theta} = \mathcal{N}\big(\theta, I_{ \sum_{k \in \mathcal{A}_h \cup \{ 0\}} n_k }\big)$ against $\mathcal{P}_{0} = \mathcal{N} \big(0, I_{\sum_{k \in \mathcal{A}_h \cup \{ 0\}} n_k} \big) $. It holds that 
\begin{align}
    D_{\mathrm{KL}} \big( \mathcal{P}_{\theta}, \mathcal{P}_{0}\big) = \|\theta\|_2^2 \leq  \epsilon_2^2 l, \nonumber
\end{align}
where the first inequality follows from \eqref{minimax_1_2} and \eqref{minimax_2.1_1}. Then  we obtain that 
\begin{align}
   \frac{1}{\vert  \widetilde{\Theta}_{\epsilon_2} \vert} \sum_{\theta \in  \widetilde{\Theta}_{\epsilon_2}} D_{\mathrm{KL}} \big( \mathcal{P}_{\theta},  \mathcal{P}_{0}\big)  \leq \epsilon_2^2 l = \alpha n_0h^2 \leq \alpha  C_h s_0\log(n_0/s_0) \leq \alpha  \log\big( \vert   \widetilde{\Theta}_{\epsilon_2} \vert \big), \nonumber
\end{align}
where the first equality is due to the choice of $\epsilon_2$, the second inequality follows form \eqref{minimax_2.1_0}, and the last inequality follows from \eqref{minimax_2.1_3} and $C_h >0$ is an small enough absolute constant .
Then by Theorem 2.5 in \cite{tsybakov2009introduction}, we have that  
\begin{align}
 \inf_{\widehat{f} \in \R^{n_0}}  \sup_{\theta \in \widetilde{\Theta}_{\epsilon_2}}  \P \Big\{ \big\| \widehat{f} - f \big\|_2^2  \geq \epsilon_2^2 l /2  \Big\} \geq \frac{\sqrt{\vert  \widetilde{\Theta}_{\epsilon_2}\vert }}{ 1+ \sqrt{\vert   \widetilde{\Theta}_{\epsilon_2} \vert}} \bigg( 1 - \alpha - \sqrt{\frac{2\alpha}{\log\big( \vert  \widetilde{\Theta}_{\epsilon_2}\vert \big) }}\bigg)    \geq \frac{1}{2}, \nonumber
\end{align}
where $\epsilon_2^2 l /2= \alpha n_0 h^2/2$ and the last inequality follows from that $\alpha>0$ is a small enough constant. This proves \eqref{minimax_2.1_2}.

\medskip
\noindent \textbf{Step 2.2.}
In this step, we focus on the scenario 
\[
\frac{s_0\log(n_0/s_0)}{ \sum_{k \in \mathcal{A}_h} n_k} <  h^2 \wedge \frac{s_0\log(n_0/s_0)}{n_{0}} \quad \mbox{and}  \quad h^2 > C_h \frac{s_0\log(n_0/s_0)}{n_{0}},
\]
with a small enough absolute constant $C_h >0$. 

Let $\epsilon_3  = \sqrt{C_h \alpha s_0\log(n_0/s_0) /l}$ with a small enough absolute constant $\alpha>0$. Define the parameter space $\widetilde{\Theta}_{\epsilon_3}$ as 
\begin{align}\label{minimax_2.2_1}
     \widetilde{\Theta}_{\epsilon_3} = \Big\{  \theta = \big(f^{\top}, (f^{(k_1)})^{\top}, \dots,  (f^{(k_a)})^{\top} \big)^{\top} \colon  f \in \epsilon_3 \widetilde{\mathcal{B}}, f^{(k)} = 0, \forall k \in \mathcal{A}_h \Big\}, 
\end{align} 
with $\widetilde{\mathcal{B}}$ defined in \textbf{Step 1}.
Since $ \widetilde{\Theta}_{\epsilon_3} \subseteq  \Theta^{\prime}_{s_0, \mathcal{A}_h}$, we have with an absolute $C_4 >0$ that
\begin{align}
 \inf_{\widehat{f} \in \R^{n_0}} \sup_{\theta \in  \Theta^{\prime}_{s_0, \mathcal{A}_h}} \P \bigg\{ \big\| \widehat{f} - f \big\|_{1/n_0}^2  \geq  C_4\frac{s_0 \log(n_0/s_0)}{n_0}\bigg\}  \geq  \inf_{\widehat{f} \in \R^{n_0}} \sup_{\theta \in  \widetilde{\Theta}_{\epsilon_3}} \P \Big\{ \big\|\widehat{f} - f  \big\|_2^2  \geq   C_4s_0 \log(n_0/s_0) \Big\}. \nonumber
\end{align}
Thus to prove \eqref{minimax_bound}, it suffices to prove that 
\begin{align}\label{minimax_2.2_2}
\inf_{\widehat{f} \in \R^{n_0}} \sup_{\theta \in  \widetilde{\Theta}_{\epsilon_3}} \P \Big\{ \big\|\widehat{f} - f  \big\|_2^2  \geq   C_4 s_0 \log(n_0/s_0) \Big\} \geq \frac{1}{2}.
\end{align}

Combining \eqref{minimax_1_3} and \eqref{minimax_2.2_1}, we derive that
\begin{align}\label{minimax_2.2_3}
   \log\big( \vert   \widetilde{\Theta}_{\epsilon_3} \vert \big)  \geq \frac{l}{2} \log\bigg( \frac{n_0 - l}{l/2}\bigg) \quad \mbox{and} \quad    \| \theta^{1} - \theta^{2} \|_2^2 \geq  \epsilon_3^2 l /2, \quad  \forall  \theta^{1}\neq \theta_{2} \in \widetilde{\Theta}_{\epsilon_3}.  
\end{align}
For any $\theta \in \widetilde{\Theta}_{\epsilon_3}$, we consider comparing  the measure $\mathcal{P}_{\theta} = \mathcal{N} \big(\theta, I_{ \sum_{k \in \mathcal{A}_h \cup \{ 0\}} n_k }\big)$ against $\mathcal{P}_{0} = \mathcal{N}(0, I_{ \sum_{k \in \mathcal{A}_h \cup \{ 0\}} n_k }\big) $. It holds that 
\begin{align}
    D_{\mathrm{KL}} \big( \mathcal{P}_{\theta}, \mathcal{P}_{0}\big) = \|\theta\|_2^2 \leq  \epsilon_3^2 l, \nonumber
\end{align}
where the first inequality follows from \eqref{minimax_1_2} and \eqref{minimax_2.2_1}. Then we obtain that 
\begin{align}
   \frac{1}{\vert  \widetilde{\Theta}_{\epsilon_3} \vert} \sum_{\theta \in  \widetilde{\Theta}_{\epsilon_3}} D_{\mathrm{KL}} \big( \mathcal{P}_{\theta},  \mathcal{P}_{0}\big)  \leq \epsilon_3^2 l  = C_h \alpha s_0\log(n_0/s_0)\leq \alpha  \log\big( \vert   \widetilde{\Theta}_{\epsilon_2} \vert \big), \nonumber
\end{align}
where the first equality is due to the choice of $\epsilon_3$, and the last inequality follows from \eqref{minimax_2.2_3} and $C_h >0$ is an small enough absolute constant .
Then by Theorem 2.5 in \cite{tsybakov2009introduction}, we can conclude that 
\begin{align}
 \inf_{\widehat{f} \in \R^{n_0}}  \sup_{\theta \in \widetilde{\Theta}_{\epsilon_3}}  \P \big\{ \| \widehat{f} - f \|_2^2  \geq \epsilon_3^2 l /2  \big\} \geq \frac{\sqrt{\vert  \widetilde{\Theta}_{\epsilon_3}\vert }}{ 1+ \sqrt{\vert   \widetilde{\Theta}_{\epsilon_3} \vert}} \bigg( 1 - \alpha - \sqrt{\frac{2\alpha}{\log\big( \vert  \widetilde{\Theta}_{\epsilon_3}\vert \big) }}\bigg)    \geq \frac{1}{2}, \nonumber
\end{align}
where $\epsilon_3^2 l /2 = C_h \alpha s_0\log(n_0/s_0)/2$ and the last inequality follows from that $\alpha$ is a small enough constant. This proves \eqref{minimax_2.2_2}. 

\medskip
Combining \eqref{minimax_1_1}, \eqref{minimax_2.1_2} and \eqref{minimax_2.2_2}, we complete the proof.
\end{proof}

\section[]{Technical details of results in \Cref{{sec-extensions}}}\label{app_4}

The proofs of \Cref{theorem_l_0-affine}, \Cref{theorem_l_0-1_all} and \Cref{theorem_l_0-K_all} can be found in Appendices \ref{app-affine}, \ref{app-tl-1-all} and \ref{app-tl-K-all}, respectively.

\subsection{Proof of Proposition \ref{theorem_l_0-affine}}\label{app-affine}

The proof of \Cref{theorem_l_0-affine} is in \Cref{app-subsec-affine} with all necessary auxiliary results in \Cref{app-affine-lemmas}.

\subsubsection{Proof of Proposition \ref{theorem_l_0-affine}}\label{app-subsec-affine}

\begin{proof}[Proof of \Cref{theorem_l_0-affine}]
This proof consists of four steps. In \textbf{Step 1}, we decompose our target quantity into several terms. We then deal with these terms individually in \textbf{Step 2} and \textbf{Step 3}. In \textbf{Step 4}, we gather all the pieces and conclude the proof.

\medskip
\noindent\textbf{Step 1.}
    It directly follows from the definition of  $\widetilde{f}^{\widetilde{A}} $ that
\begin{align}
&  \frac{1}{2n_0} \big\|\widetilde{A} y^{(1)} -  \widetilde{f}^{\widetilde{A}}   \big\|_2^2 +\widetilde{\lambda}_{{\widetilde{A}}} \|D \widetilde{f}^{\widetilde{A}}  \|_0    \leq \frac{1}{2n_0} \big\|\widetilde{A}  y^{(1)}   -   f  \big\|_2^2 +\widetilde{\lambda}_{\widetilde{A}} \|D  f \|_0.  \nonumber
\end{align}
Given that $y^{(1)} = f^{(1)} + \epsilon^{(1)}$ with $f^{(1)} = A f + \delta^A$, we derive that 
\begin{align}
  \frac{1}{2n_{0}} \big\|\widetilde{f}^{\widetilde{A}}  - \widetilde{A} A f  \big\|_2^2  
\leq  &    \frac{1}{2n_{0}} \big\|f - \widetilde{A} A f  \big\|_2^2  +   \frac{1}{n_0}\widetilde{\epsilon}^{\top}\big( \widetilde{f}^{\widetilde{A}}   - f \big)   + \widetilde{\lambda}_{\widetilde{A}} \|D   f \|_0  \nonumber\\
& \hspace{0.5cm}-  \widetilde{\lambda}_{\widetilde{A}}   \|D   \widetilde{f}^{\widetilde{A}}  \|_0 + \frac{1}{n_0}\big(   \widetilde{A} \delta^A\big)^{\top} \big( \widetilde{f}^{\widetilde{A}}   - f \big), \nonumber 
\end{align}
with $\widetilde{\epsilon} =\widetilde{A} \epsilon^{(1)} \in \R^{n_0}$. Since $\widetilde{A}A=I_{n_0}$, it holds that
\begin{align}\label{theorem-affine_1}
 \frac{1}{2n_{0}} \big\|\widetilde{f}^{\widetilde{A}}  - f  \big\|_2^2   
\leq  &   \frac{1}{n_0}\widetilde{\epsilon}^{\top} \big( \widetilde{f}^{\widetilde{A}}   - f \big) 
  + \widetilde{\lambda}_{\widetilde{A}}  \|D   f \|_0 - \widetilde{\lambda}_{\widetilde{A}}   \|D   \widetilde{f}^{\widetilde{A}}  \|_0
+ \frac{1}{n_0}\big(   \widetilde{A} \delta^A\big)^{\top} \big( \widetilde{f}^{\widetilde{A}}   -   f\big) \nonumber\\
= &  (I.1) + (I.2) + (I. 3) + (II) =(I) + (II).
\end{align}

\medskip
\noindent\textbf{Step 2.} In this step, we consider the term $(I)$ in \eqref{theorem-affine_1}. 

Let the set $\mathcal{S}$ be defined in \eqref{def-S}  with  cardinality $s_0$ and the set $\widetilde{\mathcal{S}}$ be defined as
\begin{align}\label{theorem-affine-3}
\widetilde{\mathcal{S}} = \big\{i \in [n_{0}-1]\colon \widetilde{f}^{\widetilde{A}} _i \neq \widetilde{f}^{\widetilde{A}} _{i+1}  \big\}  = \big\{i \in [n_{0}-1]\colon (D\widetilde{f}^{\widetilde{A}} )_i \neq 0  \big\}. 
\end{align}
Let the orthogonal projection operator $P^{\widetilde{\mathcal{S}} \cup \mathcal{S}}$ be defined in \Cref{lemm_l_0}, then we have that 
\begin{align}\label{theorem-affine-4}
   (I.1) = &  \frac{1}{n_{0}} \widetilde{\epsilon}^{\top} \big( P^{\widetilde{\mathcal{S}} \cup \mathcal{S}} ( \widetilde{f}^{\widetilde{A}}  - f  )  \big)
       =       \frac{1}{n_{0}} \big(P^{ \widetilde{\mathcal{S}} \cup \mathcal{S}} \widetilde{\epsilon} \big)^{\top} \big( \widetilde{f}^{\widetilde{A}}  - f  \big)   \nonumber\\
       \leq    &    \frac{1}{n_{0}}  \big\|P^{\widetilde{\mathcal{S}} \cup \mathcal{S}} \widetilde{\epsilon} \big\|_2
       \big\|\widetilde{f}^{\widetilde{A}}  - f  \big\|_2
         \leq          \frac{1}{n_{0}}  \big\|P^{\widetilde{\mathcal{S}} \cup \mathcal{S}}\widetilde{\epsilon} \big\|_2^2+  \frac{1}{4n_{0}} 
       \big\| \widetilde{f}^{\widetilde{A}}  - f    \big\|_2^2, 
\end{align}
where the first inequality follows from Cauchy--Schwartz inequality and the last inequality is based on the fact that $\vert ab \vert \leq a^2 + b^2/4$.
By  \eqref{theorem-affine-4} and \Cref{lemma_affine}, we can conclude that  that $\mathbb{P}\{\mathcal{E}\} \geq 1- n_{0}^{-c_{\epsilon}}$ with 
\[
 \quad \mathcal{E} = \bigg\{  (I.1) \leq \frac{1}{4n_{0}} 
       \big\| \widetilde{f}^{\widetilde{A}}  - f    \big\|_2^2+ C_{\epsilon}  \frac{  \big( \vert \widetilde{\mathcal{S}} \cup \mathcal{S}\vert +1 \big) \big\{1 + \log\big( n_{0}/ (\vert\widetilde{\mathcal{S}} \cup \mathcal{S}\vert+1) \big) \big\} }{n_0 / \|\widetilde{A}\|^2 }\bigg\},
\]
where $C_{\epsilon}, c_{\epsilon} >0$ are absolute constants. From now on we assume that the event $\mathcal{E}$ holds. Then it holds that 
\begin{align}\label{theorem-affine-5}
   (I)
 \leq & \frac{1}{4n_{0}} 
       \big\| \widetilde{f}^{\widetilde{A}}  - f    \big\|_2^2 +  C_{\epsilon}  \frac{  \big( \vert \widetilde{\mathcal{S}} \cup \mathcal{S}\vert +1 \big) \big\{1 + \log\big( n_{0}/ (\vert\widetilde{\mathcal{S}} \cup \mathcal{S}\vert+1) \big) \big\} }{n_0 / \|\widetilde{A}\|^2 }+  \widetilde{\lambda}_{\widetilde{A}} \|D   f \|_0 - \widetilde{\lambda}_{\widetilde{A}}   \|D   \widetilde{f}^{\widetilde{A}} \|_0\nonumber\\
= &  \frac{1}{4n_{0}} 
       \big\| \widetilde{f}^{\widetilde{A}}  - f    \big\|_2^2 +  C_{\epsilon}  \frac{  \big( \vert \widetilde{\mathcal{S}} \cup \mathcal{S}\vert +1 \big) \big\{1 + \log\big( n_{0}/ (\vert\widetilde{\mathcal{S}} \cup \mathcal{S}\vert+1) \big) \big\} }{n_0 / \|\widetilde{A}\|^2 }+  \widetilde{\lambda}_{\widetilde{A}} \big(s_0 - \vert \widetilde{\mathcal{S}} \vert \big) \nonumber\\
\leq & \frac{1}{4n_{0}} 
       \big\| \widetilde{f}^{\widetilde{A}}  - f    \big\|_2^2  + 2 \widetilde{\lambda}_{\widetilde{A}}   (s_0+1) \nonumber\\
= &\frac{1}{4n_{0}} 
       \big\| \widetilde{f}^{\widetilde{A}}  - f    \big\|_2^2  + 2 C_{\widetilde{\lambda}}\frac{( s_0+1) \big\{ 1+ \log\big(n_{0}/(s_0+1)\big) \big\} }{n_0 / \|\widetilde{A}\|^2 },
\end{align}
where 
\begin{itemize}
\item the first equality follows from the definitions of $\mathcal{S}$ and $\widetilde{\mathcal{S}}$ in \eqref{def-S} and \eqref{theorem-affine-3},
\item and the second inequality and the last equality are due to the choice of $\widetilde{\lambda}_{\widetilde{A}}$ in \eqref{tuning-parameter-affine-0} and $C_{\widetilde{\lambda}} >0$ is a large enough absolute constant.
\end{itemize}

\medskip
\noindent\textbf{Step 3.} In this step, we consider the term $(II)$ in  \eqref{theorem-affine_1}. Note that by applying the Cauchy-Schwartz inequality and utilising the fact that $\vert ab \vert \leq 2a^2 + b^2/8$, we can establish  that 
 \begin{align}\label{theorem-affine-6}
   (II) \leq     \frac{2\| \widetilde{A} \delta^A\|_2^2 }{n_{0}}  + \frac{1}{8n_{0}} \big\| \widetilde{f}^{\widetilde{A}}  - f   \big\|_2^2  \leq       \frac{ 2\| \widetilde{A} \|^2 \| \delta^A\|_2^2}{ n_0}  + \frac{1}{8n_{0}} \big\| \widehat{f} - f   \big\|_2^2.
\end{align}

\medskip
\noindent\textbf{Step 4.} Choosing  $\widetilde{\lambda}_{\widetilde{A}}$ as \eqref{tuning-parameter-affine-0}, and combining \eqref{theorem-affine_1}, \eqref{theorem-affine-5} and \eqref{theorem-affine-6}, we have with an absolute $C_1 > 0$ that  
\[
 \P  \Bigg\{  \big\|  \widehat{f} - f  \big\|_{1/n_0}^2
    \leq    C_1 \frac{  (s_{0}+1) \big\{1+\log \big(n_0/(s_0+1) \big) \big\} + \|\delta^A\|_2^2  }{n_0 / \|\widetilde{A}\|^2 }  \Bigg\} \geq 1 - n_0^{-c_{\mathcal{E}}},
\]
completing the proof.
\end{proof}

\subsubsection{Additional lemmas}\label{app-affine-lemmas}
\begin{lemma}\label{lemma-fobrnious}
 For any $n, m, k \in N^{*}$,  let $A \in \R^{n \times m}$ and $B \in R^{m \times k}$, then it holds that 
 \[
   \|AB\|_{\mathrm{F}} \leq \|A\| \|B\|_{\mathrm{F}}.
 \]
\end{lemma}
\begin{proof}
Let $\{B^{j}\}_{j=1}^{k}$ be the columns of $B$, then we have that 
\[
  \|AB\|_{\mathrm{F}}^2 =  \sum_{j=1}^k \| AB^{j} \|^2 \leq \sum_{j=1}^k \| A \|^2\|B^j \|^2 = \| A \|^2 \sum_{j=1}^k \|B^j \|^2 = \| A \|^2 \|B\|_{\mathrm{F}}^2,
\]
completing the proof.
\end{proof}

\begin{lemma}\label{lemma_affine}
For any $\mathcal{M} \subseteq [n_0-1]$, if $\vert \mathcal{M} \vert > 0$, denote it as $\mathcal{M} = \{ t_1^\mathcal{M}, \dots, t_{\vert \mathcal{M} \vert}^{\mathcal{M}} \}$. Let $t_0^{\mathcal{M}} =0$ and $t_{\vert \mathcal{M} \vert+1}^{\mathcal{M}} = n_0$. Let the subspace $\mathcal{K}^{\mathcal{M}} \subset \mathbb{R}^{n}$ be defined as $\theta \in \mathcal{K}^{\mathcal{M}}$ if and only if $\theta$ takes a constant value on $\{t_{i}^{\mathcal{M}}+1, \dots, t^{\mathcal{M}}_{i+1}\}$ for each $i [ 0 : \vert \mathcal{M} \vert]$.  Then let  $P^{\mathcal{M}}$ be the orthogonal projection operator from $\R^{n_0}$ to $\mathcal{K}^{\mathcal{M}}$. Let $\widetilde{A}\in \R^{n_0 \times n_1}$ and
 assume that $\{ \epsilon^{(1)}_{i} \}_{i=1}^{n_1}$ are mutually independent mean-zero $C_{\sigma}$-sub-Gaussian variables, $C_{\sigma}>0$ is an absolute constant. Then there exist absolute constants $C_{\epsilon}, c_{\epsilon} >0$ such that 
\begin{align}
\P \Big\{ \forall \mathcal{M}  \subseteq [n_0-1]\colon \big\| P^{\mathcal{M}} \widetilde{A}  \epsilon^{(1)} \big\|_2^2 \leq C_{\epsilon} \| \widetilde{A} \|^2 ( \vert \mathcal{M} \vert  +1 ) \big\{ 1 \vee \log  \big( n_0/ ( \vert \mathcal{M} \vert +1) \big) \big\}\Big\}  \geq 1-n_0^{-c_{\epsilon}}.  \nonumber
\end{align}

\end{lemma}

\begin{proof}
Fix $\mathcal{M} \subseteq [n_0-1]$. Since $P^{\mathcal{M}}$  is an orthogonal projection operator, it holds that $\|P^{\mathcal{M}} \| =1$. Then we have that
\begin{align}\label{add-affine-lemma-1}
      \big\| \widetilde{A}^{\top}   \big(P^{\mathcal{M}} \big)^{\top} P^{\mathcal{M}} \widetilde{A}   \big\| =   \big\|  P^{\mathcal{M}} \widetilde{A}   \big\|^2
    \leq  \big\| P^{\mathcal{M}}\big\|^2 \|\widetilde{A} \|^2  =   \|\widetilde{A} \|^2, 
\end{align}
and
\begin{align}\label{add-affine-lemma-2}
    \big\| \widetilde{A}^{\top}   \big(P^{\mathcal{M}} \big)^{\top} P^{\mathcal{M}} \widetilde{A}   \big\|_{\mathrm{F}} =   \big\|  P^{\mathcal{M}} \widetilde{A}   \big\|_{\mathrm{F}}^2
    \leq  \big\| P^{\mathcal{M}}\big\|^2 \|\widetilde{A} \|_{\mathrm{F}}^2  =   \|\widetilde{A} \|_{\mathrm{F}}^2 \leq    \|\widetilde{A} \|   \|\widetilde{A} \|_{\mathrm{F}},
\end{align}
where the first and final inequalities follow from \Cref{lemma-fobrnious}.

Note that 
\begin{align}
  \mathbb{E} \Big\{ \big\| P^{\mathcal{M}} \widetilde{A}  \epsilon^{(1)} \big\|_2^2 \Big\} 
  = & \mathbb{E} \Big\{ \big(\epsilon^{(1)} \big)^{\top} \widetilde{A}^{\top}   \big(P^{\mathcal{M}} \big)^{\top}
 P^{\mathcal{M}} \widetilde{A}  \epsilon^{(1)}  \Big\}
 = \mathbb{E} \Big[  \mathrm{tr}\Big\{ \big(\epsilon^{(1)} \big)^{\top} \widetilde{A}^{\top}   \big(P^{\mathcal{M}} \big)^{\top}
 P^{\mathcal{M}} \widetilde{A}  \epsilon^{(1)}  \Big\}\Big]
 \nonumber \\
  = & \mathbb{E} \Big[  \mathrm{tr}\Big\{ \widetilde{A}^{\top}   \big(P^{\mathcal{M}} \big)^{\top}
 P^{\mathcal{M}} \widetilde{A}  \epsilon^{(1)} \big(\epsilon^{(1)} \big)^{\top}   \Big\}\Big]
 =  \mathrm{tr}\Big[ \widetilde{A}^{\top}   \big(P^{\mathcal{M}} \big)^{\top}
 P^{\mathcal{M}} \widetilde{A}   \mathbb{E}\big\{ \epsilon^{(1)} \big(\epsilon^{(1)} \big)^{\top}  \big\} \Big ]
\nonumber \\
 \leq & C_1 \mathrm{tr}\Big\{ \widetilde{A}^{\top}   \big(P^{\mathcal{M}} \big)^{\top}
 P^{\mathcal{M}} \widetilde{A}   \Big\} 
 = C_1 \big\| P^{\mathcal{M}} \widetilde{A}   \big\|_{\mathrm{F}}^2,  \nonumber
\end{align}
where $C_1 > 0$ is an absolute constant, the first inequality follows from  $\{ \epsilon^{(1)}_{i} \}_{i=1}^{n_1}$ are mutually independent mean-zero $C_{\sigma}$-sub-Gaussian variables with an absolute constant $C_{\sigma}>0$. Then by \Cref{lemma-fobrnious}, $\|P^{\mathcal{M}}\| = 1$ and $\|P^{\mathcal{M}}\|_{\mathrm{F}}^2= \vert \mathcal{M} \vert +1  $, we have that 
\begin{align}\label{add-affine-lemma-3}
\mathbb{E} \big\{ \big\| P^{\mathcal{M}} \widetilde{A}  \epsilon^{(1)} \big\|_2^2 \big\} \leq C_1 \| \widetilde{A}   \|_{\mathrm{F}}^2 \quad \mbox{and} \quad \mathbb{E} \big\{ \big\| P^{\mathcal{M}} \widetilde{A}  \epsilon^{(1)} \big\|_2^2 \big\} \leq C_1  (\vert \mathcal{M} \vert +1)  \| \widetilde{A} \|^2. 
\end{align}

Combining Hanson--Wright inequality \citep[e.g.~Theorem 6.2.1 in][]{vershynin2018high}, \eqref{add-affine-lemma-1} and  \eqref{add-affine-lemma-2}, we have for any $u>0$  that
\begin{align}
      & \P \Big[ \Big\vert \big\| P^{\mathcal{M}}  \widetilde{A} \epsilon^{(1)} \big\|_2^2 -    \mathbb{E} \big\{ \big\| P^{\mathcal{M}} \widetilde{A}  \epsilon^{(1)} \big\|_2^2 \big\}\Big \vert \geq  u \Big]
       \leq  2\exp \bigg\{ -c_1 \min \bigg(\frac{u^2}{ \|\widetilde{A} \|^2   \|\widetilde{A} \|_{\mathrm{F}}^2}, \frac{u}{ \|\widetilde{A} \|^2 } \bigg) \bigg\}, \nonumber
\end{align}
where $c_1 >0$ is an absolute constant.
Let $u = v   \|\widetilde{A}\|_{\mathrm{F}}^2$ and obtain that
\begin{align}\label{add-affine-lemma-4}
      & \P \Big[ \Big\vert \big\| P^{\mathcal{M}}  \widetilde{A} \epsilon^{(1)} \big\|_2^2 -    \mathbb{E} \big\{ \big\| P^{\mathcal{M}} \widetilde{A}  \epsilon^{(1)} \big\|_2^2 \big\}\Big \vert \geq  v   \|\widetilde{A}\|_{\mathrm{F}}^2 \Big]
       \leq  2\exp \big\{ -c_1 \min \big( v^2, v \big)  \|\widetilde{A} \|_{\mathrm{F}}^2 /  \|\widetilde{A} \|^2   \big\}. 
\end{align}
Denote $w^2 = \min(v^2, v)$, which is equivalent to $v = \max(w, w^2)$, and
\[
 Z^2 =  \frac{\big\| P^{\mathcal{M}}  \widetilde{A} \epsilon^{(1)} \big\|_2^2  }{   \mathbb{E} \big\{ \big\| P^{\mathcal{M}} \widetilde{A}  \epsilon^{(1)} \big\|_2^2 \big\} }.
\]
Combining \eqref{add-affine-lemma-3} and \eqref{add-affine-lemma-4}, we have with an absolute constant $c_2 >0$ that 
\begin{align}\label{add-affine-lemma-5}
      \P \big\{ \vert Z^2 - 1 \vert \geq \max(w, w^2)\big\}
      \leq  & \P \Big[ \Big\vert \big\| P^{\mathcal{M}}  \widetilde{A} \epsilon^{(1)} \big\|_2^2 -    \mathbb{E} \big\{ \big\| P^{\mathcal{M}} \widetilde{A}  \epsilon^{(1)} \big\|_2^2 \big\}\Big \vert \geq C_1^{-1} \max(w, w^2)  \|\widetilde{A}\|_{\mathrm{F}}^2 \Big] \nonumber \\
       \leq & 2\exp \big\{ -c_2 w^2  \|\widetilde{A} \|_{\mathrm{F}}^2 /  \|\widetilde{A} \|^2   \big\}.
\end{align}
By Equation (3.2) in \cite{vershynin2018high},  it holds that 
\begin{align}\label{add-affine-lemma-6}
  \vert z-1 \vert \geq w \quad \mbox{implies} \quad \vert z^2-1 \vert \geq  \max(w, w^2).
\end{align}
Note that   
\begin{align}
&\P \Big[  \big\| P^{\mathcal{M}}  \widetilde{A} \epsilon^{(1)} \big\|_2    \geq  C_1 \Big\{  ( \vert \mathcal{M} \vert +1)^{1/2}    \|\widetilde{A} \| + w \|\widetilde{A} \|_{\mathrm{F}} \Big\} \Big] \nonumber\\
\leq & \P \Big[  \big\| P^{\mathcal{M}}  \widetilde{A} \epsilon^{(1)} \big\|_2  \geq   \Big\{ \mathbb{E} \big\{ \big\| P^{\mathcal{M}} \widetilde{A}  \epsilon^{(1)} \big\|_2^2 \big\} \Big\}^{1/2}  + w   \Big\{ \mathbb{E} \big\{ \big\| P^{\mathcal{M}} \widetilde{A}  \epsilon^{(1)} \big\|_2^2 \big\} \Big\}^{1/2} \Big]  \nonumber\\
\leq &  \P \big\{ \vert Z - 1 \vert \geq w \big\}  
\leq  \P \big\{ \vert Z^2 - 1 \vert \geq \max(w, w^2) \big\} 
\leq  2\exp \big\{ -c_1 w^2  \|\widetilde{A} \|_{\mathrm{F}}^2 /  \|\widetilde{A} \|^2   \big\}, \nonumber
\end{align}
where the first inequality follows from \eqref{add-affine-lemma-3}, the the third inequality follows from \eqref{add-affine-lemma-6}, and the final inequality follows from \eqref{add-affine-lemma-5} .

Let 
\[ 
w^2 =  \frac{ C_2 \|\widetilde{A} \|^2 ( \vert \mathcal{M} \vert +1 ) \big\{ 1 \vee \log \big( n_0/ ( \vert \mathcal{M} \vert +1) \big) \big\}}{   \|\widetilde{A} \|_{\mathrm{F}}^2 }.
\]
with an absolute constant $C_2 >0$, it holds that 
\begin{align}
 &  \P \Big[  \big\| P^{\mathcal{M}}  \widetilde{A} \epsilon^{(1)} \big\|_2^2    \geq 2( C_2 +1) \|\widetilde{A} \|^2 ( \vert \mathcal{M} \vert +1 ) \big\{ 1 \vee \log \big( n_0/ ( \vert \mathcal{M} \vert +1) \big\} \Big] \nonumber\\
  \leq  &  \exp \big[ -c_2 ( \vert \mathcal{M} \vert +1) \big\{1\vee \log \big( n_0/ (\vert \mathcal{M} \vert +1)\big) \big\} \big]. \nonumber
\end{align}
where $c_2 >0$ is an absolute constant.
By a union bound argument, we derive that 
\begin{align}\label{add-affine-lemma-7}
& \P \Big[ \exists \mathcal{M}  \subseteq [n_0-1]\colon \big\| P^{\mathcal{M}}  \widetilde{A}  \epsilon^{(1)} \big\|_2^2 \geq  2( C_2 +1) \|\widetilde{A} \|^2( \vert \mathcal{M} \vert +1 ) \big\{ 1 \vee \log \big( n_0/ ( \vert \mathcal{M} \vert +1) \big) \big\} \Big] \nonumber \\
\leq &  \sum_{ \mathcal{M} \subseteq [n_0 - 1] } \P \Big[ \big\| P^{\mathcal{M}}  \widetilde{A}  \epsilon^{(1)}  \big\|_2^2 \geq2 ( C_2 +1) \|\widetilde{A} \|^2 ( \vert \mathcal{M} \vert +1 ) \big\{ 1 \vee \log \big( n_0/ (\vert \mathcal{M} \vert +1) \big) \big\} \Big] \nonumber \\
\leq & \sum_{m=0}^{n_0-1} \sum_{\substack{  \mathcal{M} \subseteq [n_0 - 1] \\ \mbox{with }\vert \mathcal{M}\vert = m} } \exp \Big[ -c_2 ( m+1 ) \big\{1\vee \log \big( n_0/ (m+1)\big)\big\} \Big]  \nonumber\\
\leq & \sum_{m = 0}^{n_0-1} \binom{n_0 - 1}{m} \exp \Big[ -c_2 ( m+1 ) \big\{1\vee \log \big( n_0/ (m+1) \big)\big\}\Big]   \nonumber  \\
\leq & n_0^{-c_2} +  \sum_{m = 1}^{n_0-1}  \exp \Big[ m  \log\big( e(n_0-1)/m\big)  -c_2 ( m+1 ) \big\{1\vee \log \big( n_0/ (m+1) \big)\big\} \Big]  \nonumber\\
\leq & n_0^{-c_2} +\sum_{m = 1}^{n_0-1}  \exp \Big[  -c_3 ( m+1 ) \big\{1\vee \log \big( n_0/ (m+1) \big)\big\} \Big],  
\end{align}
where $c_3>0$ is an absolute constant and the fourth inequality is based on the fact that for any $m_1 \in \mathbb{N}^{*}$ and $m_2 \in [m_1]$ 
\[
   \binom{m_1}{m_2} \leq \Big(\frac{em_1}{m_2}\Big)^{m_2}. 
\]
The function
\[
    m \mapsto  - c_2  (m+1) \log \big( n_0/ (m+1)\big) 
\]
is convex, so its maximum over $m \in  [n_0-1]$ is attained at either $m=1$ or $m=n_0-1$. Thus, we have with an absolute constant $c_4>0$ that
\begin{align}\label{add-affine-lemma-8}
 & \sum_{m = 1}^{n_0-1} \exp \Big[ -c_3 ( m+1 ) \big(1\vee \log ( n_0/ (m+1))\big) \Big]  \nonumber \\
\leq &  (n_0-1) \max\Big[ \exp \{ -2c_3 \log ( n_0/2 ) \},  \exp \{  -c_3 n_0  \}\Big] \leq n_0^{-c_{4}}.
\end{align}
Combining \eqref{add-affine-lemma-7} and \eqref{add-affine-lemma-8}, it holds with an absolute constant $c_5 >0$ that
\begin{align}
& \P \Big[ \forall \mathcal{M}  \subseteq [n_0-1]\colon \big\| P^{\mathcal{M}}  \widetilde{\epsilon} \big\|_2^2 \leq (C_{2} +1)\|\widetilde{A} \|^2 ( \vert \mathcal{M} \vert +1 ) \big\{ 1 \vee \log \big( n_0/ ( \vert \mathcal{M} \vert +1) \big) \big\} \Big]  \geq 1- n_0^{-c_5}, \nonumber
\end{align}
completing the proof.
\end{proof}

\subsection{Proof of Proposition \ref{theorem_l_0-1_all}}\label{app-tl-1-all}

\begin{proof}[Proof of \Cref{theorem_l_0-1_all}]
This proof consists of four steps. In \textbf{Step 1}, we decompose our target quantity into several terms. We then deal with these terms individually in \textbf{Step 2} and \textbf{Step 3}. In \textbf{Step 4}, we gather all the pieces and conclude the proof.

\medskip
\noindent\textbf{Step 1.}
    It directly follows from the definition of  $\widetilde{f}^{\{ 0, 1\}}$ that
\begin{align}
&  \frac{1}{2n_0} \big\|\widetilde{P}^{n_0, n_1+n_0} \widetilde{y} -  \widetilde{f}^{\{ 0, 1\}}  \big\|_2^2 +\widetilde{\lambda} \|D \widetilde{f}^{\{ 0, 1\}} \|_0    \leq \frac{1}{2n_0} \big\|\widetilde{P}^{n_0, n_1+n_0} \widetilde{y}   -   f  \big\|_2^2 +\widetilde{\lambda} \|D  f \|_0.  \nonumber
\end{align}
Define 
\begin{equation}\label{def-combin-f}
\widetilde{f}_i = 
\begin{cases} 
f_j & \mbox{if } i = \lceil j n_1/n_0 \rceil + j \mbox{ for some } j \in [n_0], \\
f^{(1)}_{i - \vert\{ \lceil j n_1/n_0 \rceil + j\colon  j \in [n_0] \} \cap [i] \vert} & \mbox{otherwise},
\end{cases}
\end{equation}
\begin{equation}\label{def-combin-error}
\widetilde{\epsilon}_i = 
\begin{cases} 
\epsilon_j & \mbox{if } i = \lceil j n_1/n_0 \rceil + j \mbox{ for some } j \in [n_0], \\
\epsilon^{(1)}_{i - \vert\{ \lceil j n_1/n_0 \rceil + j\colon  j \in [n_0] \} \cap [i] \vert} & \mbox{otherwise},
\end{cases}
\end{equation}
and 
\begin{equation}\label{def-tilde-delta}
 \widetilde{\delta} = \widetilde{f} - P^{n_1+n_0, n_0}f.
\end{equation}
Given that $\widetilde{y} = \widetilde{f} + \widetilde{\epsilon}$, we derive that 
\begin{align}
 &  \frac{1}{2n_{0}} \big\|\widetilde{f}^{\{ 0, 1\}}  - \widetilde{P}^{n_0, n_1+n_0} P^{n_1+n_0, n_0} f  \big\|_2^2   \nonumber\\
\leq  &    \frac{1}{2n_{0}} \big\|f - \widetilde{P}^{n_0, n_1+n_0} P^{n_1+n_0, n_0} f  \big\|_2^2  +   \frac{1}{n_0} \big(   \widetilde{P}^{n_0, n_1+n_0} \widetilde{\epsilon} \big)^{\top}\big( \widetilde{f}^{\{ 0, 1\}}  - f \big)   
\nonumber\\
& \hspace{0.5cm} + \widetilde{\lambda} \|D   f \|_0 -  \widetilde{\lambda}   \|D   \widetilde{f}^{\{ 0, 1\}} \|_0 + \frac{1}{n_0}\big(   \widetilde{P}^{n_0, n_1+n_0}  \widetilde{\delta}\big)^{\top} \big( \widetilde{f}^{\{ 0, 1\}}  - f \big), \nonumber 
\end{align}
By \Cref{lemma-P}, it holds that
\begin{align}\label{theorem2_1_all}
 \frac{1}{2n_{0}} \big\|\widetilde{f}^{\{ 0, 1\}} - f  \big\|_2^2    
\leq  &   \frac{1}{n_0}\big(   \widetilde{P}^{n_0, n_1+n_0} \widetilde{\epsilon} \big)^{\top} \big( \widetilde{f}^{\{ 0, 1\}} - f \big) 
  \nonumber \\
  \hspace{0.5cm} & + \widetilde{\lambda} \|D   f \|_0 - \widetilde{\lambda}  \|D   \widetilde{f}^{\{ 0, 1\}} \|_0
+ \frac{1}{n_0}\big(   \widetilde{P}^{n_0, n_1+n_0}  \widetilde{\delta}\big)^{\top} \big(\widetilde{f}^{\{ 0, 1\}}  -   f\big) \nonumber\\
= &  (I.1) + (I.2) + (I. 3) + (II) =(I) + (II).
\end{align}

\medskip
\noindent\textbf{Step 2.} In this step, we consider the term $(I)$ in \eqref{theorem2_1_all}.

 Note that following from that  $\{ \epsilon\}_{i=1}^{n_0} \cup \{ \epsilon^{(1)}\}_{i=1}^{n_1}$ are mutually independent, the definition of $\widetilde{\epsilon}$ in \eqref{def-combin-error} and  \Cref{lemma-alignment-error}, we obtain that 
\begin{align}\label{theorem2_3_all}
\{(\widetilde{P}^{n_0, n_1+n_0}  \widetilde{\epsilon})_i\}_{i=1}^{n_0} \overset{\mbox{ind.}}{\sim} 
  \mbox{mean-zero }C_{\sigma}\{ 2 n_0/(n_1+n_0)  \}^{1/2} \mbox{-sub-Gaussian}.
\end{align}

Let the set $\mathcal{S}$ be defined in \eqref{def-S}  with  cardinality $s_0$ and the set $\widetilde{\mathcal{S}}$ be defined as
\begin{align}\label{theorem2_4_all}
\widetilde{\mathcal{S}} = \big\{i \in [n_{0}-1]\colon \widetilde{f}^{\{ 0, 1\}}_i \neq \widetilde{f}^{\{ 0, 1\}}_{i+1}  \big\}  = \big\{i \in [n_{0}-1]\colon (D\widetilde{f}^{\{ 0, 1\}})_i \neq 0  \big\}. 
\end{align}
Let the orthogonal projection operator $P^{\widetilde{\mathcal{S}} \cup \mathcal{S}}$ be defined in \Cref{lemm_l_0}, then we have that 
\begin{align}\label{theorem2_5_all}
   (I.1) = &  \frac{1}{n_{0}} \big(   \widetilde{P}^{n_0, n_1+n_0} \widetilde{\epsilon} \big)^{\top} \big( P^{\widetilde{\mathcal{S}} \cup \mathcal{S}} ( \widetilde{f}^{\{ 0, 1\}} - f  )  \big)
       =       \frac{1}{n_{0}} \big(P^{ \widetilde{\mathcal{S}} \cup \mathcal{S}} \widetilde{P}^{n_0, n_1+n_0} \widetilde{\epsilon} \big)^{\top} \big( \widetilde{f}^{\{ 0, 1\}} - f  \big)   \nonumber\\
       \leq    &    \frac{1}{n_{0}}  \big\|P^{\widetilde{\mathcal{S}} \cup \mathcal{S}} \widetilde{P}^{n_0, n_1+n_0} \widetilde{\epsilon} \big\|_2
       \big\|\widetilde{f}^{\{ 0, 1\}} - f  \big\|_2
         \leq          \frac{1}{n_{0}}  \big\|P^{\widetilde{\mathcal{S}} \cup \mathcal{S}}\widetilde{P}^{n_0, n_1+n_0} \widetilde{\epsilon} \big\|_2^2+  \frac{1}{4n_{0}} 
       \big\| \widetilde{f}^{\{ 0, 1\}} - f    \big\|_2^2, 
\end{align}
where the first inequality follows from the Cauchy--Schwartz inequality and the last inequality is based on the fact that $\vert ab \vert \leq a^2 + b^2/4$.
By \eqref{theorem2_3_all}, \eqref{theorem2_5_all} and \Cref{lemm_l_0}, we can conclude that  that $\mathbb{P}\{\mathcal{E}\} \geq 1- n_{0}^{-c_{\epsilon}}$ with 
\[
 \quad \mathcal{E} = \bigg\{  (I.1) \leq \frac{1}{4n_{0}} 
       \big\| \widetilde{f}^{\{ 0, 1\}} - f    \big\|_2^2+ C_{\epsilon}  \frac{  \big( \vert \widetilde{\mathcal{S}} \cup \mathcal{S}\vert +1 \big) \big\{1 + \log\big( n_{0}/ (\vert\widetilde{\mathcal{S}} \cup \mathcal{S}\vert+1) \big) \big\} }{n_1+n_0 }\bigg\},
\]
where $C_{\epsilon}, c_{\epsilon} >0$ are absolute constants. From now on we assume that the event $\mathcal{E}$ holds. Then it holds that 
\begin{align}\label{theorem2_7_all}
   (I)
 \leq & \frac{1}{4n_{0}} 
       \big\| \widetilde{f}^{\{ 0, 1\}} - f    \big\|_2^2 +  C_{\epsilon}  \frac{  \big( \vert \widetilde{\mathcal{S}} \cup \mathcal{S}\vert +1 \big) \big\{1 + \log\big( n_{0}/ (\vert\widetilde{\mathcal{S}} \cup \mathcal{S}\vert+1) \big) \big\} }{n_1+n_0 }+  \widetilde{\lambda} \|D   f \|_0 - \widetilde{\lambda}   \|D   \widetilde{f}^{\{ 0, 1\}}\|_0\nonumber\\
= &  \frac{1}{4n_{0}} 
       \big\| \widetilde{f}^{\{ 0, 1\}} - f    \big\|_2^2 +  C_{\epsilon}  \frac{  \big( \vert \widetilde{\mathcal{S}} \cup \mathcal{S}\vert +1 \big) \big\{1 + \log\big( n_{0}/ (\vert\widetilde{\mathcal{S}} \cup \mathcal{S}\vert+1) \big) \big\} }{n_1 +n_0}+  \widetilde{\lambda} \big(s_0 - \vert \widetilde{\mathcal{S}} \vert \big) \nonumber\\
\leq & \frac{1}{4n_{0}} 
       \big\| \widetilde{f}^{\{ 0, 1\}} - f    \big\|_2^2  + 2 \widetilde{\lambda}   (s_0+1) \nonumber\\
= &\frac{1}{4n_{0}} 
       \big\|\widetilde{f}^{\{ 0, 1\}} - f    \big\|_2^2  + 2 C_{\widetilde{\lambda}}\frac{( s_0+1) \big\{ 1+ \log\big(n_{0}/(s_0+1)\big) \big\} }{n_1 +n_0},
\end{align}
where 
\begin{itemize}
\item the first equality follows from the definitions of $\mathcal{S}$ and $\widetilde{\mathcal{S}}$ in \eqref{def-S} and \eqref{theorem2_4_all},
\item and the second inequality and the last equality are due to the choice of $\widetilde{\lambda}$ in \eqref{tuning-parameter-0_all} and $C_{\widetilde{\lambda}} >0$ is a large enough absolute constant.
\end{itemize}

\medskip
\noindent\textbf{Step 3.} In this step, we consider the term $(II)$ in  \eqref{theorem2_1_all}. Note that by applying the Cauchy-Schwartz inequality and utilising the fact that $\vert ab \vert \leq 2a^2 + b^2/8$, we can establish  that 
 \begin{align}\label{theorem2_8_all}
   (II) \leq &     \frac{2\| \widetilde{P}^{n_0, n_1+n_0} \widetilde{\delta}\|_2^2 }{n_{0}}  + \frac{1}{8n_{0}} \big\| \widetilde{f}^{\{ 0, 1\}} - f   \big\|_2^2 
     \leq    \frac{ 4 \| \widetilde{\delta}\|_2^2}{ n_1+n_0} + \frac{1}{8n_{0}} \big\| \widetilde{f}^{\{ 0, 1\}} - f   \big\|_2^2, 
\end{align}
where the second inequality follows from \Cref{lemma-alignment-delta}.

Note that 
\begin{align}\label{theorem2_9_all}
    \| \widetilde{\delta}\|_2^2 = &  \sum_{i = 1}^{n_1 + n_0} (\widetilde{f}  - P^{n_1 + n_0, n_0} f)_i^2
    = \sum_{i = 1}^{n_1 + n_0}  \bigg( \widetilde{f}_i -  \sum_{j=1}^{n_0} f_j \mathbbm{1}_{\{\lceil (j-1) (n_1 + n_0)/ n_0 \rceil +1 \leq i  \leq 
 \lceil j (n_1+n_0)/ n_0 \rceil \}}  \bigg)^2 \nonumber\\
  = &        \sum_{i = 1}^{n_1 + n_0}  \bigg( \widetilde{f}_i -  \sum_{j=1}^{n_0} f_j \mathbbm{1}_{\{\lceil (j-1) n_1 / n_0 \rceil +j \leq i \leq    \lceil jn_1/ n_0 \rceil +j \}}  \bigg)^2 
  \nonumber\\
  = &  
  \sum_{j = 1}^{n_0} \sum_{    \lceil  (j-1) n_1 / n_0 \rceil +j \leq l \leq  \lceil  j n_1 / n_0 \rceil +j } \big( \widetilde{f}_{l} - f_j \big)^2 
  \nonumber\\
  = &  
  \sum_{j = 1}^{n_0} \bigg\{ (f_j -f_j)^2 +  \sum_{    \lceil  (j-1) n_1 / n_0 \rceil +j \leq l \leq \lceil  j n_1 / n_0 \rceil +j -1 } ( f^{(1)}_{l - (j-1)} - f_j)^2 \bigg\}
    \nonumber\\
  = &  
  \sum_{j = 1}^{n_0}   \sum_{    \lceil  (j-1) n_1 / n_0 \rceil +1  \leq l \leq  \lceil  j n_1 / n_0 \rceil } ( f^{(1)}_{l} - f_j)^2, 
\end{align}
where the first equality follows from the definition of $\widetilde{\delta}$ in \eqref{def-tilde-delta}, the second equality follows from the definition of the alignment operator  $P^{n_1 + n_0, n_0}$ in \eqref{def-P}, the fifth equality follows from the definition of $\widetilde{f}$ in \eqref{def-combin-f}. We also have that 
\begin{align}\label{theorem2_10_all}
 \|\delta \|_2^2 = & \| f^{(1)} - P^{n_1, n_0}f  \|_2^2= \sum_{i=1}^{n_1}
\bigg( f^{(1)}_{i} - \sum_{j=1}^{n_0} f_j \mathbbm{1}_{\lceil (j-1) n_1/ n_0 \rceil +1 \leq i  \leq 
 \lceil j n_1/ n_0 \rceil }  \bigg)^2 \nonumber \\
 = &  \sum_{j=1}^{n_0} \sum_{\lceil (j-1) n_1/ n_0 \rceil +1 \leq l  \leq 
 \lceil j n_1/ n_0 \rceil }   \big( f^{(1)}_l -  f_j \big)^2, 
\end{align}
where the first equality follows from the definition of $\delta$ in \eqref{delta-unisource} and the second quality follows the alignment operator  $P^{n_1, n_0}$ in \eqref{def-P}.

By \eqref{theorem2_9_all} and \eqref{theorem2_10_all}, it holds that 
\[
\| \widetilde{\delta}\|_2^2 = \|\delta \|_2^2,
\]
Consequently, by \eqref{theorem2_8_all} we have that 
 \begin{align}\label{theorem2_11_all}
   (II) \leq    \frac{ 4 \|\delta\|_2^2}{ n_1+n_0} + \frac{1}{8n_{0}} \big\| \widetilde{f}^{\{ 0, 1\}} - f   \big\|_2^2.
\end{align}

\medskip
\noindent\textbf{Step 4.} Choosing  $\widetilde{\lambda}$ as \eqref{tuning-parameter-0_all}, and combining \eqref{theorem2_1_all}, \eqref{theorem2_7_all} and \eqref{theorem2_11_all}, we have with an absolute $C_1 > 0$ that  
\[
 \P  \Bigg\{  \big\|  \widetilde{f}^{\{ 0, 1\}} - f  \big\|_{1/n_0}^2
    \leq    C_1\frac{  (s_{0}+1) \big\{1+\log \big(n_1/(s_0+1) \big) \big\}+\|\delta\|_2^2 }{n_1 +n_0} \Bigg\} \geq 1 - n_0^{-c_{\mathcal{E}}},
\]
completing the proof.

\end{proof}

\section[]{Extensions: Using target data for transfer learning in non-selective multisource scenarios}\label{sec:target-multi}

This Appendix is a multisource version of \Cref{sec:target-uni}.  Different from the unisource case in \Cref{sec-one-source}, the main results in the multisource case in \Cref{sec-multiple-source} have already utilised the target data, which are used to select informative sources.  For completeness, we add the counterpart of \Cref{sec:target-uni} here.

We first introduce two alignment operators. For any $n, h \in \mathbb{N}^{*}$ and $\{m_k\}_{k=0}^h \subset \mathbb{N}^{*}$,   let the alignment operators $P^{\{m_k\}_{k=0}^h, n} \in \R^{(\sum_{k=0}^h m_k) \times n}$ and $\widetilde{P}^{n, \{m_k\}_{k=0}^h} \in \R^{n \times (\sum_{k=0}^h m_k)}$   be defined as
 \begin{equation}\label{def-P-all}
   (P^{\{m_k\}_{k=0}^h, n})_{i, j}= 
      \mathbbm{1}_{\{ \sum_{k =0}^{h}\lceil(j-1)  m_k/ n \rceil +1 \leq i  \leq \sum_{k =0}^{h} \lceil j m_k / n \rceil \}  },  \quad (i, j) \in \bigg[\sum_{k=0}^h m_k\bigg] \times [n],
 \end{equation}
and  with $\max_{k\in [0:h]} m_k \geq n$
 \begin{equation}\label{def-P-alt-all}      
      (\widetilde{P}^{n, \{m_k\}_{k=0}^h})_{i, j}  = \frac{\mathbbm{1}_{ \{ \sum_{k =0}^{h} \lceil(i-1)m_k/n  \rceil  + 1 \leq j \leq \sum_{k =0}^{h} \lceil im_k/n \rceil  \} }}{ \sum_{k =0}^{h} \big(\lceil im_k/n \rceil - \lceil(i-1)m_k/n\rceil \big)} ,
   \quad (i, j) \in [n] \times \bigg[ \sum_{k=0}^h m_k\bigg],
 \end{equation}
respectively.

Extending \eqref{estimator_l_0-1_all} from unisource to multisource, 
we introduce the target-multisource-transferred $\ell_0$-penalised estimator 
\begin{equation}\label{estimator_l_0-K_all}
    \widetilde{f}^{[0:K]} =    \widetilde{f}^{[0:K]}(\widetilde{\lambda}) = \argmin_{\theta \in \R^{n_0}} \Big\{  \frac{1}{2n_0} \Big\|\widetilde{P}^{n_0, \{n_k\}_{k=0}^K  } \widetilde{y}^{\mathrm{all}} - \theta \Big\|_2^2   + \widetilde{\lambda} \|D   \theta \|_0  \Big\},  
\end{equation}
where $\widetilde{P}^{n_0, \{n_k\}_{k=0}^K} \in \R^{n_0 \times \sum_{k \in [0:K]} n_k}$ is defined in \eqref{def-P-alt-all},   $\widetilde{\lambda} > 0$ is a tuning parameter, $D \in \R^{(n_0 - 1) \times n_0}$ is defined in \eqref{def-D}, and $\widetilde{y}^{\mathrm{all}} \in \R^{\sum_{k \in [0:K]} n_k}$ with,
\[
\widetilde{y}^{\mathrm{all}}_i = 
\begin{cases} 
y_j & \mbox{if } i =  \widetilde{J}_{j , 1}  \mbox{ for some } j \in [n_0], \\
y^{(k)}_{i - \widetilde{J}_{j , k} }  & \mbox{if } i \in \mathcal{J}_{j, k}   \mbox{ for some } j \in [n_0] \mbox{ and } k \in [K],
\end{cases} 
\quad \mbox{for } i \in \bigg[\sum_{k \in [0:K]} n_k\bigg],
\]
\begin{equation}\label{def-J}
  \widetilde{J}_{j ,k} = \sum_{\tilde{k} \in [0:(k-1)]} \lceil j n_{\tilde{k}} / n_0 \rceil + \sum_{\tilde{k} \in [0:K] \backslash [0:(k-1)]} \lceil (j-1) n_{\tilde{k}} / n_0 \rceil \quad \mbox{and}\quad \mathcal{J}_{j, k} =  \big\{i \in \mathbb{N}^{*} \colon \widetilde{J}_{j ,k} < i 
\leq  \widetilde{J}_{j, k+1} \big\},
\end{equation}
for $j \in [n_0]$ and $k \in [K]$.

The theoretical guarantees for $ \widetilde{f}^{[0:K]}$ are derived below.

\begin{proposition}\label{theorem_l_0-K_all}
Let the target data $\{y_i\}_{i= 1}^{n_0}$ be from \eqref{model-target} and  multisource data $\{y_i^{(k)}\}_{i=1, k=1}^{n_k, K}$ be from~\eqref{model-aux}.  Assume that $\{ \epsilon_i\}_{i=1}^{n_0} \cup \{ \epsilon_i^{(k)}\}_{i=1, k=1}^{n_k, K} $ are mutually independent mean-zero $C_{\sigma}$-sub-Gaussian distributed with an absolute constant $C_{\sigma} >0$. Let $\widetilde{f}^{[0:K]}$ be defined in \eqref{estimator_l_0-K_all}, with tuning parameter 
\begin{align}\label{tuning-parameter-0-K_all} 
  \widetilde{\lambda} = C_{\widetilde{\lambda}}\frac{ 1 + \log \big(n_0/(s_0+1) \big)}{\sum_{k=0}^{K}  n_k  \mathbbm{1}_{\{ n_k \geq n_0\}} },  
 \end{align}
where $C_{\widetilde{\lambda}}> 0$ is an absolute constant.  It holds with probability at least $1 - n_0^{-c}$ that 
 \begin{align}\label{upper_bound_l_0-K_all} 
    \big\|\widetilde{f}^{[0:K]}  - f  \big\|_{1/n_0}^2  
     \leq &    C\frac{  (s_{0}+1) \big\{1+\log \big(n_0/(s_0+1) \big) \big\}+ \sum_{k=1}^{K}   \|\delta^{(k)} \|_2^2}{\sum_{k=0}^{K}  n_k  \mathbbm{1}_{\{ n_k \geq n_0\}} }.
 \end{align}
\end{proposition}

\Cref{theorem_l_0-K_all} presents the estimation error bound for the the target-multisource-transferred $\ell_0$-penalised estimator defined in \eqref{estimator_l_0-K_all}, with the proof provided in \Cref{app-tl-K-all}.  Compared to the results in \Cref{sec-multiple-source}, \Cref{theorem_l_0-K_all} relaxed the condition that $\min_{k \in [K]} n_k \geq n_0$.  To further understand the results, we assume $\min_{k \in [K]}{n_k} \geq n_0$ and regard this multisource scenario as a unisource scenario by defining
$\widetilde{y}^{\mathrm{all}} =  \widetilde{f}^{\mathrm{all}} + \widetilde{\epsilon}^{\mathrm{all}}$, where $\widetilde{f}^{\mathrm{all}}, \widetilde{\epsilon}^{\mathrm{all}} \in R^{\sum_{k\in [0:K]} n_k }$ are constructed similarly to $\widetilde{y}^{\mathrm{all}}$ as using $\{f \} \cup \{f^{(k)}\}_{k=1}^K$ and $\{\epsilon \} \cup \{\epsilon^{(k)}\}_{k=1}^K$, respectively. The discrepancy level between the unisource and the target is then
\[
\widetilde{\delta}^{\mathrm{all}} =   \widetilde{f}^{\mathrm{all}} - P^{\{n_k\}_{k=0}^h, n_0} f,
\]
with $ P^{\{n_k\}_{k=0}^h, n_0} $ defined in  \eqref{def-P-all}.
 Based on our analysis of the proof of \Cref{theorem_l_0-K_all}, it holds with probability at least $1 - n_0^{-c}$ that 
 \[
    \big\|\widetilde{f}^{[0:K]}  - f  \big\|_{1/n_0}^2  
     \leq     C\frac{  (s_{0}+1) \big\{1+\log \big(n_0/(s_0+1) \big) \big\} + \|\widetilde{\delta}^{\mathrm{all}}\|_2^2}{\sum_{k=0}^{K}  n_k}.
 \]
When $(\sum_{k=0}^{K} n_k )^{-1} \|\widetilde{\delta}^{\mathrm{all}}\|_2^2 \leq   (s_{0}+1) \big\{1+\log \big(n_0/(s_0+1) \big) \big\} /n_0 $, the rate is minimax optimal by \Cref{theorem_minimax}.

\subsection{Proof of Proposition \ref{theorem_l_0-K_all}}\label{app-tl-K-all}

The proof of \Cref{theorem_l_0-K_all} can be found in  \Cref{app-tl-k-all-proof}.  Relevant notation is provided in \Cref{app-tl-k-all-notation} and all necessary auxiliary results are in \Cref{app-tl-k-all-aux}.

\subsubsection[]{Proof of \Cref{theorem_l_0-K_all}}\label{app-tl-k-all-proof}

\begin{proof}[Proof of \Cref{theorem_l_0-K_all}]
This proof consists of four steps. In \textbf{Step 1}, we decompose our target quantity into several terms. We then deal with these terms individually in \textbf{Step 2} and \textbf{Step 3}. In \textbf{Step 4}, we gather all the pieces and conclude the proof.

\medskip
\noindent\textbf{Step 1.}
    It directly follows from the definition of  $ \widetilde{f}^{[0:K]}$ that
\begin{align}
&  \frac{1}{2n_0} \big\|\widetilde{P}^{n_0, \{n_k\}_{k=0}^h   } \widetilde{y}^{\mathrm{all}} -  \widetilde{f}^{[0:K]} \big\|_2^2 +\widetilde{\lambda} \|D \widetilde{f}^{[0:K]} \|_0    \leq \frac{1}{2n_0} \big\|\widetilde{P}^{n_0, \{n_k\}_{k=0}^h   } \widetilde{y}^{\mathrm{all}}   -   f  \big\|_2^2 +\widetilde{\lambda} \|D  f \|_0.  \nonumber
\end{align}
Define 
\begin{equation}\label{def-combin-f-all}
\widetilde{f}^{\mathrm{all}}_i = 
\begin{cases} 
f_j & \mbox{if } i =  \widetilde{J}_{j , 1}  \mbox{ for some } j \in [n_0], \\
f^{(k)}_{i - \widetilde{J}_{j , k} }  & \mbox{if } i \in \mathcal{J}_{j, k}   \mbox{ for some } j \in [n_0] \mbox{ and } k \in [K],
\end{cases}
\end{equation}
and
\begin{equation}\label{def-combin-error-all}
\widetilde{\epsilon}^{\mathrm{all}}_i = 
\begin{cases} 
\epsilon_j & \mbox{if } i =  \widetilde{J}_{j , 1}  \mbox{ for some } j \in [n_0], \\
\epsilon^{(k)}_{i - \widetilde{J}_{j , k} }  & \mbox{if } i \in \mathcal{J}_{j, k}   \mbox{ for some } j \in [n_0] \mbox{ and } k \in [K],
\end{cases}
\end{equation}
where for any $j \in [n_0]$ and $k \in [K]$, $\mathcal{J}_{j , k}$   and $\widetilde{J}_{j , k}$ are define in \eqref{def-J}.
Let $ \widetilde{\delta}^{\mathrm{all}} =   \widetilde{f}^{\mathrm{all}} - P^{\{n_k\}_{k=0}^h, n_0} f$, with the alignment operator $P^{\{n_k\}_{k=0}^h, n_0}$ define in \eqref{def-P-all}. Given that $\widetilde{y}^{\mathrm{all}} =  \widetilde{f}^{\mathrm{all}} + \widetilde{\epsilon}^{\mathrm{all}}$,  we derive that 
\begin{align}
 &  \frac{1}{2n_{0}} \big\| \widetilde{f}^{[0:K]}  - \widetilde{P}^{n_0, \{n_k\}_{k=0}^h   } P^{\{n_k\}_{k=0}^h, n_0} f  \big\|_2^2   \nonumber\\
\leq  &    \frac{1}{2n_{0}} \big\|f - \widetilde{P}^{n_0, \{n_k\}_{k=0}^h   } P^{\{n_k\}_{k=0}^h, n_0} f  \big\|_2^2  +   \frac{1}{n_0} \big(   \widetilde{P}^{n_0, \{n_k\}_{k=0}^h   } \widetilde{\epsilon}^{\mathrm{all}} \big)^{\top}\big( \widetilde{f}^{[0:K]}  - f \big)   
\nonumber\\
& \hspace{0.5cm} + \widetilde{\lambda} \|D   f \|_0 -  \widetilde{\lambda}   \|D   \widetilde{f}^{[0:K]} \|_0 + \frac{1}{n_0}\big(   \widetilde{P}^{n_0, \{n_k\}_{k=0}^h   } \widetilde{\delta}^{\mathrm{all}}\big)^{\top} \big(\widetilde{f}^{[0:K]} - f \big), \nonumber 
\end{align}
By \Cref{lemma-P-alt}, it holds that
\begin{align}\label{theorem2_1_all-K}
  \frac{1}{2n_{0}} \big\|\widetilde{f}^{[0:K]} - f  \big\|_2^2     \leq  &   \frac{1}{n_0}\big(   \widetilde{P}^{n_0, \{n_k\}_{k=0}^h   } \widetilde{\epsilon}^{\mathrm{all}} \big)^{\top} \big(\widetilde{f}^{[0:K]} - f \big) 
  + \widetilde{\lambda} \|D   f \|_0 - \widetilde{\lambda}   \|D   \widetilde{f}^{[0:K]} \|_0 \nonumber\\
& \hspace{0.5cm }+ \frac{1}{n_0}\big(   \widetilde{P}^{n_0, \{n_k\}_{k=0}^h   }  \widetilde{\delta}^{\mathrm{all}}\big)^{\top} \big(\widetilde{f}^{[0:K]}  -   f\big) \nonumber\\
= &  (I.1) + (I.2) + (I. 3) + (II) =(I) + (II).
\end{align}

\medskip
\noindent\textbf{Step 2.} In this step, we consider the term $(I)$ in \eqref{theorem2_1_all-K}.

 Note that following from that  $\{ \epsilon\}_{i=1}^{n_0} \cup \{ \epsilon^{(k)}\}_{i, k=1}^{n_k, K}$ are mutually independent, the definition of $\widetilde{\epsilon}^{\mathrm{all}} $ in \eqref{def-combin-error-all} and  \Cref{lemma-alignment-error-alt}, we obtain that 
\begin{align}\label{theorem2_3_all-K}
\{(\widetilde{P}^{n_0, \{n_k\}_{k=0}^h   }  \widetilde{\epsilon}^{\mathrm{all}} )_i\}_{i=1}^{n_0} \overset{\mbox{ind.}}{\sim} 
  \mbox{mean-zero }C_{\sigma}\bigg\{ \frac{2 n_0}{ \sum_{k=0}^{K}  n_k  \mathbbm{1}_{\{ n_k \geq n_0\}} } \bigg\}^{1/2} \mbox{-sub-Gaussian}.
\end{align}

Let the set $\mathcal{S}$ be defined in \eqref{def-S}  with  cardinality $s_0$ and the set $\widetilde{\mathcal{S}}$ be defined as
\begin{align}\label{theorem2_4_all-K}
\widetilde{\mathcal{S}} = \big\{i \in [n_{0}-1]\colon \widetilde{f}^{[0:K]}_i \neq \widetilde{f}^{[0:K]}_{i+1}  \big\}  = \big\{i \in [n_{0}-1]\colon (D\widetilde{f}^{[0:K]})_i \neq 0  \big\}. 
\end{align}
Let the orthogonal projection operator $P^{\widetilde{\mathcal{S}} \cup \mathcal{S}}$ be defined in \Cref{lemm_l_0}, then we have that 
\begin{align}\label{theorem2_5_all-K}
   (I.1) = &  \frac{1}{n_{0}} \big(   \widetilde{P}^{n_0, \{n_k\}_{k=0}^h   } \widetilde{\epsilon}^{\mathrm{all}} \big)^{\top} \big( P^{\widetilde{\mathcal{S}} \cup \mathcal{S}} ( \widetilde{f}^{[0:K]} - f  )  \big) \nonumber\\
       =  &     \frac{1}{n_{0}} \big(P^{ \widetilde{\mathcal{S}} \cup \mathcal{S}} \widetilde{P}^{n_0, \{n_k\}_{k=0}^h   } \widetilde{\epsilon}^{\mathrm{all}} \big)^{\top} \big( \widetilde{f}^{[0:K]} - f  \big)   \nonumber\\
       \leq    &    \frac{1}{n_{0}}  \big\|P^{\widetilde{\mathcal{S}} \cup \mathcal{S}} \widetilde{P}^{n_0, \{n_k\}_{k=0}^h   } \widetilde{\epsilon}^{\mathrm{all}} \big\|_2
       \big\|\widetilde{f}^{[0:K]} - f  \big\|_2 \nonumber\\
         \leq   &       \frac{1}{n_{0}}  \big\|P^{\widetilde{\mathcal{S}} \cup \mathcal{S}}\widetilde{P}^{n_0, \{n_k\}_{k=0}^h   } \widetilde{\epsilon}^{\mathrm{all}} \big\|_2^2+  \frac{1}{4n_{0}} 
       \big\| \widetilde{f}^{[0:K]} - f    \big\|_2^2, 
\end{align}
where the first inequality follows from the Cauchy--Schwartz inequality and the last inequality is based on the fact that $\vert ab \vert \leq a^2 + b^2/4$.
By \eqref{theorem2_3_all-K}, \eqref{theorem2_5_all-K} and \Cref{lemm_l_0}, we can conclude that  that $\mathbb{P}\{\mathcal{E}\} \geq 1- n_{0}^{-c_{\epsilon}}$ with 
\[
 \quad \mathcal{E} = \bigg\{  (I.1) \leq \frac{1}{4n_{0}} 
       \big\|  \widetilde{f}^{[0:K]} - f    \big\|_2^2+ C_{\epsilon}  \frac{  \big( \vert \widetilde{\mathcal{S}} \cup \mathcal{S}\vert +1 \big) \big\{1 + \log\big( n_{0}/ (\vert\widetilde{\mathcal{S}} \cup \mathcal{S}\vert+1) \big) \big\} }{\sum_{k=0}^{K}  n_k  \mathbbm{1}_{\{ n_k \geq n_0\}}}\bigg\},
\]
where $C_{\epsilon}, c_{\epsilon} >0$ are absolute constants. From now on we assume that the event $\mathcal{E}$ holds. Then it holds that 
\begin{align}\label{theorem2_7_all-K}
   (I)
 \leq & \frac{1}{4n_{0}} 
       \big\|  \widetilde{f}^{[0:K]} - f    \big\|_2^2 +  C_{\epsilon}  \frac{  \big( \vert \widetilde{\mathcal{S}} \cup \mathcal{S}\vert +1 \big) \big\{1 + \log\big( n_{0}/ (\vert\widetilde{\mathcal{S}} \cup \mathcal{S}\vert+1) \big) \big\} }{\sum_{k=0}^{K}  n_k  \mathbbm{1}_{\{ n_k \geq n_0\}}   }+  \widetilde{\lambda} \|D   f \|_0 - \widetilde{\lambda}   \|D    \widetilde{f}^{[0:K]}\|_0\nonumber\\
= &  \frac{1}{4n_{0}} 
       \big\|  \widetilde{f}^{[0:K]} - f    \big\|_2^2 +  C_{\epsilon}  \frac{  \big( \vert \widetilde{\mathcal{S}} \cup \mathcal{S}\vert +1 \big) \big\{1 + \log\big( n_{0}/ (\vert\widetilde{\mathcal{S}} \cup \mathcal{S}\vert+1) \big) \big\} }{\sum_{k=0}^{K}  n_k  \mathbbm{1}_{\{ n_k \geq n_0\}}  }+  \widetilde{\lambda} \big(s_0 - \vert \widetilde{\mathcal{S}} \vert \big) \nonumber\\
\leq & \frac{1}{4n_{0}} 
       \big\|  \widetilde{f}^{[0:K]} - f    \big\|_2^2  + 2 \widetilde{\lambda}   (s_0+1) \nonumber\\
= &\frac{1}{4n_{0}} 
       \big\| \widetilde{f}^{[0:K]} - f    \big\|_2^2  + 2 C_{\widetilde{\lambda}}\frac{( s_0+1) \big\{ 1+ \log\big(n_{0}/(s_0+1)\big) \big\} }{\sum_{k=0}^{K}  n_k  \mathbbm{1}_{\{ n_k \geq n_0\}} },
\end{align}
where 
\begin{itemize}
\item the first equality follows from the definitions of $\mathcal{S}$ and $\widetilde{\mathcal{S}}$ in \eqref{def-S} and \eqref{theorem2_4_all-K},
\item and the second inequality and the last equality are due to the choice of $\widetilde{\lambda}$ in \eqref{tuning-parameter-0-K_all} and $C_{\widetilde{\lambda}} >0$ is a large enough absolute constant.
\end{itemize}

\medskip
\noindent\textbf{Step 3.} In this step, we consider the term $(II)$ in  \eqref{theorem2_1_all-K}. Note that by applying the Cauchy-Schwartz inequality and utilising the fact that $\vert ab \vert \leq 2a^2 + b^2/8$, we can establish  that 
 \begin{align}\label{theorem2_8_all-K}
   (II) \leq &     \frac{2\| \widetilde{P}^{n_0, \{n_k\}_{k=0}^h   } \widetilde{\delta}^{\mathrm{all}} \|_2^2 }{n_{0}}  + \frac{1}{8n_{0}} \big\|  \widetilde{f}^{[0:K]} - f   \big\|_2^2 \nonumber\\
   \leq &     \frac{4\| \widetilde{\delta}^{\mathrm{all}} \|_2^2 }{\sum_{k=0}^{K}  n_k  \mathbbm{1}_{\{ n_k \geq n_0\}}}  + \frac{1}{8n_{0}} \big\|  \widetilde{f}^{[0:K]} - f   \big\|_2^2, 
\end{align}
where the first inequality follows from \Cref{lemma-P-alt}.
Let $f^{(0)} = f$, then we have that 
\begin{align}\label{delta-all-1}
    \| \widetilde{\delta}^{\mathrm{all}} \|_2^2 = &  \| \widetilde{f}^{\mathrm{all}} - P^{\{n_k\}_{k=0}^h n_k, n_0}f \| \nonumber\\
    = & \sum_{i=1}^{n_0} \sum_{k=0}^{K} \sum_{ j  =\lceil (i-1)n_k /n_0\rceil +1}^{\lceil i n_k /n_0\rceil} \Big(f_j^{(k)} - \sum_{l=1}^{n_0} f_l \mathbbm{1}_{ \{ \sum_{k=0}^{K} \lceil (l-1) n_k/ n_0 \rceil +1 \leq j  \leq 
\sum_{k=0}^{K} \lceil l n_k/ n_0 \rceil \}} \Big)^{2}  \nonumber\\
= & \sum_{i=1}^{n_0} \sum_{k=0}^{K} \sum_{ j  =\lceil (j-1)n_k /n_0\rceil +1}^{\lceil j n_k /n_0\rceil } \big(f_j^{(k)} -f_i\big)^2 
\end{align}
where the second equality follows the definition of of $\widetilde{f}^{\mathrm{all}}$ in \eqref{def-combin-f-all} and the alignment operator $ P^{\{n_k\}_{k=0}^h n_k, n_0}$ in \eqref{def-P-all}.
Since for any $k \in [K]$,
\begin{align}\label{delta-all-2}
 \|\delta^{(k)} \|_2^2 = & \| f^{(k)} - P^{n_k, n_0}f  \|= \sum_{j=1}^{n_k}
\bigg( f^{(k)}_{j} - \sum_{l=1}^{n_0} f_i \mathbbm{1}_{ \{ \lceil (l-1) n_k/ n_0 \rceil +1 \leq j  \leq 
 \lceil l n_k/ n_0 \rceil \} }  \bigg)^2 \nonumber \\
 = &   \sum_{i=1}^{n_0} \sum_{\lceil (i-1) n_k/ n_0 \rceil +1 \leq j  \leq 
 \lceil i n_k/ n_0 \rceil }  \bigg( f^{(k)}_{j} - \sum_{l=1}^{n_0} f_l \mathbbm{1}_{\lceil (l-1) n_k/ n_0 \rceil +1 \leq j  \leq 
 \lceil l n_k/ n_0 \rceil }  \bigg)^2  \nonumber \\ 
 = & \sum_{i=1}^{n_0} \sum_{\lceil (i-1) n_1/ n_0 \rceil +1 \leq j  \leq 
 \lceil i n_1/ n_0 \rceil }   \big( f^{(k)}_j -  f_i \big)^2,  
\end{align}
where the second inequality follows from the definition of the alignment operator  $P^{n_k, n_0}$ in \eqref{def-P}. Combining \eqref{delta-all-1} and \eqref{delta-all-2}, it holds that 
\begin{align}\label{theorem2_10_all-K}
\|  \widetilde{\delta}^{\mathrm{all}} \|_2^2 
  =    \sum_{k=1}^{K}   \|\delta^{(k)} \|_2^2.
\end{align}

Combining \eqref{theorem2_8_all-K} and  \eqref{theorem2_10_all-K}, it holds that 
 \begin{align}\label{theorem2_11_all-K}
   (II) \leq  \frac{4 \sum_{k=1}^{K}   \|\delta^{(k)} \|_2^2.}{{\sum_{k=0}^{K}  n_k  \mathbbm{1}_{\{ n_k \geq n_0\}} } }     + \frac{1}{8n_{0}} \big\|  \widetilde{f}^{[0:K]} - f   \big\|_2^2.
\end{align}

\medskip
\noindent\textbf{Step 4.} Choosing  $\widetilde{\lambda}$ as \eqref{tuning-parameter-0-K_all}, and combining \eqref{theorem2_1_all-K}, \eqref{theorem2_7_all-K} and \eqref{theorem2_11_all-K}, we have with an absolute $C_1 > 0$ that  
\[
 \P  \Bigg\{  \big\|   \widetilde{f}^{[0:K]} - f  \big\|_{1/n_0}^2
    \leq    C_1\frac{  (s_{0}+1) \big\{1+\log \big(n_1/(s_0+1) \big) \big\}+  \sum_{k=1}^{K}   \|\delta^{(k)} \|_2^2 }{{\sum_{k=0}^{K}  n_k  \mathbbm{1}_{\{ n_k \geq n_0\}} }  } \Bigg\} \geq 1 - n_0^{-c_{\mathcal{E}}},
\]
completing the proof.

\end{proof}

\subsubsection{Additional notation}\label{app-tl-k-all-notation}

For any $i \in [n_0]$, let
\begin{equation}\label{def-mathcal-H}
 \mathcal{H}_i = \bigg\{j \in \mathbb{N}^{*}\colon  \sum_{k=0}^{K} \lceil(i-1)n_k/n_0 \rceil + 1 \leq  j \leq \sum_{k=0}^{K} \lceil in_k/n_0 \rceil  \bigg\},
\end{equation}
with $\widetilde{H}_i = \vert \mathcal{H}_i \vert$.

\subsubsection{Additional lemmas}\label{app-tl-k-all-aux}

\begin{lemma}\label{lemma-H}
Let $\mathcal{H}_i$  be defined in \eqref{def-mathcal-H} with $\widetilde{H}_i = \vert \mathcal{H}_i \vert$. It holds that
\[
\widetilde{H}_i \geq \frac{ \sum_{k =0}^{K} n_k \mathbbm{1}_{\{ n_k \geq n_0\}}}{2n_0}>0.
\]
\end{lemma}

\begin{proof}
For any $i \in [n_0]$, by the definition of   $\mathcal{H}_i$  in \eqref{def-mathcal-H} , it holds that 
\begin{align}\label{lemma-H-1}
\widetilde{H}_i =   \sum_{k=0}^{K} \big( \lceil in_k/n_0 \rceil -  \lceil(i-1) n_k/n_0 \rceil \big).
\end{align}
For any $k \in [0:h]$, if $n_k = n_0$, then it holds that 
\begin{align}\label{lemma-H-2}
 \lceil  in_k/n_0 \rceil - \lceil (i-1) n_k/n_0\rceil =   \lceil  i \rceil - \lceil i-1 \rceil   = 1 > n_k/(2n_0);
\end{align}
if $n_0 < n_k < 2 n_0$, then it holds that 
\begin{align}\label{lemma-H-3}
 \lceil i n_k/n_0 \rceil - \lceil (i-1) n_k/n_0\rceil \geq \lceil n_k/n_0 \rceil - 1 = 1  > n_k/ (2n_0); 
\end{align}
and if $ n_k \geq 2 n_0$, then it holds that 
\begin{align}\label{lemma-H-4}
     \lceil i n_k/n_0 \rceil - \lceil (i-1) n_k/n_0 \geq \lceil n_k/n_0 \rceil - 1  \geq (n_k - n_0)/n_0 \geq  n_k/ (2n_0).
\end{align}
Combining \eqref{lemma-H-1}, \eqref{lemma-H-2}, \eqref{lemma-H-3} and \eqref{lemma-H-4}, it holds that for any $i \in [n_0]$, 
\[
\widetilde{H}_i \geq \frac{ \sum_{k =0}^{K} n_k \mathbbm{1}_{\{ n_k \geq n_0\}}}{2n_0}>0,
\]
completing the proof. 

\end{proof}

\begin{lemma}\label{lemma-P-alt}
Let the alignment operators $P^{\{n_k\}_{k=0}^K, n_0} \in \R^{(\sum_{k=0}^K n_k) \times n_0}$ and $\widetilde{P}^{n_0, \{n_k\}_{k=0}^K} \in \R^{n_0 \times (\sum_{k=0}^K n_k)}$   be defined as
\eqref{def-P-all} and \eqref{def-P-alt-all}, respectively. We have that
\begin{equation}\label{eq-p-alt-1}
 \widetilde{P}^{n_0, \{n_k\}_{k=0}^K} P^{\{n_k\}_{k=0}^K, n_0}= I_{n_0},
\end{equation}
and
\begin{align}\label{eq-p-alt-2}
 \big\| \widetilde{P}^{n_0, \{n_k\}_{k=0}^K}  \big\|  \leq \frac{2n_0}{ \sum_{k =0}^{K} n_k \mathbbm{1}_{\{ n_k \geq n_0\}}}.  
\end{align}
\end{lemma}

\begin{proof}
Let $\widetilde{n} = \{n_k\}_{k=0}^K$. For any $i \in [n_0]$, let $\mathcal{H}_i$  be defined in \eqref{def-mathcal-H} with $\widetilde{H}_i = \vert \mathcal{H}_i \vert$. 

For any $i, j \in [n]$, we have that 
\begin{align}
     \big( \widetilde{P}^{n_0, \{n_k\}_{k=0}^K} P^{\{n_k\}_{k=0}^K, n_0}\big)_{i, j}
    =&  \sum_{l=1}^{\widetilde{n}}   (\widetilde{P}^{n_0, \{n_k\}_{k=0}^K})_{i, l} P^{\{n_k\}_{k=0}^K, n_0})_{l, j}
    =  \sum_{l=1}^{\widetilde{n}}  \frac{ \mathbbm{1}_{\{ l \in \mathcal{H}_i}\}}{\widetilde{H}_i}  \mathbbm{1}_{\{l \in \mathcal{H}_j\}}
    \nonumber \\
    = & \begin{cases}
        1 & \mbox{if } i=j,  \nonumber\\
        0 & \mbox{otherwise}, 
    \end{cases}
\end{align}
where the second equality follows from the definitions of $P^{\{n_k\}_{k=0}^K, n_0} $ and $\widetilde{P}^{n_0, \{n_k\}_{k=0}^K}$  in
\eqref{def-P-all} and \eqref{def-P-alt-all}, respectively, and the final equality follows from that $\{\mathcal{H}_i \}_{i=1}^{n_0}$ are disjoint sets.  Thus, it holds that 
\[
 \widetilde{P}^{n_0, \{n_k\}_{k=0}^K} P^{\{n_k\}_{k=0}^K, n_0}= I_{n_0},
\]
which completes the proof of \eqref{eq-p-alt-1}.

For any $i, j \in [n_0]$, we have that 
\begin{align}
     \big\{ \widetilde{P}^{n_0, \{n_k\}_{k=0}^K} \big(\widetilde{P}^{n, \{n_k\}_{k=0}^K} \big)^{\top} \big\}_{i, j}
    =&  \sum_{l=1}^{\widetilde{n}}   (\widetilde{P}^{n_0, \{n_k\}_{k=0}^K})_{i, l} \widetilde{P}^{n_0, \{n_k\}_{k=0}^K})_{j, l}
    =  \sum_{l=1}^{\widetilde{n}}  \frac{ \mathbbm{1}_{\{ l \in \mathcal{H}_i}\}}{\widetilde{H}_i}  \frac{ \mathbbm{1}_{\{ l \in \mathcal{H}_j}\}}{\widetilde{H}_j} 
    \nonumber \\
    = & \begin{cases}
        \frac{ 1}{\widetilde{H}_j}  & \mbox{if } i=j,  \nonumber\\
        0 & \mbox{otherwise},
    \end{cases}
\end{align}
where the second equality follows from the definitions of $\widetilde{P}^{n_0, \{n_k\}_{k=0}^K}$  in \eqref{def-P-alt-all}, and the final equality follows from that $\{\mathcal{H}_i \}_{i=1}^{n_0}$ are disjoint sets. Thus, it holds that 
\[
 \big\| \widetilde{P}^{n_0, \{n_k\}_{k=0}^K}  \big\| \leq \max_{j \in [n_0]} \frac{ 1}{\widetilde{H}_j} \leq \frac{2n}{ \sum_{k =0}^{K} n_k \mathbbm{1}_{\{ n_k \geq n_0\}}},
\]
where the last inequality follows from \Cref{lemma-H}. We complete the proof of \eqref{eq-p-alt-2}.
\end{proof}

\begin{lemma}\label{lemma-alignment-error-alt}
Let  $\widetilde{n} = \{n_k\}_{k=0}^K$ and the alignment operator $\widetilde{P}^{n_0, \{n_k\}_{k=0}^K} \in \R^{n_0 \times \widetilde{n}}$   be defined in \eqref{def-P-alt-all}. Assume $\{\widetilde{\epsilon}_i\}_{i=1}^{\widetilde{n}}$ are mutually independent mean-zero $C_{\sigma}$-sub-Gaussian variables. We have that $\{ (\widetilde{P}^{n_0, \{n_k\}_{k=0}^K} \widetilde{\epsilon})_i \}_{i=1}^{n_0}$ are mutually independent and for any $i \in [n_0]$
\[
(\widetilde{P}^{n_0, \{n_k\}_{k=0}^K} \widetilde{\epsilon})_i  \sim \mbox{mean-zero }C_{\sigma}\bigg\{ \frac{2 n_0}{ \sum_{k=0}^{K}  n_k  \mathbbm{1}_{\{ n_k \geq n_0\}} } \bigg\}^{1/2} \mbox{-sub-Gaussian}.
\]
\end{lemma}

\begin{proof}

Note that by the definition of the alignment operator $\widetilde{P}^{n_0, \{n_k\}_{k=0}^K}$ in \eqref{def-P-alt-all}, we have that for any $i \in [n_0]$,  
\begin{align}
 ( \widetilde{P}^{n_0, \{n_k\}_{k=0}^K} \widetilde{\epsilon} )_i 
  =&  \sum_{l=1}^{\widetilde{n}}   (\widetilde{P}^{n_0, \{n_k\}_{k=0}^K})_{i, l} \widetilde{\epsilon}_l
    =  \sum_{l=1}^{\widetilde{n}}  \frac{ \mathbbm{1}_{\{ l \in \mathcal{H}_i}\}}{\widetilde{H}_i}  \widetilde{\epsilon}_l =  \frac{ 1}{\widetilde{H}_i}   \sum_{ l \in \mathcal{H}_i} \widetilde{\epsilon}_l. \nonumber
\end{align}
Since  $\{\mathcal{H}_i  \}_{i=1}^{n_0}$ are disjoint sets,  we have that $\{ (\widetilde{P}^{n_0, \{n_k\}_{k=0}^K} \widetilde{\epsilon} )_i\}_{i=1}^{n_0}$ are mutually independent.
According to Proposition 2.6.1 in \cite{vershynin2018high}, we derive that for any $i \in [n_0]$,
\begin{align}
 (\widetilde{P}^{n_0, \{n_k\}_{k=0}^K} \widetilde{\epsilon})_i \sim
   \mbox{mean-zero } C_{\sigma} \widetilde{H}_i^{-1/2}\mbox{-sub-Gaussian}.  \nonumber
\end{align}

By \Cref{lemma-H},  we have that 
\[
\{(\widetilde{P}^{n_0, \{n_k\}_{k=0}^K} \widetilde{\epsilon})_i\}_{i=1}^{n_0}  \overset{\mbox{ind.}}{\sim} \mbox{mean-zero }C_{\sigma}\bigg\{ \frac{2 n_0}{ \sum_{k=0}^{K}  n_k  \mathbbm{1}_{\{ n_k \geq n_0\}} } \bigg\}^{1/2} \mbox{-sub-Gaussian}.
\]
which completes the proof.
\end{proof}

\section[]{Additional details in \Cref{sec-num}} \label{sec-add-num}

\subsection{Permutation-based algorithm}\label{app-permu-alg}
We propose a permutation-based algorithm, detailed in \Cref{alg_Per}, to choose the threshold levels for \Cref{alg_isd}. 

\begin{algorithm}[ht] 
\caption{Permutation-based algorithm for choosing the threshold level} 
    \begin{algorithmic}
        \INPUT{Target data $y \in \R^{n_0}$, source data $y^{(k)} \in \R^{n_k}, k \in [K]$, screening width $\widehat{t}^{k} \in [n_k],  k \in [K]$,  a fitting algorithm $\mathcal{A}(\cdot)$, number of permutations $B \in \mathbb{N}^*$ and quantile level $q \in (0, 1)$}
        \For{$k \in [K]$ }
        \State {$\widehat{\Delta}^{(k)} \leftarrow n_k^{-1/2} y^{(k)} - n_k^{-1/2}  P^{n_k, n_0} y $} \Comment{See \eqref{def-P} for $P^{n_k, n_0}$}
                 \EndFor
        \State{ $\widehat{k} \leftarrow \argmin_{k \in [K]}  
        \big\| \widehat{\Delta}^{(k)}  \big\|_2^2$,  
        $r = y^{(\widehat{k})} - \mathcal{A}\big(y^{(\widehat{k})} \big)$  }
        \For{$b \in [B]$ }
        \State {$r^{b} \leftarrow$ a random permutation of $r$,  $y^{(\widehat{k}), b} =  \mathcal{A}\big(y^{(\widehat{k})} \big) + r^b$,
        $\widehat{\Delta}^{b} \leftarrow n_{\widehat{k}}^{-1/2} y^{(\widehat{k}), b}  - n_{\widehat{k}}^{-1/2}  P^{n_k, n_0} y $} 
        \State{$\widehat{T}_b \leftarrow \Big\{ i \in [n_{\widehat{k}}]\colon \big\vert \widehat{\Delta}_{i}^b\big\vert \mbox{ is among the first } \widehat{t}_{\widehat{k}} \mbox{ largest of } \{|\widehat{\Delta}_j^{b}|\}_{j \in [n_{\widehat{k}}]}\Big\}$, $\tau^b \leftarrow \big\| \big( \widehat{\Delta}^b \big)_{\widehat{T}_b} \big\|_2^2  $}
        \EndFor
        \OUTPUT{The level $q$ quantile of the collection $\{\tau^b\}_{b \in [B]}$}
    \end{algorithmic}\label{alg_Per}
\end{algorithm}

\subsection{Additional results in Section \ref{sec-simulation}}\label{app-simu}

\noindent \textbf{Varying $n_0$.}  Consider the target dataset size $n_0 \in \{ 200, 400, 600, 800\}$, with corresponding $n_k = 2 n_0$ for all $k \in [K]$.  The simulation results are presented in \Cref{fig:large_size}, which demonstrates that as $n_0$ increases, the estimation errors of all methods decrease when the observational frequencies of sources are uniform and the ratio between the target dataset size and the source dataset size remains consistent.

\begin{figure}[t] 
        \centering
        \includegraphics[width=0.45 \textwidth]{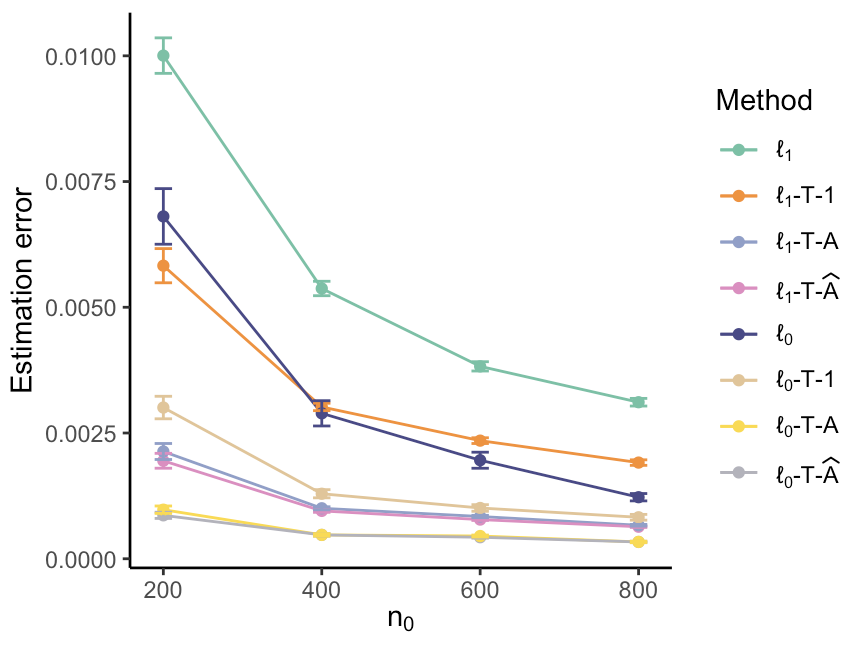}
        \caption{Estimation results for Configuration 1 and Scenario 1 with target data size $n_0 \in \{ 200, 400, 600, 800\}$ and source dataset sizes $n_k = 2 n_0$ for all $k \in [K]$.}
        \label{fig:large_size}
\end{figure}

\medskip
\noindent \textbf{Dependent data.}  We consider two types of dependence: 
\begin{itemize}
    \item Dependence across noise variables.  Let
        \begin{equation}\label{error-dep-1}
        \epsilon_{i} = \rho_1 \epsilon_{i-1} + (1-\rho_1) \widetilde{\epsilon}_{i},  \quad  \mbox{for } i \in \{2, \dots, n_0\},
        \end{equation}
        and 
        \begin{equation}\label{error-dep-2}
        \epsilon_{i}^{(k)} = \rho_1 \epsilon_{i-1}^{(k)} + (1-\rho_1) \widetilde{\epsilon}_{i}^{(k)},  \quad  \mbox{for } i \in \{2, \dots, n_0\}, \,  k \in [K],
        \end{equation}
        with $\{\epsilon_1 \} \cup \{\epsilon^{(k)}_1 \}_{k=1}^K \cup \{\widetilde{\epsilon}_i\}_{i=1}^{n_0}\cup \{\widetilde{\epsilon}^{(k)}_i\}_{i = 1, k=1}^{n_k, K} \overset{\mbox{i.i.d.}}{\sim} \mathcal{N}(0, \sigma^2)$ and $\rho_1 \in \{0.1, 0.2, 0.3, 0.4 \}$. 
    \item Dependence imposed on the discrepancy vectors.  For \textbf{Configuration 2}, 
\begin{equation}\label{dependence-delta}
\delta_{k,j} = \rho_2 \delta_{k, j-1} +(1-\rho_2) \widetilde{\delta}_{k,j}, \quad \mbox{for } j \in \{2, \ldots, n_k\}, \ k \in [K] 
\end{equation}
where
\[
\{ \delta_{k, 1}\}_{k = 1}^K  \cup\{\widetilde{\delta}_{k,j}\}_{j = 1, k=1}^{n_k, K} \stackrel{\mbox{i.i.d.}}{\sim} \mathcal{N}(0, \kappa)\mathbbm{1}\{k \in \mathcal{A}\} + \mathcal{N}(0, \widetilde{\kappa}) \mathbbm{1}\{k \notin \mathcal{A}\},
\] with $\kappa  = 0.2$, $\widetilde{\kappa} = 5$ and $\rho_2 \in \{ 0.1, 0.2, 0.3, 0.4\}$. 
\end{itemize}
The simulation results can be found in \Cref{fig:temp}.  In the left panel of \Cref{fig:temp}, we show that as the dependence level across errors increases, the estimation errors of both $\ell_1$- and $\ell_0$-penalised estimators solely using the target data, also increase. All transfer learning methods, however, remain robust against the error dependence level. In the right panel of \Cref{fig:temp}, we observe that all transfer learning methods maintain robustness against the dependence level of the discrepancy vector between the target data and sources. 

\begin{figure}[t]
        \centering
        \includegraphics[width=0.9 \textwidth]{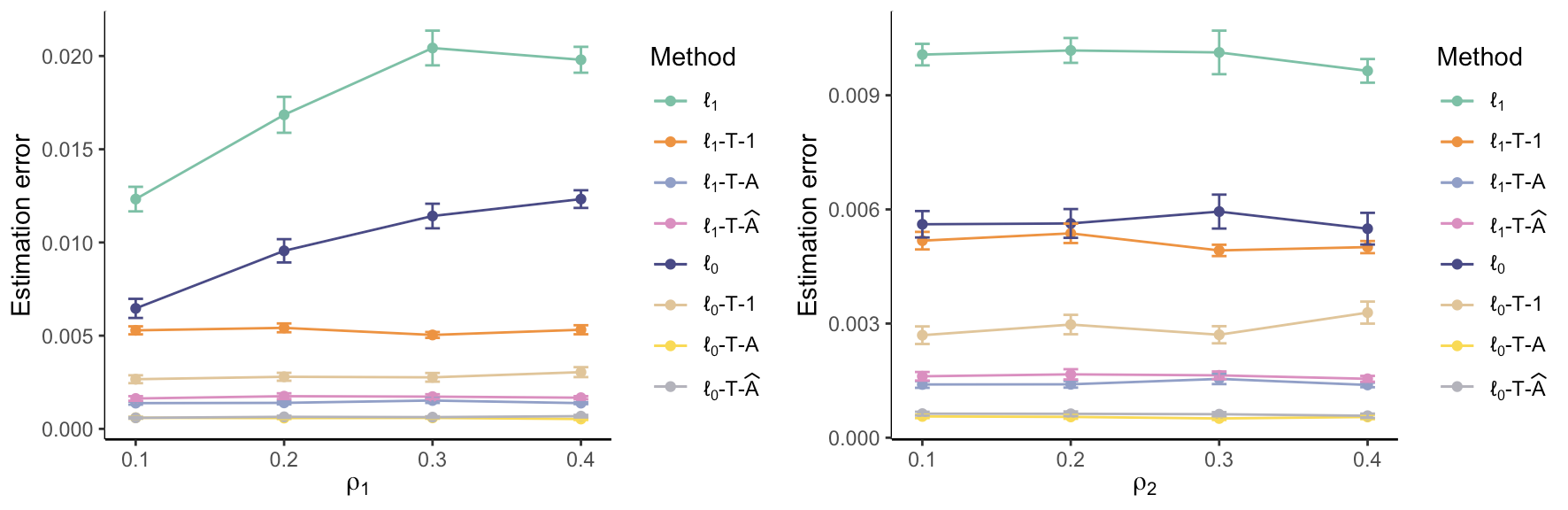}
        \caption{Estimation results for Scenario 1 with dependence. Left panel: Configuration 1 with dependence across errors defined in \eqref{error-dep-1} and \eqref{error-dep-2} and the dependence level $\rho_1 \in \{0.1, 0.2, 0.3, 0.4\}$.  Right panel: Configuration 2 with dependence across the discrepancy vector between the target and sources defined in \eqref{dependence-delta}, and the dependence level $\rho_2 \in \{0.1, 0.2, 0.3, 0.4\}$.}
        \label{fig:temp}
\end{figure}

\medskip
\noindent \textbf{Sensitivity of $\widehat{t}_k$.}  Consider a range of values $\widehat{t}_k \in \{25, 50, 75, 100\}$. The simulation results can be found in \Cref{fig:tk}, from which it is evident that the performance of all transfer learning estimators remains relatively stable across different values of screening size $\widehat{t}_k$, indicating that transferred estimators are robust to the choice of $\widehat{t}_k$. 

\begin{figure}[t]
        \centering
        \includegraphics[width=0.45 \textwidth]{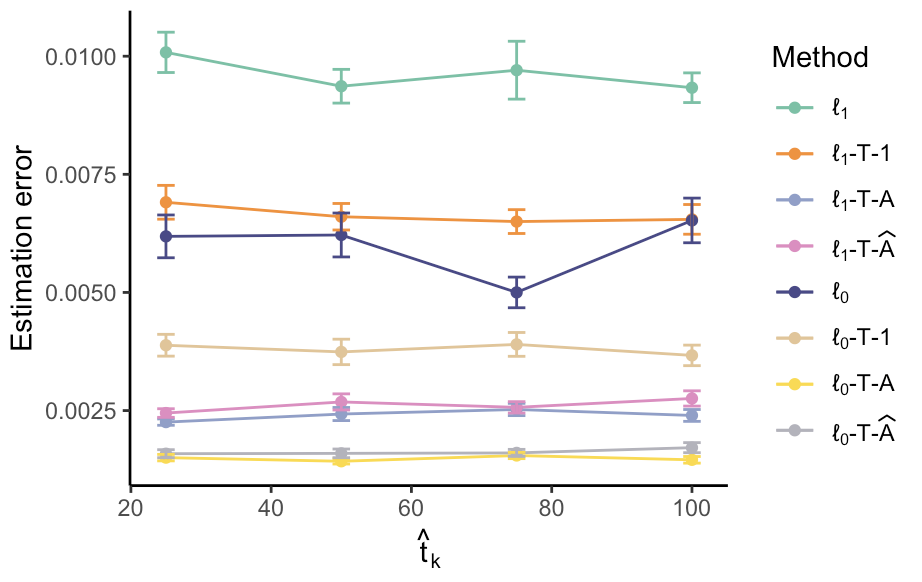}
        \caption{Estimation results for Configuration 1 and Scenario 1 with the screening size $\hat{t}_k \in \{25, 50, 75, 100 \}$.}
        \label{fig:tk}
\end{figure}

\medskip
\noindent \textbf{Scenario 2.} The simulation results of Scenarios 2 in \Cref{sec-simulation} are shown in \Cref{Fig_simulation_se_2}.

\begin{figure}[ht] 
\centering 
\includegraphics[width=0.8 \textwidth]{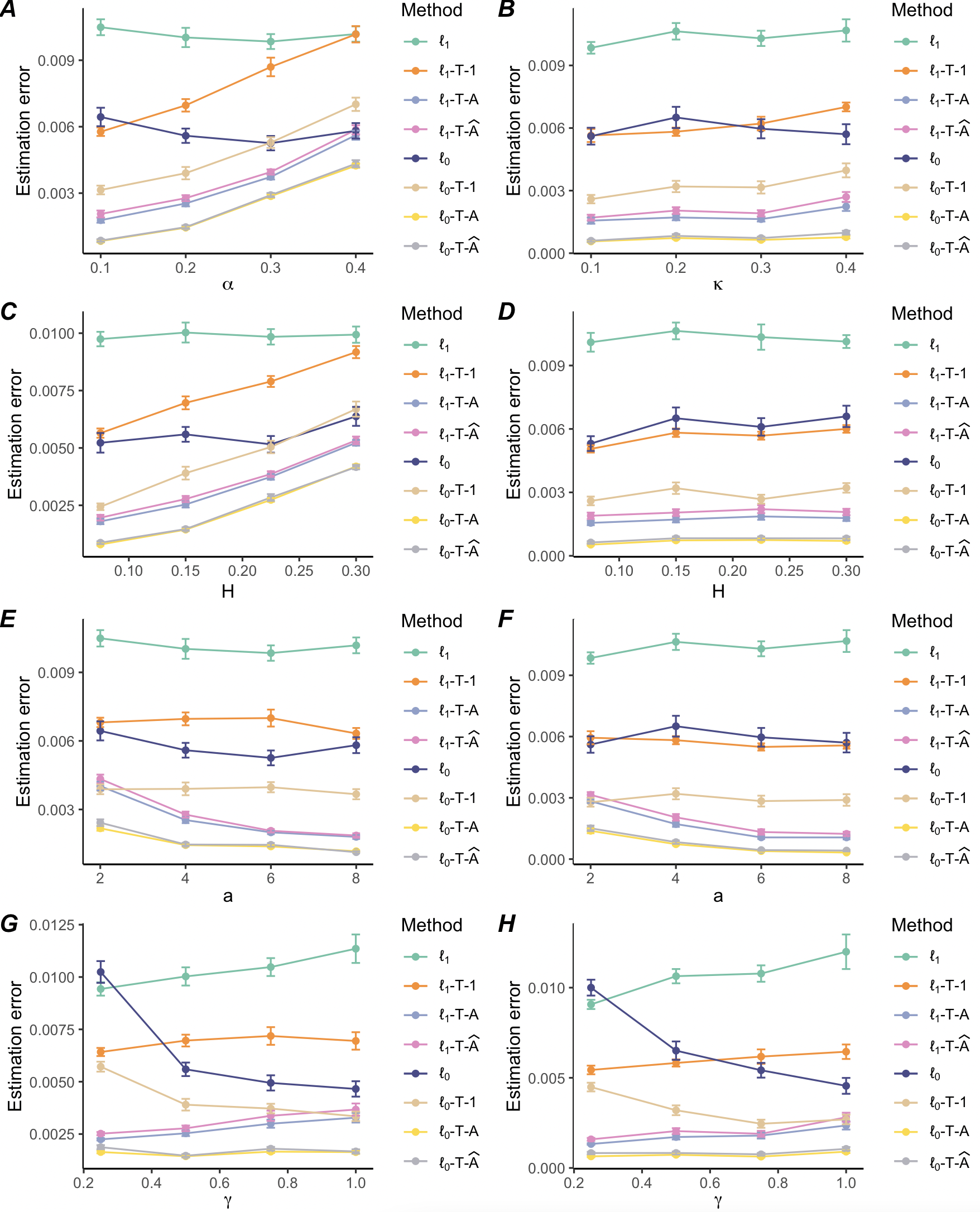}
\caption{Estimation results in Scenario 2. From left to right: Configurations 1 and 2 with dependence across the discrepancy vector between the target and source $\rho_2 = 0$.   From top to bottom: performances with varying discrepancy levels ($\alpha$ and $\kappa$), difference vector changing frequencies ($H$), cardinalities of the informative set ($a$) and change magnitudes~($\gamma$), respectively.}  \label{Fig_simulation_se_2}
\end{figure}

\subsection{Additional results in Section \ref{sec-real-data}}\label{app-real-data}

In this subsection, we provide additional results of the analyses on the U.S.~electric power operations and the air quality datasets.

In \Cref{sec-real-data}, for the U.S.~electric power operations dataset,
we only selected the Mohawk Valley sub-region as the unisource for the Central sub-region being the target, due to their geographical alignment.
We also conduct additional analyses using different sub-regions as unisource when the Central sub-region is the target.
The results can be found in \Cref{table:estimation_results-signular}.

\Cref{table:estimation_results-signular}  indicates that for both $\ell_1$-penalised and $\ell_0$-penalised estimators, the performance of using different sub-regions as unisource is generally worse than using the estimated informative sources  $\widehat{\mathcal{A}}$.
Among the different sub-regions, using the Mohawk Valley sub-region consistently outperforms other single sources, supporting our initial selection.

\begin{table}[ht]
\centering
\caption{Results for Central sub-region in the U.S. electric power operations dataset using different sub-regions as unisource. Note that the prediction errors of $\ell_1$-T-$\widehat{\mathcal{A}}$ and $\ell_0$-T-$\widehat{\mathcal{A}}$, i.e.~multisource-transferred $\ell_1$- and $\ell_0$-penalised estimator with informative sources learned by \Cref{alg_isd}, are  $0.2789$ and $0.4731$, respectively.} 
\begin{tabular}{lcc}
 \hline\hline
Unisource      & $\ell_1$-T-1     & $\ell_0$-T-1  \\ 
\hline
Mohawk Valley     & 0.3026 &   0.4298      \\
Genesee           & 0.3737 &   0.5491    \\
Capital           & 0.4075 &   0.5167    \\
Dunwoodie         & 0.7264 &   0.8262    \\
Hudson Valley     & 0.5038 &   0.6452    \\
Long Island       & 0.8662 &   0.9810    \\
Millwood          & 0.4010 &   0.5408    \\
New York City     & 0.8685 &   0.9832    \\
North             & 1.1349 &   1.2473    \\
West              & 0.4705 &   0.6438    \\
\hline
\end{tabular}
\label{table:estimation_results-signular}
\end{table}

\medskip
\noindent \textbf{The air quality dataset} collects daily air quality measurements (e.g.~$\mbox{PM}_{2.5}$, $\mbox{PM}_{10}$, $\mbox{O}_{3}$, $\mbox{NO}_{2}$ and $\mbox{CO}$) from various cities worldwide. As urban areas and industrial activities continue to grow, exploring these metrics becomes important for both policy-making and public health implications.

Two separate analyses are conducted on $\mbox{PM}_{2.5}$ data from every Saturday between 2nd July 2020 to 1st July 2023 ($156$ days), selecting Paris and London as the target datasets. For both analyses, daily $\mbox{PM}_{2.5}$ measurements from $17$ different cities (Amsterdam, Bangkok, Beijing, Chongqing, Dalian, Hamburg, Harbin, Hefei, Hong Kong, Kunming, Los Angeles, Sanya, Seoul, Shanghai, Singapore, Tianjin and Xi'an) within the same duration ($1,092$ days) serve as the multisource data.  For transfer learning estimators utilising unisource data, Paris is selected as the source for London and vice versa, due to their geographical proximity and similar urban structures.  The target dataset is split into training (even-week Saturdays),  denoted as  $y^{\mathrm{train}}$, and test datasets  (odd-week Saturdays), denoted as $y^{\mathrm{test}}$.   Using  $y^{\mathrm{train}}$, we obtain the estimated mean vector $f^{\mathrm{est}}$ and then report the mean squared prediction errors $\| f^{\mathrm{est}} - y^{\mathrm{train}} \|_{2/n_0}^2$. Results can be found in \Cref{table:estimation_results_app}.

Similar observations and conclusions as those from the previous dataset can be drawn, demonstrating the superiority of transfer learning methods especially when using informative multisources for transfer. Furthermore, the estimated informative sets from \Cref{alg_isd} present interesting city-to-city connections. For instance, data from Paris share the same patterns with those from the cities including Beijing, Hong Kong, Kunming and London, suggesting common pollution patterns.  London's air quality patterns resonate closely with those of Amsterdam, Beijing, Paris and Singapore. These interconnected trends not only highlight the similarities between these cities but also suggest collaborative strategies and interventions to address air quality issues.

\begin{table}[t]
\centering
\caption{Results for London and Paris in the air quality dataset.} 
\begin{tabular}{lccccccccc}
 \hline\hline
City & $\ell_1$ & $\ell_1$-T-1 & $\ell_1$-T-$\widehat{\mathcal{A}}$ & $\ell_1$-T-$[K]$ & $\ell_0$ & $\ell_0$-T-1 & $\ell_0$-T-$\widehat{\mathcal{A}}$ & $\ell_0$-T-$[K]$ \\
\hline
Paris & 0.9872 & 0.9043 & 0.9157 & 0.9110 & 0.9872 & 0.9872 & 0.9123 & 0.9196 \\
London & 0.9872 & 0.9179 & 0.8608 & 0.9749 & 0.9872 & 0.9785 & 0.8860 & 0.9940 \\
\hline
\end{tabular}
\label{table:estimation_results_app}
\end{table}

\end{document}